%% file: maxent_summarization_vldb2019.tex
\documentclass{vldb}
\usepackage{xspace}
\usepackage{amsmath}
\usepackage{amssymb}
\usepackage{float}
\usepackage{enumitem}
\usepackage[pdfpagelabels=false]{hyperref}
\usepackage{graphicx}
\usepackage[caption=false]{subfig}
\usepackage{listings}
\usepackage{courier}
\usepackage{soul}
\usepackage[usenames,dvipsnames]{xcolor}
\usepackage{bbm, bm}
\usepackage{url}
\usepackage{verbatim}
\usepackage{algorithm}
\usepackage{array}
\usepackage{tikz}

\newlist{myitemize}{itemize}{1}
\setlist[myitemize,1]{label=\textbullet,leftmargin=10pt,itemsep=0pt,parsep=2pt}

\input{macros}

\makeatletter
\def\BState{\State\hskip-\ALG@thistlm}
\newcommand{\vast}{\bBigg@{4}}
\newcommand{\Vast}{\bBigg@{5}}
\makeatother

\newcolumntype{M}[1]{>{\centering\arraybackslash}m{#1}}

\tikzstyle{mybox} = [draw=black, thick, rectangle]
\usetikzlibrary{positioning,decorations.pathreplacing,fit,calc}

\newcommand{\tikzmark}[2][]{%
  \tikz[remember picture,overlay,baseline=-.5ex] \node[#1] (#2) {};%
}

\newcommand*{\drawLine}[3][0pt]{%
    \node[draw=none, fit={(#2)}, inner sep=0pt] (rectg) {};%
    \draw [decorate,thick,red]%
      ([xshift=#1]rectg.east) --%
      coordinate[right=1em, midway] (#2)
      ([xshift=(#1)++(1.em)]rectg.east);%
    \node[right=-.4em of #2,red] (#2) {#3};
}%

\newcommand*{\drawLineWeird}[4][0pt]{%
    \node[draw=none, fit={(#2)}, inner sep=0pt] (rectg1) {};%
    \node[draw=none, fit={(#3)}, inner sep=0pt] (rectg2) {};%
    \draw [decorate,thick,red]%
      ([xshift=#1]rectg1.center) --%
      ([xshift=(#1)++7em,yshift=0.3em]rectg2.north);%
    \node[right=-.5em of rectg2.east,red] (#3) {#4};
}%

\newcommand*{\drawBrace}[4][0pt]{%
    \node[draw=none, fit={(#2) (#3)}, inner sep=-1pt] (rectg) {};%
    \draw [decoration={brace,amplitude=0.3em},decorate,thick,red]%
      ([xshift=#1]rectg.north east) --%
      coordinate[right=1em, midway] (#2#3)
      ([xshift=#1]rectg.south east);%
    \node[right=-0.6em of #2#3,red] (#2#3) {#4};
}%

\newcommand*{\drawBraceRotate}[4][0pt]{%
    \node[draw=none, fit={(#2) (#3)}, inner sep=-1pt] (rectg) {};%
    \draw [decoration={brace,amplitude=0.3em},decorate,thick,red]%
      ([xshift=#1]rectg.north east) --%
      coordinate[right=3.3em, midway] (#2#3)
      ([xshift=#1]rectg.south east);%
    \node[above=0em of #2#3,red,rotate=90] (#2#3) {#4};
}%

\newcommand*{\mathcolor}{}
\def\mathcolor#1#{\mathcoloraux{#1}}
\newcommand*{\mathcoloraux}[3]{%
  \protect\leavevmode
  \begingroup
    \color#1{#2}#3%
  \endgroup
}

\chardef\_=`_

\newcommand{\ie}{\textrm{i.e.}\xspace}
\newcommand{\eg}{\textrm{e.g.}\xspace}
\newcommand{\E}{\mathbb{E}}
\newcommand\inner[2]{\langle #1, #2 \rangle}

\DeclareMathOperator*{\argmin}{arg\,min}

\definecolor{mygreen}{rgb}{0,0.6,0}
\definecolor{mygray}{rgb}{0.5,0.5,0.5}
\definecolor{mymauve}{rgb}{0.58,0,0.82}

\lstdefinestyle{myJava}{ %
  language=java,
  mathescape=true,
  backgroundcolor=\color{white},   
  basicstyle=\scriptsize,        
  breaklines=true,                 
  captionpos=b,                    
  commentstyle=\color{Gray},    
  escapeinside={*@}{@*},         
  keywordstyle=\color{blue},
  otherkeywords={not},      
  stringstyle=\color{mymauve},     
}
\lstdefinestyle{mySQL}{ %
    mathescape=true,
    language=SQL,
    basicstyle=\footnotesize\ttfamily,
    deletekeywords={MIN},
    otherkeywords={LIMIT, GROUP, BY, ORDER, DESC},
    showstringspaces=false
}

\setstcolor{red}

\vldbTitle{EntropyDB: A Probabilistic Approach to Approximate Query Processing}
\vldbAuthors{Laurel Orr, Magdalena Balazinska, and Dan Suciu}
\vldbDOI{https://doi.org/10.14778/xxxxxxx.xxxxxxx}
\vldbVolume{xx}
\vldbNumber{xxx}
\vldbYear{2019}
\begin{document}

\title{EntropyDB: A Probabilistic Approach to Approximate Query Processing}
\numberofauthors{1}
\author{
\alignauthor
Laurel Orr, Magdalena Balazinska, and Dan Suciu \\
      \affaddr{University of Washington}\\
      \affaddr{Seattle, Washington, USA}\\
      \email{\{ljorr1, magda, suciu\}@cs.washington.edu}
}

\maketitle

\begin{sloppypar}
\begin{abstract}
We present \name\footnote{This is an extended version of the VLDB 2017 paper ``Probabilistic Database Summarization for Interactive Data Exploration"~\cite{orr2017probabilistic}.}, an interactive data exploration system that uses a probabilistic approach to generate a small, query-able summary of a dataset. Departing from traditional summarization techniques, we use the Principle of Maximum Entropy to generate a probabilistic representation of the data that can be used to give approximate query answers.  We develop the theoretical framework and formulation of our probabilistic representation and show how to use it to answer queries. We then present solving techniques, give two critical optimizations to improve preprocessing time and query execution time, and explore methods to reduce query error. Lastly, we experimentally evaluate our work using a 5 GB dataset of flights within the United States and a 210 GB dataset from an astronomy particle simulation. While our current work only supports linear queries, we show that our technique can successfully answer queries faster than sampling while introducing, on average, no more error than sampling and can better distinguish between rare and nonexistent values. We also discuss extensions that can allow for data updates and linear queries over joins.
\end{abstract}

\section{Introduction}
\label{sec:introduction}
\input{introduction.tex}

\section{Background}
\label{sec:background}
\input{background.tex}

\section{\name Approach}
\label{sec:probabilistic_approach}
\input{probabilistic_approach.tex}

\section{Logical Optimizations}
\label{sec:logical_optimizations}
\input{logical_optimizations.tex}

\section{System Optimizations}
\label{sec:system_optimizations}
\input{system_optimizations.tex}

\section{Statistic Selection}
\label{sec:stat_selection}
\input{statistic_selection.tex}

\section{Evaluation}
\label{sec:evaluation}
\input{evaluation.tex}

\section{Discussion}
\label{sec:discussion}
\input{discussion.tex}

\section{Related Work}
\label{sec:related_work}
\input{related_work.tex}

\section{Conclusion}
\label{sec:conclusion}
\input{conclusion.tex}

\end{sloppypar}
\small
\bibliographystyle{abbrv}
\bibliography{maxent_summarization_vldb2019}
\end{document}

%% file: macros.tex
\newcommand{\eqdef}{\stackrel{\text{def}}{=}}

\newcommand{\cut}[1]{}
\newcommand{\eat}[1]{}
\newcommand{\commentresolved}[1]{}

\newcommand{\set}[1]{\{#1\}}                    
\newcommand{\setof}[2]{\{{#1}:{#2}\}}
\usepackage{aliascnt}  		

\newtheorem{theorem}{Theorem}[section]          	
\newaliascnt{lemma}{theorem}				
\newtheorem{lemma}[lemma]{Lemma}              	
\aliascntresetthe{lemma}  					
\newaliascnt{conjecture}{theorem}			
\aliascntresetthe{conjecture}  				
\newaliascnt{remark}{theorem}				
              
\aliascntresetthe{remark}  					
\newaliascnt{corollary}{theorem}			
\newtheorem{corollary}[corollary]{Corollary}      
\aliascntresetthe{corollary}  				
\newaliascnt{definition}{theorem}			
\aliascntresetthe{definition}  				
\newaliascnt{proposition}{theorem}			
\aliascntresetthe{proposition}  				
\newaliascnt{example}{theorem}			
\newtheorem{example}[example]{Example}  	
\aliascntresetthe{example}  				

\usepackage{relsize}

\newcommand{\R}{\mathbb{R}}

\newcommand\ignore[1]{\relax}

\newcommand{\name}{EntropyDB\xspace}



%% file: introduction.tex
{\em Interactive data exploration} allows a data analyst to browse, query, transform, and visualize data at ``human speed''~\cite{crotty2016case}. It has been long recognized that general-purpose DBMSs are ill suited for interactive exploration~\cite{mozafari2015handbook}. While users require interactive responses, they do not necessarily require precise responses because either the response is used in some visualization, which has limited resolution, or an approximate result is sufficient and can be followed up with a more accurate, costly query if needed. {\em Approximate Query Processing} (AQP) refers to a set of techniques designed to allow fast but approximate answers to queries. All successful AQP systems to date rely on sampling or a combination of sampling and indices. The sample can either be computed on-the-fly, \eg, in the highly influential work on {\em online aggregation}~\cite{hellerstein1997online} or systems like DBO~\cite{jermaine2008scalable} and Quickr~\cite{kandula2016quickr}, or precomputed offline, like in BlinkDB~\cite{agarwal2013blinkdb} or Sample$+$Seek~\cite{ding2016samplelus}. Samples have the advantage that they are easy to compute, can accurately estimate aggregate values, and are good at detecting heavy hitters. However, sampling may fail to return estimates for small populations; targeted stratified samples can alleviate this shortcoming, but stratified samples need to be precomputed to target a specific query, defeating the original purpose of AQP.

In this paper, we propose an alternative approach to interactive data exploration based on the Maximum Entropy principle (MaxEnt). The MaxEnt model has been applied in many settings beyond data exploration; \eg, the {\em multiplicative weights} mechanism~\cite{hardt2010multiplicative} is a MaxEnt model for both differentially private and, by~\cite{dwork2015generalization}, statistically valid answers to queries, and it has been shown to be theoretically optimal. In our setting of the MaxEnt model, the data is preprocessed to compute a probabilistic model. Then, queries are answered by doing probabilistic inference on this model. The model is defined as the probabilistic space that obeys some observed statistics on the data and makes no other assumptions (Occam's principle). The choice of statistics boils down to a precision/memory tradeoff: the more statistics one includes, the more precise the model and the more space required. Once computed, the MaxEnt model defines a probability distribution on possible worlds, and users can interact with this model to obtain approximate query results. Unlike a sample, which may miss rare items, the MaxEnt model can infer something about every query.

Despite its theoretical appeal, the computational challenges associated with the MaxEnt model make it difficult to use in practice. In this paper, we develop the first scalable techniques to compute and use the MaxEnt model. As an application, we illustrate it with interactive data exploration. Our first contribution is to simplify the standard MaxEnt model to a form that is appropriate for data summarization (\autoref{sec:probabilistic_approach}). We show how to simplify the MaxEnt model to be a multi-linear polynomial that has one monomial for each possible tuple (\autoref{sec:probabilistic_approach},~\autoref{eq:p}) rather than its na\"{i}ve form that has one monomial for each possible world (\autoref{sec:background},~\autoref{eq:pr:i}). Even with this simplification, the MaxEnt model starts by being larger than the data. For example, our smaller experimental dataset (introduced in \autoref{sec:evaluation}) is 5 GB, but the number of possible tuples is approximately $10^{10}$, which is 74 GB if each tuple is 8 bytes. Our {\em first optimization} consists of a compression technique for the polynomial of the MaxEnt model (\autoref{subsec:log_opt:compress}); for example, for our smaller experimental dataset, the summary is below 200MB, while for our larger dataset of 210GB, it is less than 1GB. Our {\em second optimization} consists of a new technique for query evaluation on the MaxEnt model (\autoref{subsec:log_opt:aqp}) that only requires setting some variables to 0; this reduces the runtime to be on average below 500ms and always below 1 second.

As mentioned above, there is a precision/memory tradeoff when choosing which statistics to use to define the model. To alleviate this problem, our {\em third optimization} develops a statistic selection technique based on K-D trees that groups together individual statistics of similar value and uses the single group statistic in the model rather than the individual ones. We also explore optimal sorting techniques to encourage similar values to be clustered together before building our K-D trees.

We find that the main bottleneck in using the MaxEnt model is computing the model itself; in other words, computing the values of the variables of the polynomial such that it matches the existing statistics over the data. Solving the MaxEnt model is difficult; prior work for multi-dimensional histograms~\cite{markl2005consistently} uses an iterative scaling algorithm for this purpose. To date, it is well understood that the MaxEnt model can be solved by reducing it to a convex optimization problem~\cite{wainwright2008GME} of a {\em dual} function (\autoref{sec:background}), which can be solved using Gradient Descent. However, even this is difficult given the size of our model. We managed to adapt a variant of Stochastic Gradient Descent called Mirror Descent~\cite{convex-optimization-algorithms-complexity}, and our optimized query evaluation technique can compute the MaxEnt model for large datasets in under a day.

Lastly, to expand on how the MaxEnt model can be used in a full-fledged database system, we discuss handing data updates and answering join queries using the MaxEnt model. We also elaborate on the connection between the MaxEnt model and graphical models.

In summary, in this paper, we develop the following new techniques (the asterisked items are new developments from~\cite{orr2017probabilistic}):
\begin{myitemize}
\item A closed-form representation of the probability space of possible worlds using the Principle of Maximum Entropy, and a method to use the representation to answer queries in expectation (\autoref{sec:probabilistic_approach}).
\item A compression technique and optimized implementation$^*$ of the compression for the MaxEnt summary (\autoref{subsec:log_opt:compress}, \autoref{subsec:sys_opt:build_poly}).
\item Optimized query processing techniques, including implementation details$^*$ (\autoref{subsec:log_opt:aqp}, \autoref{subsec:sys_opt:poly_eval}).
\item A new method for selecting 2-dimensional statistics based on optimal matrix reorderings$^*$ and a modified K-D tree (\autoref{sec:stat_selection})
\item Detailed experiments comparing the accuracy of the MaxEnt summary versus various sampling techniques (\autoref{subsec:results:accuracy}).
\item Solving time$^*$ and query runtime evaluations showing our interactive query speeds (\autoref{subsec:results:scalability}, \autoref{subsec:results:solving_time}).
\item A discussion on how the MaxEnt summary relates to probabilistic databases and graphical models$^*$ (\autoref{sec:discussion}).
\item A description on how the MaxEnt summary can be extended to handle data updates and joins$^*$ (\autoref{subsec:discussion:joinsandupdates}).
\end{myitemize}

We implement the above techniques in a prototype system that we call \name and evaluate it on the flights and astronomy datasets. We find that \name can answer queries faster than sampling while introducing no more error, on average, and does better at identifying small populations.

%% file: background.tex
We summarize data by fitting a probability distribution over the active domain. The distribution assumes that the domain values are distributed in a way that preserves given statistics over the data but are otherwise uniform.

For example, consider a data scientist who analyzes a dataset of flights in the United States for the month of December 2013. All she knows is that the dataset includes all flights within the 50 possible states and that there are 500,000 flights in total. She wants to know how many of those flights are from CA to NY. Without any extra information, our approach would assume all flights are equally likely and estimate that there are $500,000/50^2 = 200$ flights.

Now suppose the data scientist finds out that flights leaving CA only go to NY, FL, or WA. This changes the estimate because instead of there being $500,000/50 = 10,000$ flights leaving CA and uniformly going to all 50 states, those flights are only going to 3 states. Therefore, the estimate becomes $10,000/3 = 3,333$ flights.

This example demonstrates how our summarization technique takes into account these existing statistics over flights going to and from specific states to answer queries, and the rest of this section covers its theoretical foundation.

\subsection{Possible World Semantics}
To model a probabilistic database, we use a slotted possible world semantics where rows have an inherent unique identifier, meaning the order of the tuples matters. Our set of possible worlds is generated from the active domain and size of each relation. Each database instance is one possible world with an associated probability such that the probabilities of all possible worlds sum to one.

In contrast to typical probabilistic databases where a relation is tuple-independent and the probability of a relation is calculated from the product of the probability of each tuple, we calculate a relation's probability from a formula derived from the MaxEnt principle and a set of constraints on the overall distribution\footnote{Using the MaxEnt principle will generate a probability distribution that is different from the tuple-independent distribution because the MaxEnt principle does not guarantee tuple independence.}. This approach captures the idea that the distribution should be uniform except where otherwise specified by the given constraints.

\subsection{The Principle of Maximum Entropy}
The Principle of Maximum Entropy (MaxEnt) states that subject to prior data, the probability distribution which best represents the state of knowledge is the one that has the largest entropy. This means given our set of possible worlds, $PWD$, the probability distribution $\Pr(I)$ is one that agrees with the prior information on the data and maximizes
\begin{equation*}
-\sum_{I \in PWD}\Pr(I)\log(\Pr(I))
\end{equation*}
where $I$ is a database instance, also called possible world. The above probability must be normalized, $\sum_I \Pr(I)=1$, and must satisfy the prior information represented by a set of $k$ expected value constraints:
\begin{equation}
\label{eq:e:s}
s_{j} = \E[\phi_{j}(I)], \ \ j=1,k 
\end{equation}
where $s_{j}$ is a known value and $\phi_{j}$ is a function on $I$ that returns a numerical value in $\mathbb{R}$.
One example constraint is that the number of flights from CA to WI is 0.

Following prior work on the MaxEnt principle and solving constrained optimization problems~\cite{berger1996nlpapproach,wainwright2008GME,re2012understanding}, the MaxEnt probability distribution takes the form
\begin{equation}
\label{eq:pr:i}
\Pr(I) = \frac{1}{Z}\exp\left(\sum_{j = 1}^{k}\theta_{j}\phi_{j}(I)\right) 
\end{equation}
where $\theta_{j}$ is a parameter and $Z$ is the following normalization constant:
\begin{equation*}
Z \eqdef \sum_{I \in PWD}\left(\exp\left(\sum_{j = 1}^{k}\theta_{j}\phi_{j}(I)\right)\right).
\end{equation*}

To compute the $k$ parameters $\theta_j$, we must solve the non-linear system of $k$ equations,~\autoref{eq:e:s}, which is computationally difficult.  However, it turns out~\cite{wainwright2008GME} that~\autoref{eq:e:s} is equivalent to $\partial \Psi / \partial \theta_j = 0$ where the {\em dual} $\Psi$ is defined as:
\begin{equation*}
\Psi \eqdef \sum_{j = 1}^{k} s_{j}\theta_{j} - \ln\left(Z\right).
\end{equation*}
Furthermore, $\Psi$ is concave, which means solving for the $k$ parameters can be achieved by maximizing $\Psi$. We note that $Z$ is called the {\em partition function}, and its log, $\ln(Z)$, is called the {\em cumulant}. 

Lastly, we adopt a slightly different notation where instead of $e^{\theta}$ we use $\alpha$. \autoref{eq:pr:i} now becomes
\begin{equation}
\label{eq:pr:i:new}
\Pr(I) = \frac{1}{Z}\prod_{j = 1}^k \alpha_{j}^{\phi_{j}(I)}. 
\end{equation}

%% file: probabilistic_approach.tex
This section explains how we use the MaxEnt model for approximate query answering. We first show how we use the MaxEnt framework to transform a single relation $R$ into a probability distribution represented by $P$. We then explain how we use $P$ to answer queries over $R$. For reference,~\autoref{fig:notation_tab} lists common symbols and their definitions, and~\autoref{fig:model_assumptions} lists various assumptions we incrementally make on our model and the section they are introduced.
\begin{table}
    \begin{small}
    \centering
    \begin{tabular}{|c|c|}
    \hline
    $m$ & \# attrs \\ \hline
    $k$ & \# statistics\\ \hline
    $D_i$ & domain of $A_i$ \\ \hline
    $N_i$ & $|D_i|$ \\ \hline
    $\mathbf{q}$ & linear query \\ \hline
    $\mathbf{c}_j$ & $j$th statistic query \\ \hline
    $\pi_j$ & $j$th statistic predicate \\ \hline
    $\theta_j$ & $(\mathbf{c}_j, s_j)$ \\ \hline
    $s_j$ & $j$th statistic constraint \\ \hline
    $\rho_{ij}$ & projection $\pi_j$ onto $A_i$ \\ \hline
    $J_i \subseteq [k]$ & 1-dim statistic indices \\ \hline
    $\mathcal{I} \subseteq [m]$ & attr indices \\ \hline
    $J_\mathcal{I} \subseteq [k]$ & multi-dimensional statistic indices\\ \hline
    $B_a$ & \# multi-dimensional attr sets \\ \hline
    $B_s$ & \# stats per multi-dimensional attr set \\ \hline
    \end{tabular}
    \caption{Common notation.}
    \label{fig:notation_tab}
    \end{small}
\end{table}
\subsection{Maximum Entropy Model of Data}
\label{sec:theoretical_setup}

\begin{table}
    \begin{small}
    \centering
    \begin{tabular}{|p{6cm}|p{1cm}|}
    \hline
    Queries are limited to linear queries. & \autoref{sec:theoretical_setup} \\ \hline
    Continuous attributes are discretized. & \autoref{sec:theoretical_setup} \\ \hline
    Summary includes all 1-dimensional statistics. & \autoref{sec:theoretical_setup} \\ \hline
    Statistics are collections of range predicates. & \autoref{subsec:log_opt:compress} \\ \hline
    Each set of 2D statistics is disjoint. & \autoref{subsec:log_opt:compress} \\ \hline
    Summary adds only 2D high order statistics. & \autoref{subsec:stat_select:optimal_ranges} \\ \hline
    \end{tabular}
    \caption{Model assumptions, and the section they are introduced.}
    \label{fig:model_assumptions}
    \end{small}
\end{table}

We consider a single relation with $m$ attributes and schema $R(A_1,\ldots, A_m)$ where each attribute, $A_i$, has an active domain $D_i$, assumed to be discrete and ordered.\footnote{We support continuous data types by bucketizing their active domains.}  Let $Tup = D_1 \times D_2 \times \dots \times D_m = \{t_1, \ldots, t_d\}$ be the set of all possible tuples.  Denoting  $N_i = |D_i|$, we have $d = |Tup| = \prod_{i=1}^m N_i$.

An {\em instance} for $R$ is an ordered bag of $n$ tuples, denoted $I$. For each $I$, we form a frequency vector which is a $d$-dimensional vector\footnote{This is a standard data model in several applications, such as differential privacy~\cite{li2010optimizing}.} $\mathbf{n}^I = [n^I_1,\ldots, n^I_d] \in \R^d$, where each number $n^I_i$ represents the count of the tuple $t_i \in Tup$ in $I$ (\autoref{fig:model}).  The mapping from $I$ to $\mathbf{n}^I$ is not one-to-one because the instance $I$ is ordered, and two distinct instances may have the same counts. Further, for any instance $I$ of cardinality $n$, $||\mathbf{n}^I||_1 = \sum_i n^I_i = n$. The frequency vector of an instance consisting of a single tuple $\set{t_i}$ is denoted $\mathbf{n}^{t_i} = [0,\ldots,0,1,0,\ldots,0]$ with a single value $1$ in the $i$th position; \ie, $\setof{\mathbf{n}^{t_i}}{i = 1,d}$ forms a basis for $\R^d$.

While the MaxEnt principle allows us, theoretically, to answer any query probabilistically by averaging the query over all possible instances; in this paper, we limit our main analysis to linear queries but do discuss how to handle joins in~\autoref{subsec:discussion:joinsandupdates}. A {\em linear query} is a $d$-dimensional vector $\mathbf{q} = [q_1, \ldots, q_d]$ in $\R^d$.  The answer to $\mathbf{q}$ on instance $I$ is the dot product $\inner{\mathbf{q}}{\mathbf{n}^I} = \sum_{i=1}^d q_i n^I_i$. With some abuse of notation, we will write $\mathbf{I}$ when referring to $\mathbf{n}^I$ and $\mathbf{t}_i$ when referring to $\mathbf{n}^{t_i}$.  Notice that $\inner{\mathbf{q}}{\mathbf{t}_i} = q_i$, and, for any instance $I$, $\inner{\mathbf{q}}{\mathbf{I}} = \sum_i n_i^I \inner{\mathbf{q}}{\mathbf{t}_i}$.

\autoref{fig:model} illustrates the data and query model. Any counting query is a vector $\mathbf{q}$ where all coordinates are 0 or 1 and can be equivalently defined by a predicate $\pi$ such that $\inner{\mathbf{q}}{\mathbf{I}} = |\sigma_\pi(I)|$; with more abuse, we will use $\pi$ instead of $\mathbf{q}$ when referring to a counting query.  Other SQL queries can be modeled using linear queries, too. For example,
\begin{small}
\begin{lstlisting}[style=mySQL]
SELECT A, COUNT(*)
FROM R GROUP BY A
ORDER BY COUNT(*) DESC LIMIT 10
\end{lstlisting}
\end{small}
corresponds to several linear queries, one for each group, where the outputs are sorted and the top 10 returned.

\begin{figure}
\scriptsize
\begin{tikzpicture}
\node [mybox,minimum width=8.35cm] (box1) at (0, 0){
\begin{minipage}{0.45\textwidth}
Domains:
\vspace{-4pt}
\begin{align*}
  &D_1 = \set{a_1,a_2}  &N_1 = 2\\
  &D_2 = \set{b_1,b_2}  &N_2 = 2\\
  &\mathit{Tup} = \set{(a_1,b_1),(a_1,b_2),(a_2,b_1),(a_2,b_2)}  &d = 4
\end{align*}
\end{minipage}
};
\node [mybox,minimum width=8.35cm] (box2) at (0, -1.85) {
\begin{minipage}[t]{0.22\textwidth}
Database Instance:
\vspace{2pt}
\newline
\begin{tabular}{l|l|l|l} \cline{2-3}
$I$:  & $A$ & $B$ \\ \cline{2-3}
1 & $a_1$ & $b_1$ & $\mathit{Tup}_1$ \\
2 & $a_1$ & $b_2$ & $\mathit{Tup}_2$\\
3 & $a_2$ & $b_2$ & $\mathit{Tup}_4$\\
4 & $a_1$ & $b_1$ & $\mathit{Tup}_1$\\
5 & $a_2$ & $b_2$ & $\mathit{Tup}_4$\\ \cline{2-3}
\end{tabular}
\end{minipage}
\begin{minipage}[t]{0.22\textwidth}
Query: 
\vspace{6pt}
\newline
\begin{tabular}{rl}
\texttt{q:} & \texttt{SELECT COUNT(*)} \\
            & \texttt{FROM R} \\
            & \texttt{WHERE A = a1}  
\end{tabular}
\end{minipage}
};
\node [mybox,minimum width=8.35cm] (box3) at (0, -3.4){
\begin{minipage}{0.45\textwidth}
Modeling Data and Query: n = 5,\ m = 2
\begin{align*}
& \mathbf{n}^{\mathbf{I}} = (2,1,0,2) ~~ \mathbf{q} = (1,1,0,0)  ~~
  \inner{\mathbf{q}}{\mathbf{n}^{\mathbf{I}}} = 3
  \mbox{\hspace{6pt}also denoted } \inner{\mathbf{q}}{\mathbf{I}} \mbox{\hspace{0.3cm}}
\end{align*}
\end{minipage}
};
\end{tikzpicture}
\caption{Illustration of the data and query model.}
\label{fig:model}
\end{figure}

Our goal is to compute a summary of the data that is small yet allows us to approximatively compute the answer to any linear query.  We assume that the cardinality $n$ of $R$ is fixed and known.  In addition, we know $k$ statistics, $\Phi = \setof{(\mathbf{c}_j,s_j)}{j=1,k}$, where $\mathbf{c}_j$ is a linear query and $s_j \geq 0$ is a number.  Intuitively, the statistic $(\mathbf{c}_j,s_j)$ asserts that $\inner{\mathbf{c}_j}{I} = s_j$. For example, we can write 1-dimensional and 2-dimensional (2D) statistics like $|\sigma_{A_1 = 63}(I)| = 20$ and $|\sigma_{A_1 \in [50,99]\wedge A_2 \in [1,9]}(I)| = 300$.

Next, we derive the MaxEnt distribution for the possible instances $I$ of a fixed size $n$. We replace the exponential parameters $\theta_j$ with $\ln(\alpha_j)$ so that ~\autoref{eq:pr:i:new} becomes

\begin{equation}
\label{eq:pr:n}
\Pr(I) = \frac{1}{Z}\prod_{j=1,k} \alpha_j^{\inner{\mathbf{c}_j}{\mathbf{I}}}.
\end{equation}

We prove the following about the structure of the partition function $Z$:
\begin{lemma}
The partition function is given by
\begin{equation}
\label{eq:z:p}
Z = P^n
\end{equation}
where $P$ is the multi-linear polynomial
\begin{equation}
\label{eq:p}
P(\alpha_1, \ldots, \alpha_k) \eqdef \sum_{i=1,d} \prod_{j=1,k} \alpha_j^{\inner{\mathbf{c}_j}{\mathbf{t}_i}}.
\end{equation}
\end{lemma}

\begin{proof}
Fix any $\mathbf{n} = [n_1,\ldots, n_d]$ such that $||\mathbf{n}||_1 = \sum_{i=1}^d n_i = n$.  The number of instances $I$ of cardinality $n$ with $\mathbf{I} = \mathbf{n}$ is $n!/\prod_i n_i!$. Furthermore, for each such instance, $\inner{\mathbf{c}_j}{\mathbf{I}} = \inner{\mathbf{c}_j}{\mathbf{n}} = \sum_i n_i\inner{\mathbf{c}_j}{\mathbf{t}_i}$.  Therefore,

\begin{align*}
Z = & \sum_I \Pr(I) = \sum_{\mathbf{n}: ||\mathbf{n}||_1 = n} \frac{n!}{\prod_{i} n_i!} \prod_{j=1,k} \alpha_j^{\sum_{i}n_i \inner{\mathbf{c}_j}{\mathbf{t}_i}} \\
= & \left(\sum_{i=1,d} \prod_{j=1,k} \alpha_j^{\inner{\mathbf{c}_j}{\mathbf{t}_i}}\right)^n = P^n.
\end{align*}
\end{proof}

This restructuring of the partition function is valid because we represent an instance as an ordered bag rather than an unordered one.

The {\em data summary}, denoted $(P, \{\alpha_j\}, \Phi)$, consists of the polynomial $P$ (\autoref{eq:p}), the values of its parameters $\alpha_j$, and the statistics $\Phi$. The statistics must be included as the polynomial parameters are defined by the linear queries $\mathbf{c}_j$ in the statistics $\Phi$.

\begin{example}
\label{ex:simple}
Consider a relation with three attributes $R(A, B, C)$, and assume that the domain of each attribute has 2 distinct elements. Assume $n = 10$ and the only statistics in $\Phi$ are the following 1-dimensional statistics:
\begin{equation*}
\begin{array}{lll}
(A = a_1,\ 3) & (B = b_1,\ 8) & (C = c_1,\ 6) \\
(A = a_2,\ 7) & (B = b_2,\ 2) & (C = c_2,\ 4). \\
\end{array}
\end{equation*}
The first statistic asserts that $|\sigma_{A=a_1}(I)|=3$, etc. The polynomial $P$ is
\begin{align*}
P = 
&\alpha_{1}\beta_{1}\gamma_{1} + \alpha_{1}\beta_{1}\gamma_{2} + 
\alpha_{1}\beta_{2}\gamma_{1} + \alpha_{1}\beta_{2}\gamma_{2} + \\
&\alpha_{2}\beta_{1}\gamma_{1} + \alpha_{2}\beta_{1}\gamma_{2} + 
\alpha_{2}\beta_{2}\gamma_{1} + \alpha_{2}\beta_{2}\gamma_{2}
\end{align*}
where $\alpha_1,\alpha_2$ are variables associated with the statistics on $A$, $\beta_1,\beta_2$ are for $B$\footnote{We abuse notation here for readability. Technically, $\alpha_{i} = \alpha_{a_i}$, $\beta_{i} = \alpha_{b_i}$, and $\gamma_{i} = \alpha_{c_i}$.}, and $\gamma_1,\gamma_2$ are for $C$.

Consider the concrete instance that satisfies the above statistics
\begin{align*}
& I = \{(a_1,b_2,c_2), (a_1,b_1,c_2), (a_1,b_1,c_2), (a_2,b_2,c_1) \\
& (a_2,b_1,c_1), (a_2,b_1,c_1), (a_2,b_1,c_1), (a_2,b_1,c_1), (a_2,b_1,c_1)\}.
\end{align*}
Then, $\Pr(I) = \alpha_1^3\alpha_2^7\beta_1^8\beta_2^2\gamma_1^6\gamma_2^4/P^{10}$ where $\alpha_1^3$ represents $\alpha_1$ raised to the third power, $\alpha_2^7$ represents $\alpha_2$ to the seventh power, and so on.
\end{example}

\begin{example}
\label{ex:complex1}
Continuing the previous example, we add the following multi-dimensional statistics to $\Phi$:
\begin{equation*}
\begin{array}{ll}
(A = a_1 \land B = b_1,\ 2) & (B = b_1 \land C = c_1,\ 5) \\
(A = a_2 \land B = b_2,\ 1) & (B = b_2 \land C = c_1,\ 1). \\
\end{array}
\end{equation*}
$P$ is now
\begin{align}
\label{eq:pnew}
P = 
&\alpha_{1}\beta_{1}\gamma_{1}\mathcolor{red}{[\alpha\beta]_{1,1}}\mathcolor{red}{[\beta\gamma]_{1,1}} + \alpha_{1}\beta_{1}\gamma_{2}\mathcolor{red}{[\alpha\beta]_{1,1}} + \nonumber \\
&\alpha_{1}\beta_{2}\gamma_{1}\mathcolor{red}{[\beta\gamma]_{2,1}} + \alpha_{1}\beta_{2}\gamma_{2} + \nonumber \\
&\alpha_{2}\beta_{1}\gamma_{1}\mathcolor{red}{[\beta\gamma]_{1,1}} + \alpha_{2}\beta_{1}\gamma_{2} + \nonumber \\
&\alpha_{2}\beta_{2}\gamma_{1}\mathcolor{red}{[\alpha\beta]_{2,2}}\mathcolor{red}{[\beta\gamma]_{2,1}} + \alpha_{2}\beta_{2}\gamma_{2}\mathcolor{red}{[\alpha\beta]_{2,2}}.
\end{align}
The red variables are the added 2-dimensional statistic variables; we use $[\alpha\beta]_{1,1}$ to denote a {\em single} variable corresponding to a 2-dimensional statistics on the attributes $AB$. Notice that each  red variable only occurs with its related 1-dimensional variables. $[\alpha\beta]_{1,1}$, for example, is only in the same term as $\alpha_{1}$ and $\beta_{1}$.  

Now consider the earlier instance $I$. Its probability becomes $\Pr(I) = \alpha_1^3\alpha_2^7\beta_1^8\beta_2^2\gamma_1^6\gamma_2^4[\alpha\beta]_{1,1}^2[\alpha\beta]_{2,2}^1[\beta\gamma]_{1,1}^5[\beta\gamma]_{2,1}^1/P^{10}$.
\end{example}

To facilitate analytical queries, we choose the set of statistics $\Phi$ as follows:
\begin{myitemize}
\item Each statistic $\phi_j=(\mathbf{c}_j,s_j)$ is associated with some predicate $\pi_j$ such that $\inner{\mathbf{c}_j}{\mathbf{I}} = |\sigma_{\pi_j}(I)|$. It follows that for every tuple $t_i$, $\inner{\mathbf{c}_j}{\mathbf{t}_i}$ is either 0 or 1; therefore, each variable $\alpha_j$ has degree 1 in the polynomial $P$ in~\autoref{eq:p}. 

\item For each domain $D_i$, we include a complete set of 1-dimensional statistics in our summary. In other words, for each $v \in D_i$, $\Phi$ contains one statistic with predicate $A_i = v$.  We denote $J_i \subseteq [k]$ the set of indices of the 1-dimensional statistics associated with $D_i$; therefore, $|J_i| = |D_i| = N_i$.

\item We allow multi-dimensional statistics to be given by arbitrary predicates. They may be overlapping and/or incomplete; \eg, one statistic may count the tuples satisfying $A_1  \in [10,30] \wedge A_2 = 5$ and another count the tuples satisfying $A_2 \in [20, 40] \wedge A_4=20$. 

\item We assume the number of 1-dimensional statistics dominates the number of attribute combinations; \ie, $\sum_{i=1}^m N_i \gg 2^m$.


\item If some domain $D_i$ is large, it is beneficial to reduce the
  size of the domain using equi-width buckets. In that case, we assume the elements of $D_i$ represent buckets, and $N_i$ is the number of buckets.

\item We enforce our MaxEnt distribution to be {\em overcomplete}~\cite[pp.40]{wainwright2008GME} (as opposed to {\em minimal}). More precisely, for any attribute $A_i$ and any instance $I$, we have $\sum_{j \in J_i} \inner{\mathbf{c}_j}{\mathbf{I}} = n$, which means that some statistics are redundant since they can be computed from the others and from the size of the instance $n$.
\end{myitemize}

Note that as a consequence of overcompleteness, for any attribute $A_i$, one can write $P$ as a linear expression

\begin{equation}
\label{eq:p1i}
P = \sum_{j \in J_i} \alpha_j P_j
\end{equation}
where each $P_j$, $j \in J_i$ is a polynomial that does not contain the variables $(\alpha_j)_{j \in J_i}$. In Example~\ref{ex:complex1}, the 1-dimensional variables for $A$ are $\alpha_1$, $\alpha_2$, and indeed, each monomial in~\autoref{eq:pnew} contains exactly one of these variables. One can write $P$ as $P = \alpha_1 P_1 + \alpha_2 P_2$ where $\alpha_1 P_1$ represents the first two lines and $\alpha_2 P_2$ represents the last two lines in~\autoref{eq:pnew}. $P$ is also linear in $\beta_1$, $\beta_2$ and in $\gamma_1$, $\gamma_2$.

\subsection{Query Answering}
\label{subsec:query:answer}
In this section, we show how to use the data summary to approximately answer a linear query $q$ by returning its expected value $\E[\inner{\mathbf{q}}{\mathbf{I}}]$. The summary $(P, \{\alpha_j\}, \Phi)$ defines a probability space on the possible worlds as it parameterizes $\Pr(I)$ (\autoref{eq:pr:n} and \ref{eq:p}). We start with a well known result in the MaxEnt model. If $\mathbf{c}_\ell$ is the linear query associated with the variable $\alpha_\ell$, then

\begin{equation}
\label{eq:ec}
  \E[\inner{\mathbf{c}_\ell}{\mathbf{I}}] = \frac{n \alpha_\ell}{P}\frac{\partial P}{\partial \alpha_\ell}.
\end{equation}
We review the proof here. The expected value of $\inner{\mathbf{c}_\ell}{\mathbf{I}}$ over the probability space (\autoref{eq:pr:n}) is

\begin{align*}
  \E[\inner{\mathbf{c}_\ell}{\mathbf{I}}] = &\frac{1}{P^n} \sum_{\mathbf{I}} \inner{\mathbf{c}_\ell}{\mathbf{I}}\prod_j \alpha_j^{\inner{\mathbf{c}_j}{\mathbf{I}}}
 =  \frac{1}{P^n} \sum_{\mathbf{I}} \frac{\alpha_\ell \partial}{\partial \alpha_\ell} \prod_j \alpha_j^{\inner{\mathbf{c}_j}{\mathbf{I}}} \\
 = & \frac{1}{P^n} \frac{\alpha_\ell \partial}{\partial \alpha_\ell} \sum_{\mathbf{I}} \prod_j \alpha_j^{\inner{\mathbf{c}_j}{\mathbf{I}}}
 =  \frac{1}{P^n} \frac{\alpha_\ell \partial P^n}{\partial \alpha_\ell} = \frac{n}{P}\frac{\alpha_\ell\partial P}{\partial \alpha_\ell}.
\end{align*}

To compute a new linear query $\mathbf{q}$, we add it to the statistical queries $c_j$, associate it with a fresh variable $\beta$,
and denote $P_{\mathbf{q}}$ the extended polynomial:

\begin{align}
\label{eq:pq}
  P_{\mathbf{q}}(\alpha_1, \ldots, \alpha_k,\beta) \eqdef \sum_{i=1,d} \prod_{j=1,k} \alpha_j^{\inner{\mathbf{c}_j}{\mathbf{t}_i}}\beta^{\inner{\mathbf{q}}{\mathbf{t}_i}}
\end{align}

Notice that $P_{\mathbf{q}}[\beta=1] \equiv P$; therefore, the extended data summary defines the same probability space as $P$. With $\beta = 1$, we can apply ~\autoref{eq:ec} to the query $\mathbf{q}$ to derive:

\begin{equation}
\label{eq:eq}
  \E[\inner{\mathbf{q}}{\mathbf{I}}] = \frac{n}{P}\frac{\partial P_{\mathbf{q}}}{\partial \beta}.
\end{equation}

This leads to the following na\"{i}ve strategy for computing the expected value of $\mathbf{q}$: extend $P$ to obtain $P_{\mathbf{q}}$ and apply formula~\autoref{eq:eq}. One way to obtain $P_{\mathbf{q}}$ is to iterate over all monomials in $P$ and add $\beta$ to the monomials corresponding to tuples counted by $\mathbf{q}$. As this iteration is inefficient,~\autoref{subsec:log_opt:aqp} describes how to avoid modifying the polynomial altogether.

\subsection{Probabilistic Model Computation}
\label{subsec:solving}
We now describe how to compute the parameters of the summary. Given the statistics $\Phi = \setof{(\mathbf{c}_j,s_j)}{j=1,k}$, we need to find values of the variables $\setof{\alpha_j}{j=1,k}$ such that $\E[\inner{\mathbf{c}_j}{\mathbf{I}}] = s_j$ for all $j=1,k$. As explained in~\autoref{sec:background}, this is equivalent to maximizing the dual function $\Psi$:

\begin{equation}
\label{eq:dual}
\Psi \eqdef \sum_{j=1}^k s_j \ln(\alpha_j) - n \ln P.
\end{equation}

Indeed, maximizing $P$ reduces to solving the equations $\partial \Psi/\partial \alpha_j = 0$ for all $j$. Direct calculation gives us $\partial \Psi/\partial \alpha_j = \frac{s_j}{\alpha_j} - \frac{n}{P}\frac{\partial P}{\partial \alpha_j} = 0$, which is equivalent to $s_j - \E[\inner{\mathbf{c}_j}{\mathbf{I}}]$ by ~\autoref{eq:ec}. The dual function $\Psi$ is concave, and hence it has a single maximum value that can be obtained using convex optimization techniques such as Gradient Descent.

In particular, we achieve fastest convergence rates using a variant of Stochastic Gradient Descent (SGD) called Mirror Descent~\cite{convex-optimization-algorithms-complexity},  where each iteration chooses some $j=1,k$ and updates $\alpha_j$ by solving $\frac{n \alpha_j}{P}\frac{\partial P}{\partial \alpha_j} = s_j$ while keeping all other parameters fixed. In other words, the step of SGD is chosen to solve $\partial \Psi/\partial \alpha_j = 0$. Denoting $P_{\alpha_{j}} \eqdef \frac{\partial P}{\partial \alpha_j}$ and solving, we obtain:
\begin{equation}
\label{eq:update_step} 
\alpha_{j} = \frac{s_j(P - {\alpha_{j}}P_{\alpha_{j}})}{(n-s_j)P_{\alpha_{j}}}.
\end{equation}
Since $P$ is linear in each $\alpha$, neither $P - {\alpha_{j}}P_{\alpha_{j}}$ nor $P_{\alpha_{j}}$ contain any $\alpha_{j}$ variables.

We repeat this for all $j$, and continue this process until all
differences $|s_j - \frac{n\alpha_{j}P_{\alpha_{j}}}{P}|$, $j=1,k$,
are below some threshold. Alg.~\ref{alg:solve} shows pseudocode for the solving process.
\begin{small}
\begin{algorithm}[t]
\caption{Solving for the $\alpha$s}
\label{alg:solve}
\begin{lstlisting}[style=myJava] 
maxError = infinity
while maxError >= threshold do
  maxError = -1
  for each *@$\alpha_j$@* do
      value = *@$\frac{s_j(P - {\alpha_{j}}P_{\alpha_{j}})}{(n-s_j)P_{\alpha_{j}}}$@*
      *@$\alpha_j$@* = value
      error = *@$value - \frac{n\alpha_{j}P_{\alpha_{j}}}{P}$@*
      maxError = max(error, maxError)
\end{lstlisting}
\end{algorithm}
\end{small}

%% file: logical_optimizations.tex
We now discuss two logical optimizations: (1) summary compression in Sec.~\ref{subsec:log_opt:compress} and (2) optimized query processing in Sec.~\ref{subsec:log_opt:aqp}. In~\autoref{sec:system_optimizations}, we discuss the implementation of these optimizations.

\subsection{Compression of the Data Summary}
\label{subsec:log_opt:compress}

The summary consists of the polynomial $P$ that, by definition, has $|Tup|$ monomials where $|Tup| = \prod_{i=1}^m N_i$. We describe a technique that compresses the summary by factorizing the polynomial to a size closer to $O(\sum_i N_i)$ than $O(\prod_i N_i)$.

Before walking through a more complex example describing the factorization process, we show the factorized version of the polynomial from~\autoref{ex:complex1}.
\begin{example}
\label{ex:complex1_fctored}
Recall that our relation has three attributes $A$, $B$, and $C$ with domain size of 2, and our summary has four multidimensional statistics. The factorization of P is
\begin{align}
\label{eq:pnew_factored}
P = &(\alpha_{1} + \alpha_{2})(\beta_{1} + \beta_{2})(\gamma_{1} + \gamma_{2}) + \nonumber \\
&(\gamma_{1} + \gamma_{2})(\alpha_{1}\beta_{1}(\mathcolor{red}{[\alpha\beta]_{1,1}} - 1) + \alpha_{2}\beta_{2}(\mathcolor{red}{[\alpha\beta]_{2,2}} - 1)) + \nonumber \\ 
&(\alpha_{1} + \alpha_{2})(\beta_{1}\gamma_{1}(\mathcolor{red}{[\beta\gamma]_{1,1}} - 1) + \beta_{2}\gamma_{1}(\mathcolor{red}{[\beta\gamma]_{2,1}} - 1)) + \nonumber \\
&\alpha_{1}\beta_{1}\gamma_{1}(\mathcolor{red}{[\alpha\beta]_{1,1}} - 1)(\mathcolor{red}{[\beta\gamma]_{1,1}} - 1) + \nonumber \\
&\alpha_{2}\beta_{2}\gamma_{1}(\mathcolor{red}{[\alpha\beta]_{2,2}} - 1)(\mathcolor{red}{[\beta\gamma]_{2,1}} - 1).
\end{align}
\end{example}
As we will see, the factorization starts with a product of 1-dimensional statistics and uses the inclusion/exclusion principle to include the multi-dimensional statistics. Note that for this particular example, because the active domain is so small (eight possible tuples), the factorized polynomial is not smaller than the expanded one. We explain the polynomial size in~\autoref{th:compress_size}.

We now walk through a more complex example with three attributes, $A$, $B$, and $C$, each with an active domain of size $N_1=N_2=N_3=1000$. Suppose first that we have only 1D statistics.  Then, instead of representing $P$ as a sum of $1000^3$ monomials, \ie $P = \sum_{i,j,k \in [1000]} \alpha_i\beta_j \gamma_k$, we factorize it to $P = (\sum \alpha_i)(\sum \beta_j)(\sum \gamma_k)$; the new representation has size $3 \cdot 1000$.

Now, suppose we add a single 3D statistic on $ABC$: $A = 3 \wedge B = 4 \wedge C=5$.  The new variable, call it $\delta$, occurs in a single monomial of $P$, namely $\alpha_3\beta_4\gamma_5\delta$.  Thus, we can compress $P$ to $(\sum \alpha_i)(\sum \beta_j)(\sum \gamma_k) + \alpha_3\beta_4\gamma_5(\delta-1)$.

Instead, suppose we add a single 2D range statistic  on $AB$, say $A \in [101, 200] \wedge B \in [501, 600]$ and call its associated variable $\delta_1$.  This will affect $100 \cdot 100 \cdot 1000$ monomials.  We can avoid enumerating them by noting that they, too, factorize.  The polynomial compresses to $(\sum \alpha_i)(\sum \beta_j)(\sum \gamma_k) + (\sum_{i=101}^{200}\alpha_i)(\sum_{j=501}^{600} \beta_j)(\sum \gamma_k)(\delta_1-1)$.

Finally, suppose we have three 2D statistics and one 3D statistic: the previous one on $AB$ plus the statistics $B \in [551, 650] \wedge C \in [801, 900]$ and $B \in [651, 700] \wedge C \in [701, 800]$ on $BC$ and $A \in [101, 150] \wedge B \in [551, 600] \wedge C \in [801, 850]$ on $ABC$.  Their associated variables are $\delta_1$, $\delta_2$, $\delta_3$, and $\delta_4$ (\autoref{fig:polynomial_join} shows a table of the statistics).  Now we need to account for the fact that $100 \cdot 50 \cdot 100$ monomials contain both $\delta_1$ and $\delta_2$ and that $50 \cdot 50 \cdot 50$ monomials contain $\delta_1$, $\delta_2$, and $\delta_4$.  Applying the inclusion/exclusion principle, $P$ compresses to the equation shown in~\autoref{fig:eq:ex:p} (the \textbf{i}, \textbf{ii}, and \textbf{iii} labels are referenced later). The size, counting only the $\alpha$s, $\beta$s, and $\gamma$s for simplicity, is $3000 + 1200 + 1350 + 150 + 250 + 150 + 150 + 150 = 6400 \ll 1000^3$.
\begin{figure}[t]
\begin{small}
  \begin{align*}
   P= & \overbrace{(\sum \alpha_i)(\sum \beta_j)(\sum \gamma_k)}^{(\textbf{i})}
    + \overbrace{(\sum \gamma_k)}^{(\textbf{ii})}\overbrace{(\sum_{101}^{200}\alpha_i)(\sum_{501}^{600} \beta_j)(\delta_1-1)}^{(\textbf{iii})} \nonumber \\
    + &\overbrace{(\sum \alpha_i)}^{(\textbf{ii})} \overbrace{\left[(\sum_{551}^{650} \beta_j)(\sum_{801}^{900} \gamma_k)(\delta_2-1)+(\sum_{651}^{700} \beta_j)(\sum_{701}^{800} \gamma_k)(\delta_3-1)\right]}^{(\textbf{iii})} \nonumber \\
    + & \overbrace{(\sum_{101}^{150}\alpha_i)(\sum_{551}^{600} \beta_j)(\sum_{801}^{850} \gamma_k)(\delta_4-1)}^{(\textbf{iii})} \nonumber \\
    + &\overbrace{(\sum_{101}^{200} \alpha_i)(\sum_{551}^{600} \beta_j)(\sum_{801}^{900} \gamma_k)(\delta_1-1)(\delta_2-1)}^{(\textbf{iii})} \nonumber \\
    + &\overbrace{(\sum_{101}^{150} \alpha_i)(\sum_{551}^{600} \beta_j)(\sum_{801}^{850} \gamma_k)(\delta_1-1)(\delta_4-1)}^{(\textbf{iii})} \nonumber \\
    + &\overbrace{(\sum_{101}^{150} \alpha_i)(\sum_{551}^{600} \beta_j)(\sum_{801}^{850} \gamma_k)(\delta_2-1)(\delta_4-1)}^{(\textbf{iii})} \nonumber \\
    + &\overbrace{(\sum_{101}^{150} \alpha_i)(\sum_{551}^{600} \beta_j)(\sum_{801}^{850} \gamma_k)(\delta_1-1)(\delta_2-1)(\delta_4-1)}^{(\textbf{iii})}.
  \end{align*}
\end{small}
\caption{Example of a compressed polynomial $P$ after applying the inclusion/exclusion principle.}
\label{fig:eq:ex:p}
\end{figure}

Before proving the general formula for $P$, note that this compression is related to standard algebraic factorization techniques involving kernel extraction and rectangle coverings \cite{hosangadi2004factoringae}; both techniques reduce the size of a polynomial by factoring out divisors. The standard techniques, however, are unsuitable for our use because they require enumeration of the product terms in the sum-of-product (SOP) polynomial to extract kernels and form cube matrices. Our polynomial in SOP form is too large to be materialized, making these techniques infeasible. We leave it as future work to investigate other factorization techniques geared towards massive polynomials.

We now make the following three assumptions for the rest of the paper.
\begin{myitemize}
  \item Each predicate has the form $\pi_j = \bigwedge_{i=1}^m \rho_{ij}$ where $m$ is the number of attributes, and $\rho_{ij}$ is the projection of $\pi_j$ onto $A_i$. If $j \in J_i$ ($J_i$ is the set of indices of the 1-dimensional statistics), then $\pi_j \equiv \rho_{ij}$. For any set of indices of multi-dimensional statistics $S \subset [k]$, we denote $\rho_{iS} \eqdef \bigwedge_{j \in S} \rho_{ij}$, and $\pi_S \eqdef \bigwedge_i \rho_{iS}$; as usual, when $S =\emptyset$, then $\rho_{i\emptyset} \equiv \texttt{true}$.

  \item Each $\rho_{ij}$ is a range predicate $A_i \in [u,v]$.

  \item For each $\mathcal{I} \subseteq [m]$, the multi-dimensional statistics whose attributes are exactly those in $\mathcal{I}$ are disjoint; \ie, for $j_1$, $j_2$ whose attributes are $\mathcal{I}$, $\rho_{i{j_1}}, \rho_{i{j_2}} \not\equiv \texttt{true}$ for $i \in \mathcal{I}$ (\ie there is a predicate on $A_i$), $\rho_{i{j_1}}, \rho_{i{j_2}} \equiv \texttt{true}$ for $i \not\in \mathcal{I}$, and $\pi_{j_1} \wedge \pi_{j_2} \equiv \texttt{false}$. Attributes for different $\mathcal{I}$ may overlap, but for a particular $\mathcal{I}$, there is no overlap.
\end{myitemize}

Using this, define $J_{\mathcal{I}} \subseteq [k]$ for $\mathcal{I} \subseteq [m]$ to be the set of indices of multi-dimensional statistics whose attributes are $\setof{A_i}{i \in \mathcal{I}}$. This means for $\mathcal{I}$ such that $|\mathcal{I}| = 1$, $J_{\mathcal{I}} = \emptyset$ because 1-dimensional statistics are not considered multi-dimensional statistics. Further, let $B_a = |\setof{\mathcal{I}}{J_{\mathcal{I}} \neq \emptyset}|$ be the number of unique multi-dimensional attributes we have statistics on and $B_s^{\mathcal{I}} = |J_{\mathcal{I}}|$ be the number of multi-dimensional statistics for the attribute set defined by $\mathcal{I}$. (These parameters are discussed further in ~\autoref{sec:stat_selection}).

Finally, define $J_{\mathcal{I}^+} \subseteq \mathcal{P}([k])$\footnote{$\mathcal{P}([k])$ is the power set of $\set{1, 2, \ldots, k}$} for $\mathcal{I}^+ \subseteq \mathcal{P}([m])$ to be the set of sets of the maximal number of multi-dimensional statistic indices from $\bigcup_{\mathcal{I} \in \mathcal{I}^+} J_{\mathcal{I}}$ such that each set's {\em combined} attributes are $\setof{A_i}{i \in \bigcup \mathcal{I}^+}$ and each set's intersection does not conflict (\ie, not \texttt{false}). In other words, for each $S \in J_{\mathcal{I}^+}$, $\rho_{iS} \not\in \{\texttt{true}, \texttt{false}\}$ for $i \in \bigcup \mathcal{I}^+$ and $\rho_{iS} \equiv \texttt{true}$ for $i \notin \bigcup \mathcal{I}^+$.  

For example, suppose we have the three 2D statistics and one 3D statistic from before: $\pi_{j_1} = A \in [101, 200] \wedge B \in [501, 600]$, $\pi_{j_2} = B \in [551, 650] \wedge C \in [801, 900]$, $\pi_{j_3} = B \in [651, 700] \wedge C \in [701, 800]$, and $\pi_{j_4} = A \in [101, 150] \wedge B \in [551, 600] \wedge C \in [801, 850]$. Then, some example $J_{\mathcal{I}^+}$ are: $J_{\{\{1, 2\}\}} = \{\{j_1\}\}$, $J_{\{\{2, 3\}\}} = \{\{j_2\}, \{j_3\}\}$, and $J_{\{\{1, 2, 3\}\}} = \{\{j_4\}\}$. $\{\{j_2, j_3\}\} \notin J_{\{\{2, 3\}\}}$ because $\rho_{2j_2} \wedge \rho_{2j_3} \equiv \texttt{false}$. Further, $J_{\{\{1, 2\}, \{2, 3\}\}} = \{\{j_1, j_2\}\}$ because $\rho_{2j_1} \wedge \rho_{2j_2} \not\equiv \texttt{false}$. However, $\{j_1, j_3\} \notin J_{\{\{1, 2\}, \{2, 3\}\}}$ because $\rho_{2j_1} \wedge \rho_{2j_3} \equiv \texttt{false}$, and $\{j_4\} \notin J_{\{\{1, 2\}, \{2, 3\}\}}$ because $j_4 \notin J_{\{1, 2\}}$ and $j_4 \notin J_{\{2, 3\}}$. Lastly, $J_{\{\{1, 2\}, \{2, 3\}, \{1, 2, 3\}\}} = \{\{j_1, j_2, j_4\}\}$

Using these definitions, we now get the compression shown in~\autoref{th:compress}.
\begin{theorem} \label{th:compress} The polynomial $P$ is equivalent to:
\begin{small}
  \begin{align*}
    P = &\overbrace{\left(\prod_{i \in [m]} \sum_{j \in J_i} \alpha_j \right)}^{(\textbf{i})} +\\
    & \left[\sum_{\substack{\mathcal{I} \subseteq [m] }}\overbrace{\left(\prod_{i \notin \mathcal{I}} \sum_{j \in J_i} \alpha_j \right)}^{(\textbf{ii})} \overbrace{\sum_{\ell = 1}^{B_a} \sum_{\substack{\mathcal{I}^+ \subseteq \mathcal{P}([m]), \\ |\mathcal{I}^+| = \ell, \\ \bigcup\mathcal{I}^+ = \mathcal{I}}}\sum_{S \in J_{\mathcal{I}^+}}}^{(\textbf{iii})}\right. \\
    &\left. \underbrace{\left(\prod_{i \in \bigcup \mathcal{I}^+} \sum_{\substack{j \in J_i, \\ \pi_j \land \rho_{iS} \not\equiv \texttt{false}}} \alpha_j \right)\left(\prod_{j \in S}(\alpha_j - 1) \right)}_{(\textbf{iii})} \right]
  \end{align*}
\end{small}
\end{theorem}
The proof uses induction on the size of $\mathcal{I}$.

To give intuition, the very first sum gives the sum over the 1D statistics, $(\textbf{i})$. The next sum handles the multi-dimensional statistics. When $\mathcal{I}$ is empty, $(\textbf{iii})$ will be zero. When there is no $\mathcal{I}^+$ matching the criteria or $J_{\mathcal{I}^+}$ is empty, that portion of the summation will be zero. When there exists some $S \in J_{\mathcal{I}^+}$, the summand sums up all 1-dimensional variables $\alpha_j$, $j\in J_i$ that are in the $i$th projection of the predicate $\pi_S$ (this is what the condition $(\pi_j \wedge \rho_{iS})\not\equiv\texttt{false}$ checks) and multiplies with terms $\alpha_j-1$ for $j \in S$.

Our algorithm to build the polynomial is non-trivial and is described in~\autoref{subsec:sys_opt:build_poly}. The algorithm can be used during query answering to compute the compressed representation of $P_{\mathbf{q}}$ from $P$ (Sec.~\ref{subsec:query:answer}) by rebuilding \textbf{iii} for the new $\mathbf{q}$. However, as this is inefficient and may increase the size of our polynomial, our system performs query answering differently, as explained in~\autoref{subsec:log_opt:aqp}.

We now analyze the size of the compressed polynomial $P$. Since $B_a < 2^m$ and $\sum_{i=1}^m N_i \gg 2^m$, $B_a$ is dominated by $\sum_{i=1}^m N_i$. For some $\mathcal{I}$, part $(\textbf{ii})$ of the compression is $O(\sum_{i=1}^m N_i)$. Part $(\textbf{iii})$ of the compression is more complex. For some $S \in J_{\mathcal{I}^+}$, the innermost summand is of size $O(\sum_{i=1}^m N_i + |S|)$. As $|S| \leq B_a \ll \sum_{i=1}^m N_i$, the summand is only $O(\sum_{i=1}^m N_i)$. This innermost summand only occurs when $J_{\mathcal{I}^+}$ is nonempty, which happens once for all possible combinations of the $B_a$ multi-dimensional attributes. Therefore, letting $R = \max_{\mathcal{I}^+}|J_{\mathcal{I}^+}|$ (we discuss this next), part $(\textbf{iii})$ is of size $O(2^{B_a} R \sum_{i=1}^m N_i)$. Putting it together, since we only are concerned with $\mathcal{I}$ such that $\bigcup \mathcal{I}^+ = \mathcal{I}$ for some $\mathcal{I}^+$ and we have $2^{B_a}$ relevant $\mathcal{I}^+$, the polynomial is of size $O(\sum_{i=1}^m N_i + 2^{B_a} (\sum_{i=1}^m N_i + 2^{B_a} R \sum_{i=1}^m N_i)) = O(2^{B_a} R \sum_{i=1}^m N_i)$.

Lastly, to discuss $R$. For a particular $\mathcal{I}^+$, $|J_{\mathcal{I}^+}|$ is the number of sets of multi-dimensional statistics whose {\em combined} attributes are $\setof{A_i}{i \in \bigcup \mathcal{I}^+}$ and whose intersection does not conflict. In the worse case, there are no conflicts (\eg if $\bigcap \mathcal{I}^+ = \emptyset$). Then, there will be at most $\prod_{\mathcal{I} \in \mathcal{I}^+} B_s^{\mathcal{I}}$ statistics for a particular $\mathcal{I}^+$. Since the largest $\mathcal{I}^+$ has $B_a$ elements, an upper bound on $R$ is $\hat{B}_s^{B_a}$ where $\hat{B}_s = \max_{\mathcal{I}}B_s^{\mathcal{I}}$. We assume $\hat{B}_s \geq 2$, and therefore we get the following theorem.

\begin{theorem}
\label{th:compress_size}
The size of the polynomial is $O(\hat{B}_s^{B_a} \sum_{i=1}^m N_i)$ where $B_a$ is the number of unique multi-dimensional attribute sets and $\hat{B}_s$ is the largest number of statistics over some $\mathcal{I}$.
\end{theorem}

In the worst case, if one gathers all possible multi-dimensional statistics, this compression will be worse than the uncompressed polynomial, which is of size $O(\prod_{i = 1}^m N_i)$, approximately equal to $O((\max_i N_i)^m)$. However, in practice, $B_a < m$ and $B_s \leq \max_i N_i$ which results in a significant reduction of polynomial size to one closer to $O(\sum_{i=1}^m N_i)$ than $O(\prod_{i = 1}^m N_i)$.

\subsection{Optimized Query Answering}
\label{subsec:log_opt:aqp}
In this section, we assume that the query $\mathbf{q}$ is a counting query defined by a conjunction of predicates, one over each attribute $A_i$; \ie, $\mathbf{q} = |\sigma_\pi(I)|$, where
\begin{equation}
\label{eq:piq}
\pi = \rho_1 \wedge \cdots \wedge \rho_m 
\end{equation}

and $\rho_i$ is a predicate over the attribute $A_i$. If $\mathbf{q}$ ignores $A_i$, then we simply set $\rho_i \equiv \texttt{true}$.  Our goal is to compute $\E[\inner{\mathbf{q}}{\mathbf{I}}]$. In Sec.~\ref{subsec:query:answer}, we described a direct approach that consists of constructing a new polynomial $P_{\mathbf{q}}$ and returning ~\autoref{eq:eq}. However, as described in Sec.~\ref{subsec:query:answer} and Sec.~\ref{subsec:log_opt:compress}, this may be expensive.

We describe here an optimized approach to compute $\E[\inner{\mathbf{q}}{\mathbf{I}}]$ directly from $P$. The advantage of this method is that it does not require any restructuring or rebuilding of the polynomial. Instead, it can use any optimized oracle for evaluating $P$ on given inputs. Our optimization has two parts: a new formula $\E[\inner{\mathbf{q}}{\mathbf{I}}]$ and a new formula for derivatives.

{\bf New formula for $\E[\inner{\mathbf{q}}{\mathbf{I}}]$:} Let $\pi_j$ be the predicate associate to the $j$th statistical query. In other words, $\inner{\mathbf{c}_j}{\mathbf{I}} = |\sigma_{\pi_j}(\mathbf{I})|$. The next lemma applies to any query $\mathbf{q}$ defined by some predicate $\pi$. Recall that $\beta$ is the new variable associated to $\mathbf{q}$ in $P_{\mathbf{q}}$ (Sec.~\ref{subsec:query:answer}).

\begin{lemma}\label{lemma:query_derivative} For any $\ell$ variables $\alpha_{j_1}, \ldots, \alpha_{j_\ell}$ of $P_{\mathbf{q}}$:

  (1) If the logical implication $\pi_{j_1} \wedge \cdots \wedge \pi_{j_\ell} \Rightarrow \pi$ holds, then
  \begin{align}
    \label{eq:aux}
    \frac{\alpha_{j_1}\cdots \alpha_{j_\ell}\partial^\ell P_{\mathbf{q}}}{\partial \alpha_{j_1} \cdots \partial \alpha_{j_\ell}} = &
    \frac{\alpha_{j_1}\cdots \alpha_{j_\ell} \beta \partial^{\ell+1} P_{\mathbf{q}}}{\partial \alpha_{j_1} \cdots \partial \alpha_{j_\ell}\partial \beta}
  \end{align}

  (2) If the logical equivalence $\pi_{j_1} \wedge \cdots \wedge \pi_{j_\ell} \Leftrightarrow \pi$ holds, then
  \begin{align}
    \frac{\alpha_{j_1}\cdots \alpha_{j_\ell}\partial^\ell P_{\mathbf{q}}}{\partial \alpha_{j_1} \cdots \partial \alpha_{j_\ell}} = \frac{\beta \partial P_{\mathbf{q}}}{\partial \beta} \label{eq:aux2}
  \end{align}
\end{lemma}

\begin{proof}
  (1) The proof is immediate by noting that every monomial of $P_{\mathbf{q}}$ that contains all variables $\alpha_{j_1}, \ldots, \alpha_{j_\ell}$ also contains $\beta$; therefore, all monomials on the LHS of ~\autoref{eq:aux} contain $\beta$ and thus remain unaffected by applying the operator $\beta \partial / \partial \beta$.

  (2) From item (1), we derive ~\autoref{eq:aux}; we prove now that the RHS of ~\autoref{eq:aux} equals $\frac{\beta \partial P_{\mathbf{q}}}{\partial \beta}$.  We apply item (1) again to the implication $\pi \Rightarrow \pi_{j_1}$ and obtain $\frac{\beta \partial P_{\mathbf{q}}}{\partial \beta} = \frac{\beta \alpha_{j_1} \partial^2 P_{\mathbf{q}}}{\partial \beta \partial \alpha_{j_1}}$ (the role of $\beta$ in ~\autoref{eq:aux} is now played by $\alpha_{j_1}$). As $P$ is linear, the order of partials does not matter, and this allows us to remove the operator $\alpha_{j_1}\partial/\partial \alpha_{j_1}$ from the RHS of ~\autoref{eq:aux}.  By repeating the argument for $\pi \Rightarrow \pi_{j_2}$, $\pi \Rightarrow \pi_{j_3}$, etc, we remove $\alpha_{j_2}\partial/\partial \alpha_{j_2}$, then $\alpha_{j_3}\partial/\partial \alpha_{j_3}$, etc from the RHS.
\end{proof}

\begin{corollary} (1) Assume $\mathbf{q}$ is defined by a point predicate $\pi = (A_1=v_1 \wedge \cdots \wedge A_\ell=v_\ell)$ for some $\ell \leq m$.  For each $i=1,\ell$, denote $j_i$ the index of the statistic associated to the value $v_i$. In other words, the predicate $\pi_{j_i} \equiv (A_i = v_i)$.  Then,
  \begin{align}
    \label{eq:pqopt}
    \E[\inner{\mathbf{q}}{\mathbf{I}}] = & \frac{n}{P}\frac{\alpha_{j_1}\cdots \alpha_{j_\ell} \partial^\ell P}{\partial \alpha_{j_1} \cdots \partial \alpha_{j_\ell}}
  \end{align}
  (2) Let $\mathbf{q}$ be the query defined by a predicate as in ~\autoref{eq:piq}.  Then,
  \begin{align}
    \label{eq:pqopt2}
    \E[\inner{\mathbf{q}}{\mathbf{I}}] = & \sum_{j_1 \in J_1: \pi_{j_1} \Rightarrow \rho_1} \cdots \sum_{j_m \in J_m: \pi_{j_m} \Rightarrow \rho_m} \frac{n}{P}\frac{\alpha_{j_1}\cdots \alpha_{j_m} \partial^m P}{\partial \alpha_{j_1} \cdots \partial \alpha_{j_m}}
  \end{align}
\end{corollary}

\begin{proof}
(1) ~\autoref{eq:pqopt} follows from ~\autoref{eq:eq}, ~\autoref{eq:aux2}, and the fact that $P_{\mathbf{q}}[\beta=1] \equiv P$.
(2) Follows from (1) by expanding $\mathbf{q}$ as a sum of point queries as in Lemma.~\ref{lemma:query_derivative} (1).
\end{proof}

In order to compute a query using ~\autoref{eq:pqopt2}, we would have to examine all $m$-dimensional points that satisfy the query's predicate, convert each point into the corresponding 1D statistics, and use ~\autoref{eq:pqopt} to estimate the count of the number of tuples at this point. Clearly, this is inefficient when $\mathbf{q}$ contains any range predicate containing many point queries.

{\bf New formula for derivatives} Thus, to compute $\E[\inner{\mathbf{q}}{\mathbf{I}}]$, one has to evaluate several partial derivatives of $P$. Recall that $P$ is stored in a highly compressed format, and therefore, computing the derivative may involve nontrivial manipulations. Instead, we use the fact that our polynomial is overcomplete, meaning that $P = \sum_{j \in J_i} \alpha_j P_j$, where $P_j$, $j \in J_i$ does not depend on any variable in $\setof{\alpha_j}{j \in J_i}$ (\autoref{eq:p1i}). Let $\rho_i$ be any predicate on the attribute $A_i$.  Then,
\begin{align}
\sum_{j_i \in J_i: \pi_{j_i} \Rightarrow \rho_i} \frac{\alpha_{j_i} \partial P}{\partial \alpha_{j_i}} = & P[\bigwedge_{j \in J_i: \pi_{j_i} \not\Rightarrow \rho_i} \alpha_j = 0]
\end{align}
Thus, in order to compute the summation on the left, it suffices to compute $P$ after setting to $0$ the values of all variables $\alpha_j$, $j \in J_i$ that do not satisfy the predicate $\rho_i$ (this is what the condition $\pi_{j_i} \not\Rightarrow \rho_i$ checks).

Finally, we combine this with ~\autoref{eq:pqopt2} and obtain the following, much simplified formula for answering a query $\mathbf{q}$, defined by a predicate of the form ~\autoref{eq:piq}:
\begin{align}
\label{eq:pqopt3}
  \E[\inner{\mathbf{q}}{\mathbf{I}}] = \frac{n}{P} P[\bigwedge_{i=1,m} \bigwedge_{j \in J_i: \pi_{j_i} \not\Rightarrow \rho_i} \alpha_j = 0]
\end{align}
In other words, we set to 0 all 1D variables $\alpha_j$ that correspond to values that do {\em not} satisfy the query, evaluate the polynomial $P$, and multiply it by $\frac{n}{P}$ (which is a precomputed constant independent of the query). For example, if the query ignores an attribute $A_i$, then we leave the 1D variables for that attribute, $\alpha_j$, $j \in J_i$, unchanged. If the query checks a range predicate, $A_i \in [u, v]$, then we set $\alpha_j=0$ for all 1D variables $\alpha_j$ corresponding to values of $A_i$ outside that range.

\begin{example}
\label{ex:query_answer}
Consider three attributes $A$, $B$, and $C$ each with domain 1000 and three multi-dimensional statistics: one $AB$ statistic $A \in [101, 200] \wedge B \in [501, 600]$, two $BC$ statistics $B \in [551, 650] \wedge C \in [801, 900]$ and $B \in [651, 700] \wedge C \in [701, 800]$, and one $ABC$ statistic $A \in [101, 150] \wedge B \in [551, 600] \wedge C \in [801, 850]$.  The polynomial $P$ is shown in ~\autoref{fig:eq:ex:p}.  Consider the query $\mathbf{q}$:
\begin{small}
\begin{lstlisting}
SELECT COUNT(*) FROM R 
WHERE A in [36,150] AND C in [660,834]
\end{lstlisting}
\end{small}
We estimate $\mathbf{q}$ using our formula $\frac{n}{P} P[\alpha_{1:35}=0,\ \alpha_{151:1000}=0,\ \gamma_{1:659}=0,\ \gamma_{835:1000}=0]$. There is no need to compute a representation of a new polynomial.
\end{example}

%% file: system_optimizations.tex
In~\autoref{sec:logical_optimizations}, we discussed two main optimizations: polynomial compression and query answering by setting certain 1D variables to zero. We now discuss how to implement these optimizations efficiently and analyze the runtime.

\subsection{Building the Polynomial}
\label{subsec:sys_opt:build_poly}
Recall from~\autoref{th:compress} that building our compressed polynomial starts with the product of the sum of all 1D statistics and then uses the inclusion/exclusion principle to modify the terms to include the correct multi-dimensional statistics. The terms that need to be modified are those that satisfy the predicates associated with the sets of non-conflicting multi-dimensional statistics. \ie, for some $\mathcal{I}^+$, the terms to be modified are the 1D terms $\alpha_j$ such that $\pi_j \wedge \pi_{S} \not\equiv \texttt{false}$ for $S \in J_{\mathcal{I}^+}$. The algorithmic challenge is how we find non-conflicting statistics $J_{\mathcal{I}^+}$ for some $\mathcal{I}^+$ and, once we know $J_{\mathcal{I}^+}$, how we find the 1D statistics that need to be modified.

To solve the latter problem, assume we have some $J_{\mathcal{I}^+}$. Since each multi-dimensional statistic is a range predicate over the elements in our domain and we have complete 1D statistics over the elements in our domain, once we know the range predicate, $\pi_S$ for $S \in J_{\mathcal{I}^+}$, we can easily find the associated 1D statistics.

Take the example in~\autoref{fig:eq:ex:p} which we will refer to throughout this section. If we know that $J_{\{\{1,2\},\{2,3\},\{1,2,3\}\}} = \{\{j_1, j_2, j_4\}\}$, then, by examining the range predicates associated with those three multi-dimensional statistics, we can determine that $\alpha_i$ for $i \in [101,150]$, $\beta_j$ for $j \in [551, 600]$, and $\gamma_k$ for $k \in [801, 850]$ are the 1D statistics that need to be modified to include $\delta_1$, $\delta_2$, and $\delta_4$.

The other problem is how we find the groups of multi-dimensional statistics that do not conflict for some group of attribute sets (\ie the $J_{\mathcal{I}^+} \neq \emptyset$). To solve this, we assume we are given four inputs: a list of attributes, a list of $B_a$ multi-dimensional attribute sets (\eg [$AB$, $BC$, $ABC$] for~\autoref{fig:eq:ex:p}), a dictionary of the 1D statistics with attributes $A_i$ as keys (denoted {\tt\textcolor{OliveGreen}{1DStats}}), and a dictionary of multi-dimensional statistics with indices into the list of $B_a$ attributes sets as keys (denoted {\tt\textcolor{OliveGreen}{multiDStats}}).

A straightforward algorithm to build the polynomial is shown in Alg.~\ref{alg:P_building_unop} where the red notations indicate the time complexity of each section of pseudo code. The function \textcolor{Maroon}{combinations}({\tt k}, $B_a$) generates a list of all possible length $k$ index sets from $[1, B_a]$. Note that we abuse the notation for dictionary selection slightly in that if \texttt{idx} is $\{1, 2\}$, {\tt\textcolor{OliveGreen}{multiDStats}[idx]} would select both the multi-dimensional stats of 1 and 2, \eg $AB$ and $BC$, and {\tt\textcolor{OliveGreen}{1DStats}[not idx]} would select no 1D stats since all attributes are used in $AB$ and $BC$. The function \textcolor{Maroon}{buildTerm}({\tt group}) builds the term shown in the last line of~\autoref{th:compress}. It generates a sum of the 1D statistics associated with the group and multiplies the sum by one minus the multi-dimensional variables in the group. 

The last function to discuss is \textcolor{Maroon}{findNoConflictGrps} which returns a dictionary with keys as sets of multi-dimensional attribute indices and values of groups of conflict free statistics. For example, for $k = 2$, a key would be $\{1, 2\}$ with value $\{\delta_1, \delta_2\}$ indicating that $\delta_1$ and $\delta_2$ do not conflict. The algorithm works by treating each multi-dimensional index set in \texttt{idx} as a relation with rows of the statistics associated with that index set. It then computes a theta-join over these relations with the join condition being if the statistics are conflict free.

For example, $\delta_1$ and $\delta_2$ are conflict free but not $\delta_1$ and $\delta_3$. Further, statistics over disjoint attributes sets are also conflict free. If we had a relation $R(A,B,C,D)$ and some statistic over $AB$ and another over $CD$, all of those multi-dimensional statistics from $AB$ would be satisfiable with all other from $CD$. 

\begin{small}
\begin{algorithm}[t]
\caption{Unoptimized Building P}
\label{alg:P_building_unop}
\begin{lstlisting}[style=myJava]
// add part *@\textcolor{Gray}{\textbf{i}}@* to P
P = *@\textcolor{Maroon}{1DProdSum}@*(*@\tt\textcolor{OliveGreen}{1DStats}@*)*@\tikzmark{line1}@*
for (k in [1:*@$B_a$@*]) do
  for (idx in *@\textcolor{Maroon}{combinations}@*(k, *@$B_a$@*)) do
    // add part *@\textcolor{Gray}{\textbf{ii}}@* to P
    P.addTerms(*@\textcolor{Maroon}{1DProdSum}@*(*@\tt\textcolor{OliveGreen}{1DStats}@*[not idx]))
    satGrps = *@\textcolor{Maroon}{findNoConflictGrps}@*(*@\tt\textcolor{OliveGreen}{multiDStats}@*[idx])*@\tikzmark{line3}@*
    // add part *@\textcolor{Gray}{\textbf{iii}}@* to P          *@\tikzmark{line4}@*
    for (group in satGrps) do *@\tikzmark{top3}@*
      P.addTerms(*@\textcolor{Maroon}{buildTerm}@*(group)) *@\tikzmark{bottom3}@*
\end{lstlisting}
\begin{tikzpicture}[overlay, remember picture]
    \drawBrace[1em]{top3}{bottom3}{{\scriptsize $O(mN + \hat{B}_s^k)$}};
    \drawLineWeird{line3}{line4}{{\scriptsize $O((k-1)(mN)^{2}\hat{B}_s^k)$}}
    \drawLine{line1}{{\scriptsize $O(mN)$}}
\end{tikzpicture}
\end{algorithm}
\end{small}

To understand the runtime complexity of the algorithm, start with the function \textcolor{Maroon}{buildTerm}({\tt group}). For ease of notation, we will denote $N = \max_i N_i$. Recall that $mN$ is the total number of distinct values across all attributes, $B_a$ is the number of attribute sets, and $\hat{B}_s$ is the largest number of statistics per attribute set. The runtime of \textcolor{Maroon}{buildTerm}({\tt group}) for a single {\tt satGrps} of size $k$ is $O(mN + \hat{B}_s^k)$ because a single {\tt satGrps} will only add each 1D statistic at most once and at most $\hat{B}_s^k$ correction terms. This runtime also includes the time for \textcolor{Maroon}{1DProdSum} because the 1D statistics that are not in \texttt{satGrps} will be added in \textcolor{Maroon}{1DProdSum}.

The runtime of \textcolor{Maroon}{findNoConflictGrps} involves computing the cross product of the multi-dimensional statistics and comparing the 1D statistics associated with each multi-dimensional statistic to determine if there is a conflict. Specifically, it computes a right deep join tree of the multi-dimensional statistics, and at each step in the tree, iterates over the 1D statistics in each right child conflict free group to see if there is a conflict or not with one of the incoming left child multi-dimensional statistics.

\autoref{fig:polynomial_join} shows the \textcolor{Maroon}{findNoConflictGrps} join algorithm for the attribute sets $AB$, $BC$, and $ABC$ with one added statistic $\delta_5$ on $BC$ of $B \in [401,550] \land C \in [751,850]$. The function to find and return a single conflict free group is \textcolor{Maroon}{CFG}({$\delta_L$, $\{\delta\}_S$}) (stands for conflict free group) where $\delta_L$ stands for the left multi-dimensional statistic and $\{\delta\}_S$ stands for the right, current conflict free group. The $S$ subscript is because we are building a new set of multi-dimensional statistics to add to some $J_{\mathcal{I}^+}$. We are abusing notation slightly because in~\autoref{sec:logical_optimizations}, $S$ stood for the set of indices whereas here, it sands for the set of statistic variables.

\textcolor{Maroon}{CFG} first determines which attributes are shared by $\delta_L$ and $\{\delta\}_S$. Then, for each such attribute, it iterates over $\delta_L$'s associated 1D statistics (at most $mN$ of them) and checks if at least one of these statistics also exists in the 1D statistics associated with $\{\delta\}_S$ (containment has runtime $mN$). This ensures, for all attributes $A_i$, that $\rho_{iL} \land \rho_{iS} \not\equiv \texttt{false}$. If $A_i$ is shared, then $\rho_{iL} \land \rho_{iS}$ will contain the shared 1D statistic found earlier, and if $A_i$ is not shared, then $\rho_{iL} \land \rho_{iS} \equiv \texttt{true}$. Note that a conflict free group is not always found. Therefore, the runtime of the join is $O(2(mN)^2\hat{B}_s^3)$. For joins of arbitrary size, the runtime is $O((mN)^2(k-1)\hat{B}_s^k)$.

Putting it all together we get the runtime of
\begin{multline*}
=mN + \sum_{k=1}^{B_a}\binom{B_a}{k}[mN + (k-1)(mN)^{2}\hat{B}_s^k + \hat{B}_s^k]
\end{multline*}
\vspace{-13pt}
\begin{multline*}
=mN + (2^{B_a}-1)mN + (\hat{B}_s+1)^{B_a} - 1 +\\
(mN)^2\left[\left[\sum_{k=0}^{B_a}\binom{B_a}{k}[k\hat{B}_s^k - \hat{B}_s^k]\right] + 1\right]
\end{multline*}
\vspace{-13pt}
\begin{multline*}
=mN + (2^{B_a}-1)mN + (\hat{B}_s+1)^{B_a} - 1 + (mN)^2 +\\
(mN)^2\left[\sum_{k=0}^{B_a}\binom{B_a}{k}k\hat{B}_s^k\right] - (mN)^2(\hat{B}_s + 1)^{B_a}
\end{multline*}
\vspace{-13pt}
\begin{multline*}
=mN + (2^{B_a}-1)mN + (\hat{B}_s+1)^{B_a} - 1 + (mN)^2 +\\
(mN)^2B_a\hat{B}_s(\hat{B}_s + 1)^{B_a - 1} - (mN)^2(\hat{B}_s + 1)^{B_a}
\end{multline*}

\begin{figure}[t]
    \centering
    \includegraphics[width=0.47\textwidth]{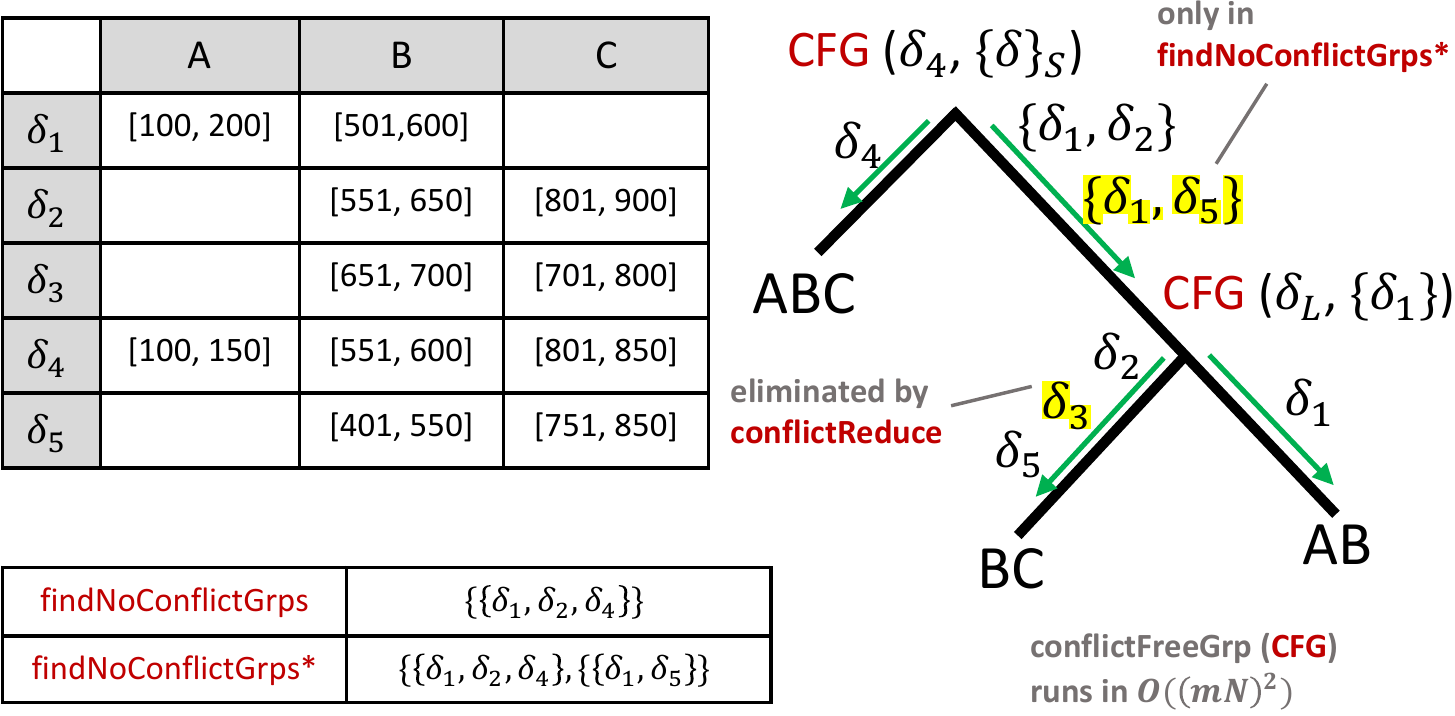}
    \caption{Figure of \textcolor{Maroon}{findNoConflictGrps} and \textcolor{Maroon}{findNoConflictGrps*} join algorithm for the three attribute sets $AB$, $BC$, and $ABC$. The highlighted statistics of $\delta_3$ and $\{\delta_1,\delta_5\}$ are the difference between the unoptimized and optimized algorithms.}
    \label{fig:polynomial_join}
\end{figure}

This algorithm, however, is suboptimal because it must run \textcolor{Maroon}{findNoConflictGrps} for all $2^{B_a}$ attribute sets. It is better to run \textcolor{Maroon}{findNoConflictGrps} once for the all multi-dimensional statistics (\ie compute the full theta-join of all $B_a$ attribute sets) and reconstruct the smaller groups without paying the cost of checking for conflicts (\ie perform selections over the full theta-join). Further, there are some statistics that will not appear in any other term besides when $k = 1$ in the loop above. Take $\delta_3$, for example. It is handled in line 2 of~\autoref{fig:eq:ex:p} but does not appear later on. These insights lead to a more optimized algorithm in Alg.~\ref{alg:P_building}.

The function \textcolor{Maroon}{conflictReduce} is like a semi-join reduction. It removes multi-dimensional statistics that do not appear in any conflict free group later on. For example, \texttt{redMDStats} would only contain $\delta_1$, $\delta_2$, and $\delta_4$. The function \textcolor{Maroon}{findNoConflictGrps*} acts just as \textcolor{Maroon}{findNoConflictGrps} except instead of computing an inner theta join, it computes a full outer theta join. The reason being that \texttt{satGrps} needs to keep track of all conflict free groups even if they contain less that $B_a$ statistics. For example, take the $\delta_5$ statistic of $BC$ as shown in~\autoref{fig:polynomial_join}. $\delta_1$ and $\delta_5$ are conflict free, but $\delta_1$, $\delta_5$, and $\delta_4$ are not conflict free because $\delta_5$ conflicts with $\delta_4$ in attribute $B$. In this case, \textcolor{Maroon}{findNoConflictGrps*} would return a dictionary with the keys $\{1, 2, 3\}$ and $\{1, 2\}$ and values $\{\delta_1, \delta_2, \delta_4\}$ and $\{\delta_1, \delta_5\}$, respectively. The outer join ensures $\{\delta_1, \delta_5\}$ is not lost. Note the time complexity of \textcolor{Maroon}{findNoConflictGrps} and \textcolor{Maroon}{findNoConflictGrps*} is the same because they must compute the full cross product and then filter. Lastly, the \texttt{group[idx]} index selection selects the statistics associated with the attribute sets in idx.

\begin{small}
\begin{algorithm}[t]
\caption{Optimized Building P}
\label{alg:P_building}
\begin{lstlisting}[style=myJava]
// add part *@\textcolor{Gray}{\textbf{i}}@* to P
P = *@\textcolor{Maroon}{1DProdSum}@*(*@\tt\textcolor{OliveGreen}{1DStats}@*)*@\tikzmark{line1}@*
// add terms when k = 1
for (idx in [1:*@$B_a$@*]) do *@\tikzmark{top1}@*
  // add part *@\textcolor{Gray}{\textbf{ii}}@* to P
  P.addTerms(*@\textcolor{Maroon}{1DProdSum}@*(*@\tt\textcolor{OliveGreen}{1DStats}@*[not idx]))
  // add part *@\textcolor{Gray}{\textbf{iii}}@* to P
  for (group in *@\tt\textcolor{OliveGreen}{multiDStats}@*[idx]) do
    P.addTerms(*@\textcolor{Maroon}{buildTerm}@*(group)) *@\tikzmark{bottom1}@*
redMDStats = *@\textcolor{Maroon}{conflictReduce}@*(*@\tt\textcolor{OliveGreen}{multiDStats}@*)*@\tikzmark{line2}@*
satGrps = *@\textcolor{Maroon}{findNoConflictGrps*}@*(redMDStats)*@\tikzmark{line3}@*
                         *@\tikzmark{line4}@*
for (k in [2:*@$B_a$@*]) do *@\tikzmark{top2}@*
  for (idx in *@\textcolor{Maroon}{combinations}@*(k, *@$B_a$@*)) do
    // add part *@\textcolor{Gray}{\textbf{ii}}@* to P
    P.addTerms(*@\textcolor{Maroon}{1DProdSum}@*(*@\tt\textcolor{OliveGreen}{1DStats}@*[not idx]))
    // add part *@\textcolor{Gray}{\textbf{iii}}@* to P
    for (group in satGrps) do
      P.addTerms(*@\textcolor{Maroon}{buildTerm}@*(group[idx]))*@\tikzmark{bottom2}@*
\end{lstlisting}
\begin{tikzpicture}[overlay, remember picture]
    \drawBraceRotate[2.1em]{top2}{bottom2}{\scriptsize
    \begin{minipage}{1cm}
            \begin{align*}
                O(&(2^{B_a}-B_a - 1)*\\&(mN+\hat{B}_s^{B_a}))
            \end{align*}
        \end{minipage}
    };
    \drawBrace[4.2em]{top1}{bottom1}{\scriptsize
    \begin{minipage}{1cm}
            \begin{align*}
                O(&mNB_a+\\&B_a\hat{B}_s)
            \end{align*}
        \end{minipage}
    };
    \drawLineWeird{line3}{line4}{{\scriptsize $O(({B_a}-1)(mN)^{2}\hat{B}_s^{B_a})$}}
    \drawLine[-0.3em]{line2}{{\scriptsize $O(\binom{B_a}{2}(B_a\hat{B}_s)^2)$}}
    \drawLine{line1}{{\scriptsize $O(mN)$}}
\end{tikzpicture}
\end{algorithm}
\end{small}

We will now show that this algorithm's time complexity is more optimal than Alg.~\ref{alg:P_building_unop} because although it loops through \texttt{satGrps}, selecting out a subterm is faster than rebuilding one, especially after semi-join reduction. Note that if \texttt{group[idx]} is has already been added to the polynomial from a previous group, it is just ignored when \texttt{addTerms} is called. As the red notation indicates, the runtime of the first for loop is $O(mNB_a+B_a\hat{B}_s)$ because for each \texttt{idx}, there are $\hat{B}_s$ multi-dimensional statistics and $mN$ 1D statistics to add the term.

The runtime of \textcolor{Maroon}{conflictReduce} involves comparing pairs of multi-dimensional statistics to see if they will participate in any conflict free groups of size two or more. For each $\binom{B_a}{2}\hat{B}_s^2$ possible pairs of multi-dimensional statistics, the conflict checking requires examining the 1D statistics of the pair, just like in \textcolor{Maroon}{CFG}({$\delta_L$, $\{\delta\}_S$}). The next function, \textcolor{Maroon}{findNoConflictGrps*}, has the same runtime as before except instead of being run for all $k$, it is run only once for $k=B_a$.

The last part to analyze is the for loop that iterates over all \texttt{satGrps}. In our runtime analysis, we add a percentage $p \in [0, 1]$ to indicate that only a fraction of the possible $\hat{B}_s^{B_a}$ groups are used in the last loop. This decrease is because of \textcolor{Maroon}{conflictReduce} and because, in practice, there are drastically fewer than $\hat{B}_s^{B_a}$ resulting conflict free groups. In practice, $p \leq 0.1$. As the inner most for loop is that same as in Alg.~\ref{alg:P_building_unop} except for $k = B_a$, the runtime is $O(\hat{B}_s^{B_a} + mN)$. As this happens for all combinations from $k = 2, B_a$, the overall runtime is $O((2^{B_a} - B_a - 1)(p\hat{B}_s^{B_a} + mN))$ where the minus is because $k$ starts at two instead of zero.

Adding up the different runtime components, we get the overall runtime of Alg.~\ref{alg:P_building} is
\begin{multline*}
=mN + B_a(mN + \hat{B}_s) + \binom{B_a}{2}(mN\hat{B}_s)^2 + \\(B_a-1)(mN)^{2}\hat{B}_s^{B_a} + (2^{B_a}-B_a-1)(mN+p\hat{B}_s^{B_a})
\end{multline*}

To show the improvement of the optimized algorithm,~\autoref{fig:p_runtime_analysis} shows the runtime difference between Alg.~\ref{alg:P_building_unop} and Alg.~\ref{alg:P_building} (\ie Alg.~\ref{alg:P_building_unop} - Alg.~\ref{alg:P_building}) when $mN = 5000$ (the trends are similar for other values of $mN$). The three columns are for $B_a = 2, 3, 4$, the colors represents the different values of $p$, and $\hat{B}_s$ varies from 100 to 2000. Note that the y-axis of the three columns are on a different scale in order to show the variation between the different values of $p$.

\begin{figure}[t]
    \centering
    \includegraphics[width=0.48\textwidth]{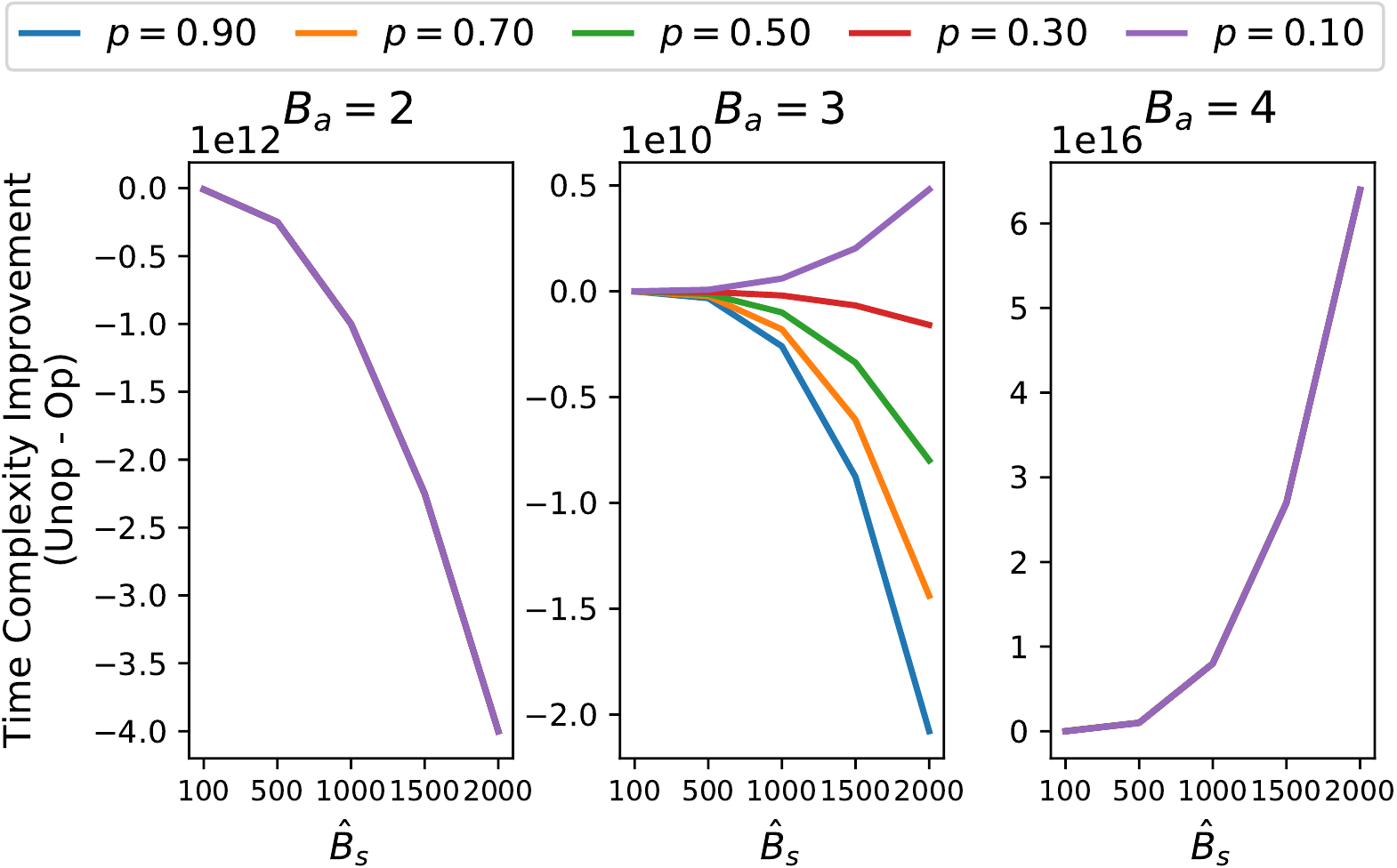}
    \caption{Algorithmic complexity improvement of the unoptimized algorithm over the optimized algorithm (complexity difference) for $mN = 5000$ and varying $B_a$, $\hat{B}_s$, and $p$. A positive number indicates the optimized algorithm is faster.}
    \label{fig:p_runtime_analysis}
\end{figure}

We see that $p$ matters for $B_a = 3$, and when $p$ falls between 0.3 and 0.1, the optimized version is faster. As, in practice, $p \leq 0.1$, the optimized version is best for $B_a > 2$. The trend shown in $B_a = 4$ is the same for $B_a > 4$ and thus not included in the plot. This shows that asymptotically, Alg.~\ref{alg:P_building} is optimal.

\subsection{Polynomial Evaluation}
\label{subsec:sys_opt:poly_eval}
Recall from~\autoref{subsec:log_opt:aqp} and~\autoref{ex:query_answer} that for a linear query $\mathbf{q}$ defined by some predicate $\pi$, we can answer the query in expectation (\ie $\E[\inner{\mathbf{q}}{\mathbf{I}}$) by setting all 1D variables $\alpha_j$ that correspond to values that do {\em not} satisfy $\pi$ to zero. This is more efficient than taking multiple derivatives, but simply looping over all variables can still be too slow as the compressed polynomial can, at worst, have exponentially many variables (see~\autoref{th:compress}).

To improve performance, we implement four main optimizations: (1) storing the compressed polynomial in memory, (2) parallelizing the computation, (3) fast containment check using bit vectors, (4) caching of subexpression evaluation. The first and second, storing in memory and parallelization, are straightforward, standard techniques that improve looping computations. Note, we can parallelize the computation because each polynomial term can be evaluated independently.

The next optimization, using bit vectors, is to optimize both \textcolor{Maroon}{findNoConflictGrps} and determining if a variable needs to be set to zero or not during query evaluation. It is important to understand that the polynomial is hierarchical with nested levels of sums of products of sums. For each subterm (\ie sum or product term in our polynomial), we store (a) a map with variable keys and values of the nested subterms containing that variable, (b) a bit vector of which attributes are contained in the subterm, and (c) a bit vector of which multi-dimensional attribute sets are contained in the subterm.

Take~\autoref{ex:query_answer} which references~\autoref{fig:eq:ex:p}. Take the subterm referenced by $\textbf{i}$. This subterm has a map of all variables pointing to one of three nested sum subterms. The attribute bit vector has a 1 in all places, representing it contains all possible attributes, and the multi-dimensional bit vector is all 0s. Now take the subterm $(\sum_{551}^{650} \beta_j)(\sum_{801}^{900} \gamma_k)(\delta_2-1)$. It has a map of only $\beta$ variables from $[551, 650]$, $\gamma$ variables from $[801, 900]$, and $\delta_2$ all pointing to one of three nested subterms. The attribute bit vector would only have 1 in the $B$ and $C$ dimensions, and the multi-dimensional bit vector would have a 1 in the dimension representing the attribute set $BC$.

These objects allow us to quickly check if some 1-dimensional or multi-dimensional statistic is contained in the term or if there are any variables that need to be set to zero (by using the attribute bit vectors) and which subterms those variables are in. To further see the benefit of this optimization, recall the runtime analysis from~\autoref{subsec:sys_opt:build_poly} where \textcolor{Maroon}{findNoConflictGrps} required $(mN)^2$ steps to check if two statistics were conflict free as it iterates over all 1D statistics. Using maps with variable keys allows us to quickly check if a 1D statistic is contained in another, bringing the runtime down to $mN$ for a single pair. The attribute bit vectors can also allow us to skip iterating over subsets of 1D statistics by quickly checking which attributes two statistics share. If some attribute is not shared between two statistics, then that attribute can cause no conflicts and does not need to be iterated over.

This leads to the last technique of caching. Caching is used to avoid recomputing subterms and takes advantage of the attribute bit vectors and variable hash maps described above. Since we solve for all the variables $\alpha_j$ of our model once and they remained fixed throughout query answering, if there is a subterm of our model that does not contain any variable the needs to be set to zero, we can reuse that subterm's value. We store this value along with the map and bit vectors.

By utilizing these techniques, we reduced the time to learn the model (solver runtime) from 3 months to 1 day and saw a decrease in query answering runtime from around 10 sec to 500 ms (95\% decrease). More runtime results are in~\autoref{sec:evaluation}.

%% file: statistic_selection.tex

In this section, we discuss how we choose the multi-dimensional statistics. We investigate different heuristic techniques for both finding optimal statistic ranges and reordering the data prior to statistic collection optimally.  Recall that our summary always includes all 1D statistics of the form $A_i = v$ for all attributes $A_i$ and all values $v$ in the active domain $D_i$. We describe here how to tradeoff the size of the summary for the precision of the MaxEnt model.

\subsection{Optimal Ranges}
\label{subsec:stat_select:optimal_ranges}
The first choice we make is to include only 2D statistics. It has been shown that restricting to pairwise correlations offers a reasonable compromise between the number of statistics needed and the summary's accuracy~\cite{tzoumas2013efficiently}. This means each multi-dimensional statistic predicate $\pi_j$ is equivalent to a range predicate over two attributes $A_{i_1} \in [u_1,v_1] \wedge A_{i_2} \in [u_2,v_2]$. If $A_{i_1}$ and $A_{i_2}$ are two dimensions of a rectangle, $\pi_j$ defines a sub-rectangle in this space. As the 2D predicates are disjoint, if $\pi_{j_1}$ and $\pi_{j_2}$ both define rectangles over $A_{i_1}$ and $A_{i_2}$, then these rectangles do not overlap. 

As mentioned in~\autoref{subsec:log_opt:compress}, we have two parameters to consider: $B_a$, the number of distinct attribute pairs we gather statistics on, and $B_s$, the number of statistics to gather per each attribute pair. We choose to make $B_s$ be the same for all multi-dimensional statistics. The problem is as follows: given $B_a$ and $B_s$, which $B_a$ attribute pairs $A_{i_1}A_{i_2}$ do we collect statistics on and which $B_s$ statistics do we collect for each attribute pair? This is a complex problem, and we make the simplifying assumption that $B_s$, the number of statistics, is given, but we explore different choices of $B_a$ in Sec.~\ref{sec:evaluation}. We leave it to future work to investigate automatic techniques for determining the total budget, $B_a*B_s$.

Given $B_a$, we consider two different approaches when picking pairs: attribute correlation and attribute cover. The first focuses only on correlation by picking the set of attribute pairs that have the highest combined correlation\footnote{This can be found by calculating, for all attribute pairs, the chi-squared value on the contingency table of $A_{i_1}$ and $A_{i_2}$ and sorting from highest to lowest chi-squared value.} such that every pair has at least one attribute not included in any previously chosen, more correlated pair. This is similar to computing a Chow-Liu tree which is a maximum weight spanning tree over a graph where attributes are nodes and edge weights are the mutual information between pairs of attributes~\cite{chow1968approximating}. The difference is that we use the chi-squared metric rather than the mutual information as chi-squared is a common independence test for categorical data. We leave it to future work to evaluate different correlation metrics.

The second approach focuses on attribute cover by picking the set of pairs that cover the most attributes with the highest combined correlation. For example, if $B_a = 2$ and we have the attribute pairs $BC$, $AB$, $CD$, and $AD$ in order of most to least correlated, if we only consider correlation, we would choose $AB$ and $BC$. However, if we consider attribute cover, we would choose $AB$ and $CD$. We experiment with both of these choices in Sec.~\ref{sec:evaluation}.

Next, we assume for each attribute $A_i$, its domain $D_i$ is ordered and viewed as an array such that $D_i[1] \leq D_i[2] \leq ...$. This allows us to define a $D_{i_1} \times D_{i_2}$ space (a $N_{i_1} \times N_{i_2}$ matrix denoted $\mathcal{M}$) representing the frequency of attribute pairs. In particular, for some $x \in [1, N_{i_1}]$ and $y \in [1, N_{i_2}]$, $\mathcal{M}[x,y] = |\sigma_{A_{i_1}=D_{i_1}[x] \wedge A_{i_2}=D_{i_2}[y]}(I)|$. Our goal is to choose the best $B_s$ 2D range predicates $[l^x,u^x] \times [l^y,u^y]$ where $l^x$ and $u^x$ are lower and upper index bounds on the $x$ axis (likewise for the $y$ axis). We consider three heuristics and show experimental results to determine which technique yields, on average, the lowest error on query results.

{\bf LARGE SINGLE CELL} In this heuristic, the range predicates are single point predicates, $A_{i_1}=D_{i_1}[x] \wedge A_{i_2}=D_{i_2}[y]$, and we choose the points $(x,y)$ as the $B_s$ most popular values in the two dimensional space; \ie, the $B_s$ largest values of $|\sigma_{A_{i_1}=D_{i_1}[x] \wedge A_{i_2}=D_{i_2}[y]}(I)|$.

{\bf ZERO SINGLE CELL} In this heuristic, we select the empty/zero/nonexistent cells; \ie, we choose $B_s$ points $(x,y)$ s.t. $\sigma_{A_{i_1}=D_{i_1}[x] \wedge A_{i_2}=D_{i_2}[y]}(I)=\emptyset$. If there are fewer than $B_s$ such points, we choose the remaining points as in LARGE SINGLE CELL. The justification for this heuristic is that, given only the 1D statistics, the MaxEnt model will produce false positives (``phantom'' tuples) in empty cells; this is the opposite problem encountered by sampling techniques, which return false negatives. This heuristic has another advantage because the value of $\alpha_j$ in $P$ is always 0 and does not need to be updated during solving.

{\bf COMPOSITE} This method partitions $\mathcal{M}$ into a set of $B_s$ disjoint rectangles and associates one statistic with each rectangle. For example if $\pi_{j_1}$ is $A_{i_1} \in [u_1,v_1] \wedge A_{i_2} \in [u_2,v_2]$ and $\pi_{j_2}$ is $A_{i_1} \in [u_3,v_3] \wedge A_{i_2} \in [u_4,v_4]$, then the composite statistic of $\pi_{j_1}$ and $\pi_{j_2}$ is $A_{i_1} \in ([u_1,v_1] \lor [u_3,v_3]) \wedge A_{i_2} \in ([u_2,v_2] \lor [u_4,v_4])$. We choose to combine the statistics by an attribute-wise union because our factorization algorithm requires it. Part $(\textbf{iii})$ of~\autoref{th:compress} multiplies the multi-dimensional statistic correction term (\ie $(\delta - 1)$) by a sum of the 1D statistics associated with it. In our example, we would multiply the composite statistic correction term by $(\alpha_{u_1} + \ldots + \alpha_{v_1} + \alpha_{u_3} + \ldots + \alpha_{v_3})(\alpha_{u_2} + \ldots + \alpha_{v_2} + \alpha_{u_4} + \ldots + \alpha_{v_4})$, which can be represented by a rectangle or bounding box. As we must maintain that the composite statistics can be represented by disjoint rectangles, we use an adaptation of K-D trees to partition the data. 

Recall that a K-D tree partitions a k-dimensional space by iterating over each axis $i$ and splitting the space at the median of the $i$th axis. Each child is then partitioned on the $i+1$ axis. The only difference between our K-D tree algorithm and the traditional one is our splitting condition. Instead of splitting on the median, we split on the value that has the lowest sum squared average value difference.

For a child partition with boundary $[l^x,u^x] \times [l^y,u^y]$, the split condition for the $x$ axis is shown in~\autoref{eq:kd_tree_split} where $\bar{s_l}$ is the average value of the left partition candidate; \ie,
$$\bar{s}_l = \frac{\sum_{(x,y) \in [l^x,m^x] \times [l^y,u^y]} \left(\mathcal{M}[x,y]\right)}{(m^x-l^x + 1)(u^y-l^y+1)}.$$

$\bar{s_r}$ is for the right partition candidate which uses $[m^{x}+1, u^x]$ instead of $[l^x,m^x]$.

\begin{equation}
\label{eq:kd_tree_split}
\begin{split}
\argmin_{m^x} \left[\sum_{(x,y) \in [l^x,m^x] \times [l^y,u^y]} \left(\mathcal{M}[x,y] - \bar{s_l}\right)^2 \right.\\
\left. + \sum_{(x,y) \in [m^{x}+1,u^x] \times [l^y,u^y]} \left(\mathcal{M}[x,y] - \bar{s_r}\right)^2\right]^{1/2}.
\end{split}
\end{equation}
An equivalent expression is used for the $y$ axis.

\begin{figure}[!t]
  \centering
  \subfloat[]{\includegraphics[width=0.18\textwidth]{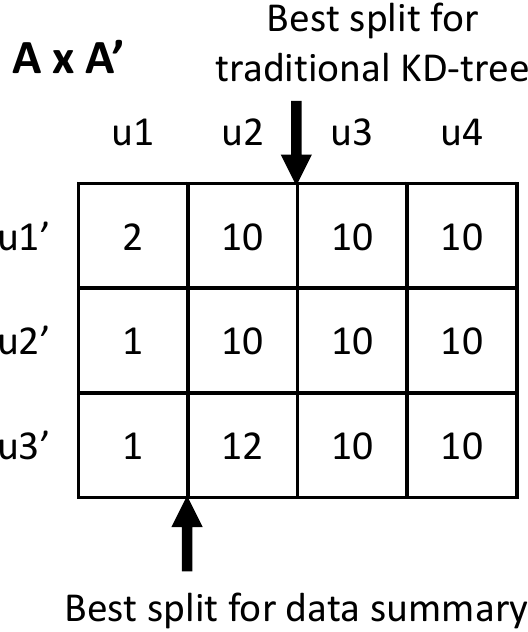}\label{fig:kd_split_sorta}}
  \hfill
  \subfloat[]{\includegraphics[width=0.27\textwidth]{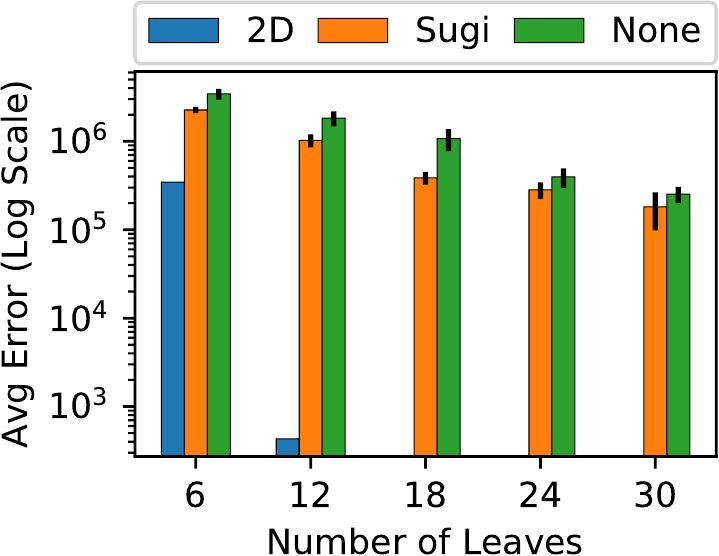}\label{fig:kd_split_sortb}}
  \caption{(a) Example K-D tree showing the traditional split on the median versus our split minimizing average error. (b) Comparison of not sorting, using {\bf SUGI} sort, or using {\bf 2D} sort before running the K-D tree algorithm.}
\end{figure}

We choose this split because we want our K-D tree to best represent the true values. Suppose we have cell counts on dimensions $A$ and $A'$ as shown in~\autoref{fig:kd_split_sorta}. For the next vertical split, if we followed the standard K-D tree algorithm, we would choose the second split. Instead, our method chooses the first split. Using the first split minimizes the sum squared error.

Our {\bf COMPOSITE} method repeatedly splits the attribute domains $D_{i_1}$ and $D_{i_2}$ (alternating) by choosing the split value following~\autoref{eq:kd_tree_split} until it exhausts the budget $B_s$. Then, for each rectangle $[l^x_j,u^x_j] \times [l^y_j,u^y_j]$ in the resulting K-D tree, it creates a 2D statistic $(\mathbf{c}_j, s_j)$, where the query $\mathbf{c}_j$ is associated with the number of tuples satisfying the 2D range predicate and the numerical value
$$ s_j \eqdef |\sigma_{A_{i_1} \in [D_{i_1}[l^x_j],D_{i_1}[u^x_j]] \wedge A_{i_2} \in [D_{i_2}[l^y_j],D_{i_2}[u^y_j]]}(I)|.$$

In~\autoref{subsubsec:results:statistic_selection_technique} we evaluate the three different heuristic selection techniques.

\subsection{Optimal Ordering}
\label{subsec:stat_select:optimal_ordering}
Here we describe how to improve the {\bf COMPOSITE} method by reordering the domains of the attributes, \ie the values in the matrix $\mathcal{M}$, because the split condition (\autoref{eq:kd_tree_split}) depends on the similarity of values within the bounds $[l^x,u^x] \times [l^y,u^y]$. Since our K-D tree relies on the sort order of the underlying matrix, we can permute the rows and columns before building the K-D tree to achieve a lower error.

To show how data ordering can improve the average sum squared error across the leaves, take the K-D tree plots in~\autoref{fig:kd_optimal_sort_example}. The K-D tree splits are shown in black lines on top of frequency heatmaps. The average error is printed below the x-axis. The trees are built on 12 by 12 data with individual cell frequencies ranging from 0 to 4,000,000. The data is constructed such that there is an optimal ordering that achieves 0 average sum squared error.  The left plot is unordered while the right plot more optimally sorts the data (we describe the sorting in~\autoref{subsec:stat_select:optimal_ordering:heuristic_sort}). It can be seen that (b) has grouped together similar values which means leaves have lower error.
\begin{figure}[t!]
  \centering
  \includegraphics[width=0.5\textwidth]{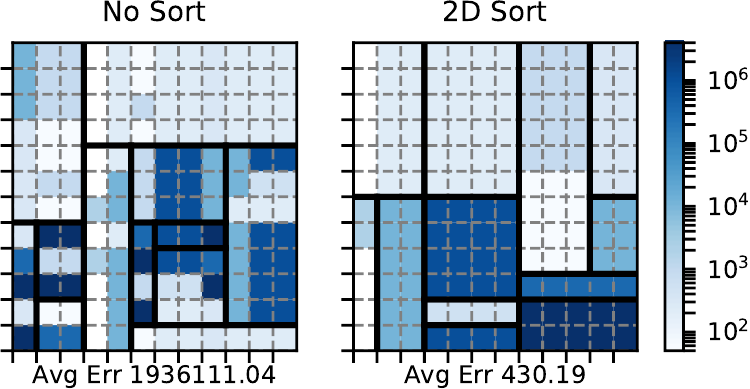}
  \caption{Plots showing the frequency heatmaps and the K-D trees built on data that is unsorted (left) and sorted using the {\bf 2D} sort algorithm (right). The average K-D tree leaf error is shown below.}
  \label{fig:kd_optimal_sort_example}
\end{figure}
To formalize the problem, let the matrix $\mathcal{M} = D_{i_1} \times D_{i_2}$ of size $N_{i_1} \times N_{i_2}$ be the frequency of values in the domains of attributes $A_{i_1}$ and $A_{i_2}$. For some index point $(x,y)$, $\mathcal{M}[x,y] = |\sigma_{A_{i_1} = D_{i_1}[x] \wedge A_{i_2} = D_{i_2}[y]}(I)|$. Denote the set of K-D tree leaves generated from running the K-D tree algorithm as $\mbox{KD}(\mathcal{M}) = \setof{[l^x_j,u^x_j] \times [l^y_j,u^y_j]}{j = 1,B_s}$. The K-D tree error is
\begin{equation}
\label{eq:kd_tree_error}
\mbox{err}(\mbox{KD}(\mathcal{M})) \eqdef \frac{1}{B_s}\left[\sum_{(x,y) \in [l^x_j,u^x_j] \times [l^y_j,u^y_j]} \left(\mathcal{M}[x,y] - \bar{s}_j\right)^2\right]^{1/2}
\end{equation}
where $\bar{s}_j$ is the average value per cell; \ie, $\bar{s}_j = s_j / (u^x_j-l^x_j + 1)(u^y_j-l^y_j+1)$.

Our goal is to solve
\begin{equation}
\label{eq:kd_tree_arg_min}
\argmin_{\pi_x, \pi_y} \mbox{err}(\mbox{KD}(\pi_x\mathcal{M}\pi_y))
\end{equation}
where $\pi_x$ and $\pi_y$ are row and column permutation matrices, respectively.

To solve this, we rely on heuristic techniques.

\subsection{Heuristic Sorts}
\label{subsec:stat_select:optimal_ordering:heuristic_sort}
Inspired by the work in finding optimal matrix reorderings for data visualization and Rectangle Rule List minimization~\cite{makinen2000reordering,behrisch2016matrix,applegate2007compressing}, we experiment with two different heuristic sort algorithms described in~\cite{makinen2000reordering} to more optimally order $\mathcal{M}$ and reduce~\autoref{eq:kd_tree_error}. At a high level, these heuristic techniques aim to permute a matrix to group together similar values. In doing so, this helps to mimimize our K-D tree error because a rectangle around these values will have lower error.

Both of the sort heuristic algorithms alternate between reordering the rows and columns until either a maximum iteration has been reached or there is no change to the sort order. The first sort algorithm, Sugiyama sort ({\bf SUGI}), is traditionally used on binary data and sorts the rows (columns) by the average index of the one-valued columns (rows). We modify the sort to sort by the average index of zero-valued columns (rows) instead to encourage more zero-valued rectangles and lower the likelihood of having a zero-valued cell in a non-zero rectangle. The second sort, {\bf 2D} sort, sorts the rows (columns) by the sum of index times the values in the columns (rows); \ie, a weighted column (row) sum weighted by the index value. Note that the index starts at one, not zero.

For example,~\autoref{fig:kdtree_sort_example} shows a matrix with index values in blue next to the rows and columns. The top diagram shows how {\bf SUGI} sort reorders the rows of the matrix by the average index value of the zeros. As the rows are sorted in ascending order, the middle row moves to the top, the third row to the middle, and the first row to the bottom.

The second diagram shows how {\bf 2D} sort reorders the rows of the matrix by the index weighted sum of the values. In this case, the result is that the second and third rows switch. The sorts would then continue by reordering the columns by the same techniques and so on.

\begin{figure}[!t]
    \centering
    \includegraphics[width=\linewidth]{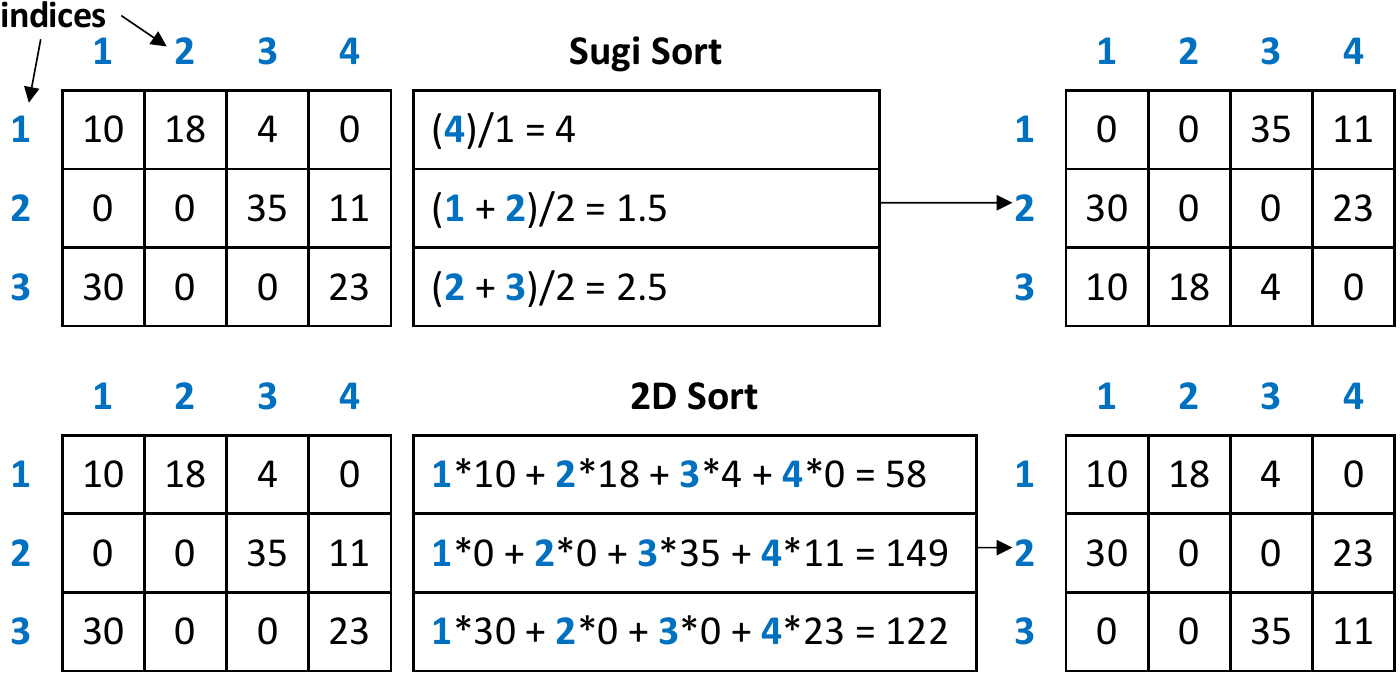}
    \caption{Sorting a matrix's rows by {\bf 2D} sort (top) and {\bf SUGI} sort (bottom).}
    \label{fig:kdtree_sort_example}
\end{figure}

To evaluate the two different sort heuristics, we first generate a 12 x 12 matrix $\mathcal{M}$ that has an optimal permutation order such that $\mbox{err}(\mbox{KD}(\mathcal{M})) = 0$ when using 12 K-D tree leaves. We randomly permute the rows and columns of $\mathcal{M}$ ten times and compare running the K-D tree algorithm directly on the unsorted matrix versus first doing {\bf SUGI} or {\bf 2D} sort with various values for the number of K-D tree leaves.

\autoref{fig:kd_split_sortb} shows the difference in average K-D tree error (\autoref{eq:kd_tree_error}) and standard deviation for the three methods across ten trials with the number of leaves varying from 6 to 30. You can see that {\bf 2D} sort greatly outperforms {\bf SUGI} sort and has no standard deviation because it always reaches the same sort order. It also very quickly converges to having zero error, but it does not learn the optimal order because it does not get zero error with 12 leaves.  {\bf 2D} sort's success is due to the fact that it takes values into account as well as the index position, therefore grouping together cells with similar frequencies. {\bf SUGI} sort, on the other hand, merely tries to group together the zeros. We do see, however, that {\bf SUGI} sort is better than no sort.

We show in~\autoref{subsubsec:results:statistic_selection_accuracy} how using {\bf 2D} sort impacts the overall query error of our MaxEnt technique.

%% file: evaluation.tex
In this section, we evaluate the performance of \name in terms of query accuracy and query execution time. We compare our approach to uniform sampling and stratified sampling.

\subsection{Implementation}
\label{sec:implementation}
We implemented our polynomial solver and query evaluator in Java 1.8, in a prototype system that we call \name. We created our own polynomial class and variable types to implement our factorization. We parallelized our polynomial evaluator (see~\autoref{subsec:sys_opt:poly_eval}) using Java's parallel streaming library. We also used Java to store the polynomial factorization in memory.

Lastly, we stored the polynomial variables in a Postgres 9.5.5 database and stored the polynomial factorization in a text file. We perform all experiments on a 64bit Linux machine running Ubuntu 5.4.0. The machine has 120 CPUs and 1 TB of memory\footnote{The maximum amount of memory used in experiments was approximately 40 GB, meaning a system this large is not required.}. For the timing results, the Postgres database, which stores all the samples, also resides on this machine and has a shared buffer size of 250 GB.

\subsection{Experimental Setup}
For all our summaries, we ran our solver for 30 iterations or until the error was below $1 \times 10^{-6}$ using the method presented in Sec.~\ref{subsec:solving}. Our summaries took under 1 day to compute with the majority of the time spent building the polynomial and solving for the parameters.

We evaluate \name on two real datasets as opposed to benchmark data to measure query accuracy in the presence of naturally occurring attribute correlations. The first dataset comprises information on flights in the United States from January 1990 to July 2015~\cite{rita}. We load the data into PostgreSQL, remove null values, and bin all real-valued attributes into equi-width buckets. We further reduce the size of the active domain to decrease memory usage and solver execution time by binning cities such that the two most popular cities in each state are separated and the remaining less popular cities are grouped into a city called `Other'. We use equi-width buckets to facilitate transforming a user's query into our domain and to avoid hiding outliers, but it is future work to try different bucketization strategies. The resulting relation, \texttt{FlightsFine(fl\_date, origin\_city, dest\_city, fl\_time, distance)}, is 5 GB in size.

To vary the size of our active domain, we also create \texttt{FlightsCoarse(fl\_date, origin\_state, dest\_state, fl\_time, distance)}, where we use the origin state and destination state as flight locations. The left table in ~\autoref{fig:attr_size} shows the resulting active domain sizes. 

The second dataset is 210 GB in size. It comprises N-body particle simulation data~\cite{ChaNGaScaling}, which captures the state of astronomy simulation particles at different moments in time (snapshots). The relation \texttt{Particles(density, mass, x, y, z, grp, type, snapshot)} contains attributes that capture particle properties and a binary attribute, grp, indicating if a particle is in a cluster or not. We bucketize the continuous attributes (density, mass, and position coordinates) into equi-width bins. The right table in ~\autoref{fig:attr_size} shows the resulting domain sizes.

\begin{figure}
    \scriptsize
    \centering
    \begin{tabular}{|M{39pt}|M{33pt}|M{33pt}|}
    \hline
    & \texttt{Flights} \texttt{Coarse} & \texttt{Flights} \texttt{Fine} \\ \hline
    \texttt{fl\_date (FD)} & 307 & 307 \\ \hline
    \texttt{origin (OS/OC)} & 54 & 147 \\ \hline
    \texttt{dest (DS/DC)} & 54 & 147 \\ \hline
    \texttt{fl\_time (ET)} & 62 & 62 \\ \hline
    \texttt{distance (DT)} & 81 & 81 \\ \hline
    \# possible tuples & $4.5 \times 10^9$ & $3.3 \times 10^{10}$ \\ \hline
   \end{tabular}
    \begin{tabular}{|c|c|}
    \hline
    & \texttt{Particles} \\ \hline
    \texttt{density} & 58 \\ \hline
    \texttt{mass} & 52 \\ \hline
    \texttt{x} & 21 \\ \hline
    \texttt{y} & 21 \\ \hline
    \texttt{z} & 21 \\ \hline
    \texttt{grp} & 2 \\ \hline
    \texttt{type} & 3 \\ \hline
    \texttt{snapshot} & 3 \\ \hline
    \# possible & \\
    tuples & $5.0 \times 10^8$ \\ \hline
    \end{tabular}
    \caption{Active domain sizes. Each cell shows the number of distinct values after binning. Abbreviations shown in brackets are used in figures to refer to attribute names: e.g., OS stands for origin\texttt{\_}state. }
    \label{fig:attr_size}
\end{figure}

\subsection{Query Accuracy}
\label{subsec:results:accuracy}

We first compare \name using our best statistic selection techniques of {\bf COMPOSITE} and \textbf{2D} sort (see~\autoref{subsec:results:statistic_choice}) to uniform and stratified sampling on the flights dataset. We use one percent samples, which require approximately 100 MB of space when stored in PostgreSQL. To approximately match the sample size, our largest summary requires only 600 KB of space in PostgreSQL to store the polynomial variables and approximately 200 MB of space in a text file to store the polynomial factorization. This, however, could be improved and compressed further beyond what we did in our prototype implementation.

We compute correlations on \texttt{FlightsCoarse} across all attribute pairs and identify the following pairs as having the largest correlations (C stands for ``coarse''): 1C = (origin\texttt{\_}state, distance), 2C = (destination\texttt{\_}state, distance), 3CF = (fl\texttt{\_}time, distance)\footnote{Pair 3 is the same for \texttt{FlightsCoarse} and \texttt{FlightsFine}}, and 4C = (origin\texttt{\_}state, destination\texttt{\_}state). We use the corresponding attributes, which are also the most correlated, for the finer-grained relation and refer to those attribute pairs as 1F, 2F, and 4F.

Following the discussion in~\autoref{sec:stat_selection}, we have two parameters to vary: $B_a$ (``breadth'') and $B_s$ (``depth''). In order to keep the total number of statistics constant, we require that $B_a*B_s = 3000$. This threshold allows for the polynomial to be built and solved in under a day. Using this threshold, we build four summaries to show the difference in choosing statistics based solely on correlation (choosing statistics in order of most to least correlated) versus attribute cover (choosing statistics that cover the attributes with the highest combined correlation). The first summary, No2D, contains only 1D statistics. The next two, Ent1\&2 and Ent3\&4, use 1,500 statistics across the attribute pairs (1, 2) and (3, 4), respectively. The final one, Ent1\&2\&3, uses 1,000 statistics for the three attribute pairs (1, 2, 3). We do not include 2D statistics related to the flight date attribute because this attribute is relatively uniformly distributed and does not need a 2D statistic to correct for the MaxEnt's underlying uniformity assumption.~\autoref{fig:method_summary} summarizes the summaries.



For sampling, we choose to compare with a uniform sample and four different stratified samples. We choose the stratified samples to be along the same attribute pairs as the 2D statistics in our summaries; \ie, pair 1 through pair 4.

\begin{figure}[t]
    \scriptsize
    \centering
    \begin{tabular}{|c|c|c|c|c|c|}
    \hline
    & MaxEnt Method & No2D & 1\&2 & 3\&4 & 1\&2\&3 \\ \hline
    Pair 1 & (origin, distance) & & X & & X \\ \hline
    Pair 2 & (dest, distance) & & X & & X \\ \hline
    Pair 3 & (time, distance) & & & X & X \\ \hline
    Pair 4 & (origin, dest) & & & X & \\ \hline
    \end{tabular}
    \caption{MaxEnt 2D statistics including in the summaries. The top row is the label of the MaxEnt method used in the graphs.}
    \label{fig:method_summary}
\end{figure}

To test query accuracy, we use the following query template:

\begin{small}
\begin{lstlisting}
SELECT A1,..., Am, COUNT(*)
FROM R WHERE A1=`v1' AND ... AND Am=`vm'
GROUP BY A1,..., Am
\end{lstlisting}
\end{small}

We test the approaches on 400 unique \texttt{(A1,.., Am)} values. We choose the attributes for the queries in a way that illustrates the strengths and weaknesses of \name. For the selected attributes, 100 of the values used in the experiments have the largest count (heavy hitters), 100 have the smallest count (light hitters), and 200 (to match the 200 existing values) have a zero true count (nonexistent/null values). To evaluate the accuracy of \name, we compute a $|true-est|/(true+est)$ (a measure of relative difference) on the heavy and light hitters. To evaluate how well \name distinguishes between rare and nonexistent values, we compute the F measure,
$$2*\textrm{precision}*\textrm{recall}/(\textrm{precision}+\textrm{recall})$$
with 
$$\textrm{precision} = \frac{|\{est_t > 0\ :\ t \in \textrm{light hitters}\}|}{|\{est_t > 0\ :\ t \in (\textrm{light hitters}\ \cup\ \textrm{null values})\}|}$$
and
$$ \textrm{recall} = \frac{|\{est_t > 0\ :\ t \in \textrm{light hitters}\}|}{100}.$$
We do not compare the execution time of \name to sampling for the flights data because the dataset is small, and the execution time of \name is, on average, below 0.5 seconds and at most 1 sec. Sec.~\ref{subsec:results:scalability} reports execution time for the larger data.

\begin{figure}[t]
    \centering
    \includegraphics[width=0.48\textwidth]{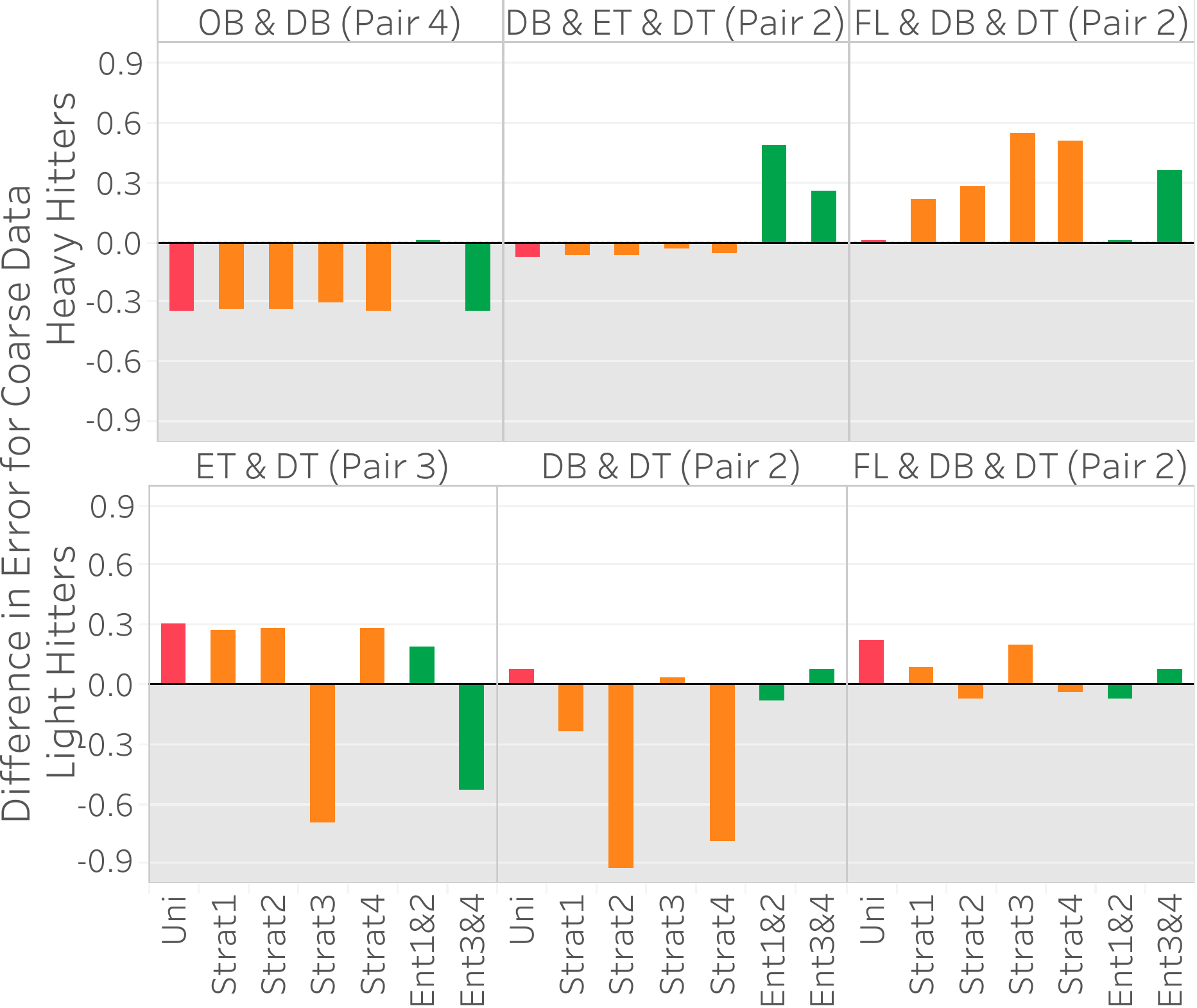}
    \caption{Query error difference between all methods and Ent1\&2\&3 over \texttt{FlightsCoarse}. The pair in parenthesis in the column header corresponds to the 2D statistic pair(s) used in the query template. For reference, pair 1 is (origin/OB, distance/DT), pair 2 is (dest/DB, distance/DT), pair 3 is (time/ET, distance/DT), and pair 4 is (origin/OB, dest/DB).}
    \label{fig:coarsequeries}
\end{figure}

~\autoref{fig:coarsequeries} (top) shows query error differences between all methods and Ent1\&2\&3 (\ie, average error for method X minus average error for Ent1\&2\&3) for three different heavy hitter queries over \texttt{FlightsCoarse}. Hence, bars above zero indicate that Ent1\&2\&3 performs better and vice versa. Each of the three query templates uses a different set of attributes that we manually select to illustrate different scenarios. The attributes of the query are shown in the column header in the figure, and any 2D statistic attribute-pair contained in the query attributes is in parentheses. Each bar shows the average of 100 query instances selecting different values for each template. 

As the figure shows, Ent1\&2\&3 is comparable or better than sampling on two of the three queries and does worse than sampling on query 1. The reason it does worse on query 1 is that it does not have any 2D statistics over 4C, the attribute-pair used in the query, and 4C is fairly correlated. Our lack of a 2D statistic over 4C means we cannot correct for the MaxEnt's uniformity assumption. On the other hand, all samples are able to capture the correlation because the 100 heavy hitters for query 1 are responsible for approximately 25\% of the data. This is further shown by Ent3\&4, which has 4C as one of its 2D statistics, doing better than Ent1\&2\&3 on query 1.

Ent1\&2\&3 is comparable to sampling on query 2 because two of its 2D statistics cover the three attributes in the query. It is better than both Ent1\&2 and Ent3\&4 because each of those methods has only one 2D statistic over the attributes in the query. Finally, Ent1\&2\&3 is better than stratified sampling on query 3 because it not only contains a 2D statistic over 2C but also correctly captures the uniformity of flight date. This uniformity is also why Ent1\&2 and a uniform sample do well on query 3. Another reason stratified sampling performs poorly on query 3 is because the result is highly skewed in the attributes of destination state and distance but remains uniform in flight date. The top 100 heavy hitter tuples all have the destination of `CA' with a distance of 300. This means even a stratified sample over destination state and distance will likely not be able to capture the uniformity of flight date within the strata for `CA' and 300 miles.

~\autoref{fig:coarsequeries} (bottom) shows results for three different light hitter queries over \texttt{FlightsCoarse}. In this case, \name always does better than uniform sampling. Our performance compared to stratified sampling depends on the stratification and query. Stratified sampling outperforms Ent1\&2\&3 when the stratification is exactly along the attributes involved in the query. For example, for query 1, the sample stratified on pair 3 outperforms \name by a significant amount because pair 3CF is computed along the attributes in query 1. Interestingly, Ent3\&4 and Ent1\&2 do better than Ent1\&2\&3 on query 1 and query 2, respectively. Even though both of the query attributes for query 1 and query 2 are statistics in Ent1\&2\&3, Ent1\&2 and Ent3\&4 have more statistics and are thus able to capture more zero elements. Lastly, we see that for query 3, we are comparable to stratified sampling because we have a 2D statistic over pair 2C, and the other attribute, flight date, is relatively uniformly distributed in the query result.

We ran the same queries over the \texttt{FlightsFine} dataset and found \textit{identical} trends in error difference. We therefore omit the graph.

\begin{figure}[t]
    \centering
    \includegraphics[width=0.98\linewidth]{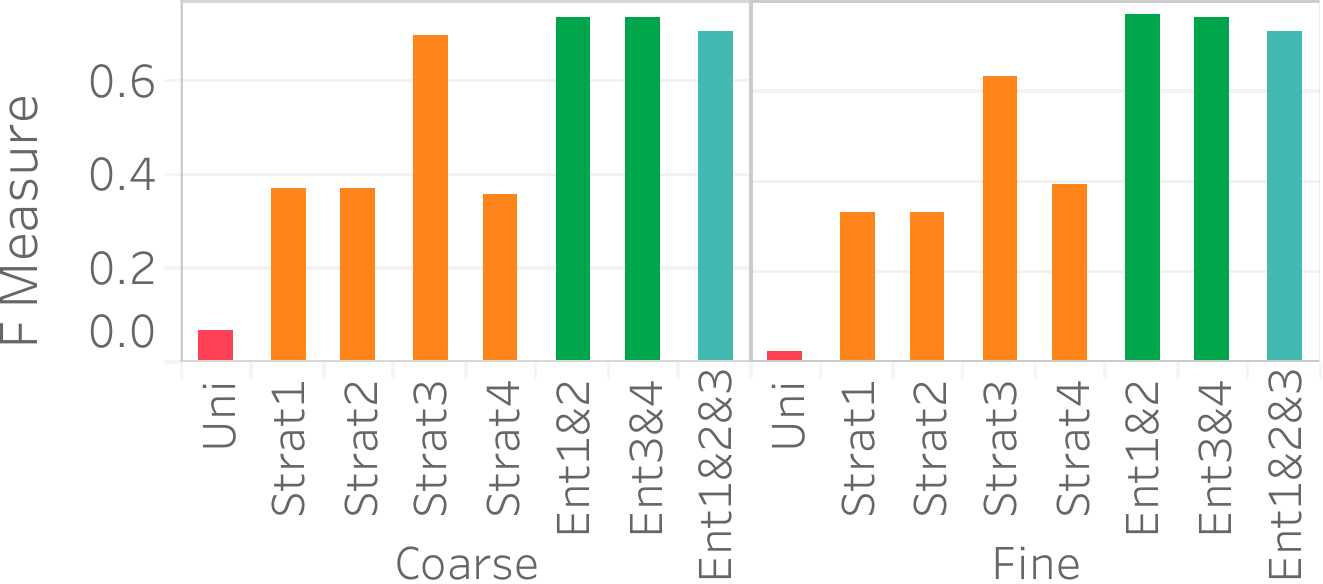}
    \caption{F measure for light hitters and null values over \texttt{FlightsCoarse} (left) and \texttt{FlightsFine} (right).}
    \label{fig:fmeasure}
\end{figure}

An important advantage of our approach is that it more accurately distinguishes between rare values and nonexistent values compared with stratified sampling, which often does not have samples for rare values when the stratification does not match the query attributes. To assess how well our approach works on those rare values, ~\autoref{fig:fmeasure} shows the average F measure over fifteen 2- and 3-dimensional queries selecting light hitters and null values.

We see that Ent1\&2 and 3\&4 have F measures close to 0.72, beating all stratified samples and also beating Ent1\&2\&3. The key reason why they beat Ent1\&2\&3 is that these summaries have the largest numbers of statistics, which ensures they have more fine grained information and can more easily identify regions without tuples. Ent1\&2\&3 has an F measure close to 0.69, which is slightly lower than the stratified sample over pair 3CF but better than all other samples. The reason the sample stratified over pair 3CF performs well is that the flight time attribute has a more skewed distribution and has more rare values than other dimensions. A stratified sample over that dimensions will be able to capture this. On the other hand, Ent1\&2\&3 will estimate a small count for any tuple containing a rare flight time value and will be rounded to 0.

\subsection{Execution Times}
\subsubsection{Scalability}
\label{subsec:results:scalability}

To measure the performance of \name on large-scale datasets, we use three subsets of the 210 GB \texttt{Particles} table. We select data for one, two, or all three snapshots (each snapshot is approximately 70 GB in size). We build a 1 GB uniform sample for each subset of the table as well as a stratified sample over the pair density and group with the same sampling percentage as the uniform sample. We then build two MaxEnt summaries; EntNo2D uses no 2D statistics, and EntAll contains 5 2D statistics with 100 statistics over each of the most correlated attributes, not including snapshot. We do not use any presorting method for this experiment. We run a variety of 4D selection queries such as the ones from Sec.~\ref{subsec:results:accuracy}, split into heavy hitters and light hitters. We record the query accuracy and execution time.

\begin{figure}[t]
    \centering
    \includegraphics[width=0.95\linewidth]{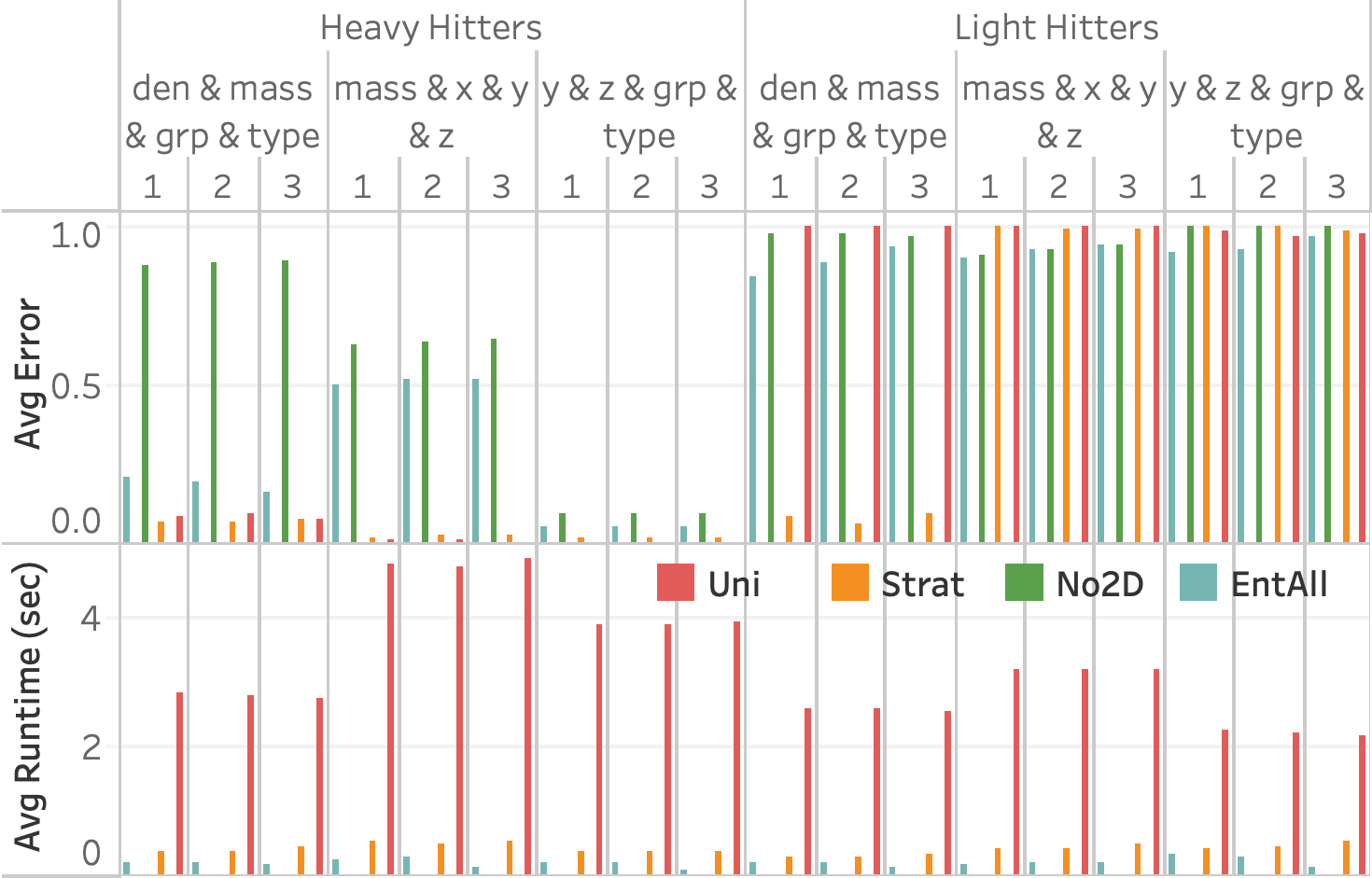}
    \caption{Query average error and execution time for three 4D selection queries on the \texttt{Particles} table. The stratified sample (orange) is stratified on (den, grp).}
    \label{fig:vulcan}
\end{figure}

~\autoref{fig:vulcan} shows the query accuracy and execution time for three different selection queries as the number of snapshots increases. We see that \name consistently does better than sampling on query execution time, although both \name and stratified sampling execute queries in under one second. Stratified sampling outperforms uniform sampling because the stratified samples are generally smaller than their equally selective uniform sample.

In terms of query accuracy, sampling always does better than \name for the heavy hitter queries. This is expected because the bucketization of \texttt{Particles} is relatively coarse grained, and a 1 GB sample is sufficiently large to capture the heavy hitters. We do see that EntAll does significantly better than EntNo2D for query 1 because three of its five statistics are over the attributes of query 1 while only 1 statistic is over the attributes of queries 2 and 3. However, the query results of query 3 are more uniform, which is why EntNo2D and EntAll do well. 

For the light hitter queries, none of the methods do well except for the stratified sample in query 1 because the query is over the attributes used in the stratification. EntAll does slightly better than stratified sampling on queries 2 and 3.

\subsubsection{Solving Time}
\label{subsec:results:solving_time}
To show the data loading and model solving time of \name, we use \texttt{FlightsFine} and measure the time it takes for \name to read in a dataset from Postgres to collect statistics, to build the polynomial, and to solve for the model parameters for various $B_a$ and $B_s$ (see~\autoref{fig:solve_runtimes}). When $B_a = 2$, we gather statistics over pair 1 and pair 2 (MaxEnt1\&2), and when $B_a = 3$, we gather statistics over pair 1, 2, and 3 (MaxEnt3).)
 
\begin{figure}[t]
  \centering
  \includegraphics[width=0.47\textwidth]{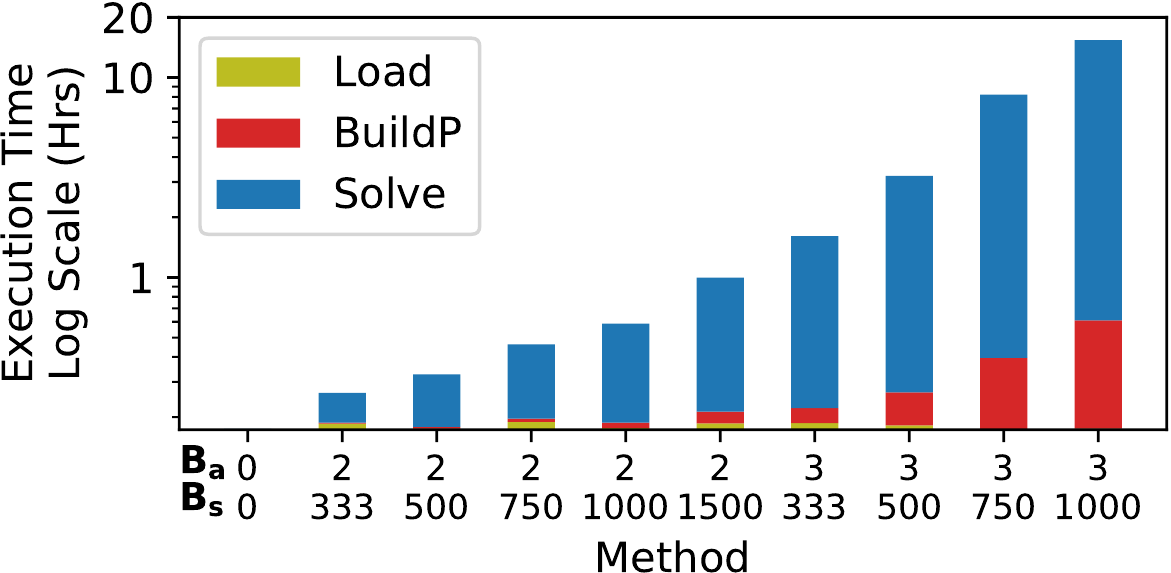}
  \caption{\name log scale execution times for loading the data, building the polynomial, and solving for the parameters for various configurations of $B_a$ and $B_s$ on \texttt{FlightsFine}.}
  \label{fig:solve_runtimes}
\end{figure}

We see that the overall polynomial building and solving execution time grows exponentially as $B_a$ and $B_s$ increase while the data loading time remains constant. The smallest model has a execution time of 10.5 minutes while largest model (MaxEnt3) has a execution time of 15.4 hours. The experiment further demonstrates that $B_a$ impacts execution time more than $B_s$. The method with $B_a = 2, B_s = 750$ has a faster execution time than the method with $B_a = 3, B_s = 500$ even though the total number of statistics, 1,500 in both, is the same. 

Note that the data loading time (yellow) will increase as the dataset gets larger, but once all the histograms and statistics are computed, the time to build the polynomial and solver time are independent of the original data size; they only depend on the model complexity.

\subsubsection{Group By Queries}
\label{subsec:results:groupby_time}

To further expand on execution time results, we measure the execution time to compute eight various 2- and 3-dimensional group-by queries instead of single point queries (sixteen group-by queries in total) to show how the execution time depends on the active domain. As \name can only issue a single point query at a time (the query evaluation is already parallelized), the group-by queries are run as sequences of point queries over the domain of the query attributes.

\autoref{fig:groupby_time} shows a scatter plot of the query domain size versus the execution time for 2- and 3-dimensional group-by queries for the same models as used in~\autoref{subsec:results:solving_time} (\ie $B_a = 0, 2, 3$ and $B_s$ varying form 333 to 1500). If a model takes longer than 10 minutes to compute a group-by query, we terminate its execution. Each color represents a different combination of $B_a$ and $B_s$. Note that running a 3-dimensional group-by query on Postgres on the full \texttt{FlightsFine} can take up to 17 minutes.

\begin{figure}[t]
  \centering
  \includegraphics[width=0.47\textwidth]{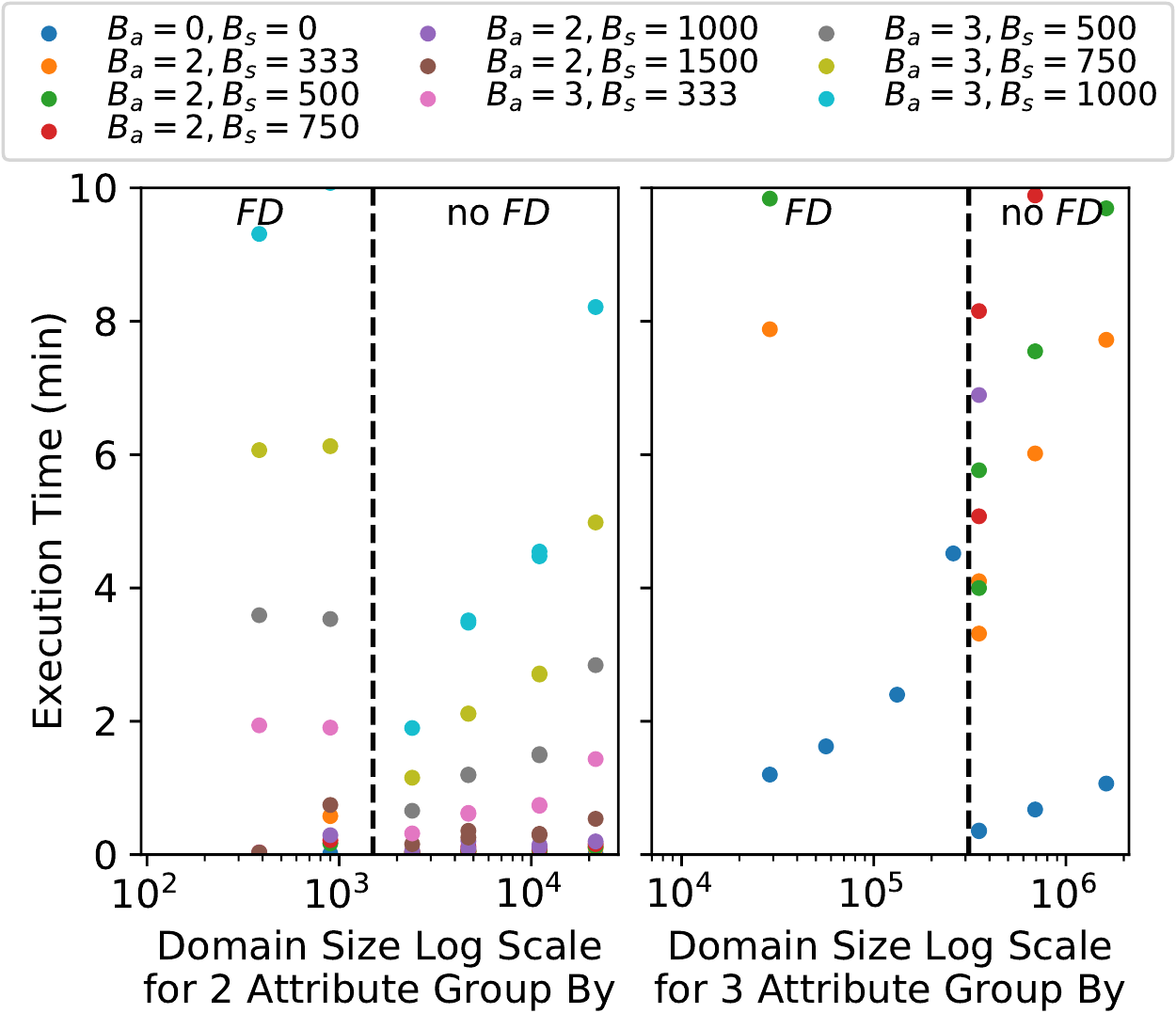}
  \caption{Log scale query execution times for 2- and 3-dimensional group-by queries versus size of query's active domain on \texttt{FlightsFine} for various configurations of $B_a$ and $B_s$. The dashed line demarcates queries with \texttt{fl\_date} as a group-by attribute and those that do not.}
  \label{fig:groupby_time}
\end{figure}

The overall trend we see is that models with a larger $B_a$ are slower to execute, and for models with the same $B_a$, larger $B_s$ is slower. For example, the average execution for 2-dimensional group-by queries for $B_a = 2$ is 8 seconds while it is 87 seconds for $B_a = 3$. This is not surprising and matches the results from~\autoref{fig:solve_runtimes}. We again see that $B_a = 2, B_s = 750$ is faster than $B_a = 3, B_s = 500$ even though the total number of statistics is the same.

Some of the large models had an execution time of longer than 10 minutes for some of the 3-dimensional queries, which is why their scatter point is not shown on all queries. Even though their execution time was more than 10 minutes, each individual point query still ran in under a second.

The results also show a surprising trend in that the execution time dips after the black dashed line and then starts slowing increasing again. This dashed black line demarcates queries containing the \texttt{fl\_date} attribute, the one attribute not included in any 2-dimensional statistic. Note that because the active domain of \texttt{FD} is small, the smaller domain queries happen to contain \texttt{FD}, but the size of the domain is independent of the dip in the execution time.

This unintuitive result is explained by the optimizations in~\autoref{subsec:sys_opt:poly_eval}. By using bit vectors and maps to indicate which attributes and variables are contained in a polynomial subterm, we can quickly decide if that subterm needs to be set to zero or not for query evaluation. This mainly improves evaluation of correction subterms (\ie $(\delta - 1)$ times 1D sums) because the 1D sums contain subsets of the active domain and are more likely overlap with the variables that can be set to zero. The more quickly we can decide if a subterm is zero, the faster the evaluation.

For example, take the polynomial in~\autoref{fig:eq:ex:p}. If we are evaluating a query for $A = 155 \land B = 700 \land C = 700$, then all other $\alpha$, $\beta$, and $\gamma$ variables need to be set to zero except $\alpha_{155}$, $\beta_{700}$, and $\gamma_{700}$. This means the polynomial sums on lines 3, 5, 6, and 7 can all be set to zero without having to evaluate each individual subterm on those lines because $\alpha_{155}$, $\beta_{700}$, and $\gamma_{700}$ are not contained in any of those subterms and $A$, $B$, and $C$ attributes are meant to be zero.

The \texttt{FD} attribute being one of the group-by attributes indicates that all variables representing \texttt{FD} except for the one being selected can be set to zero. However, as \texttt{FD} is not part of a statistic, there are no correction terms being multiplied by subsets of the \texttt{FD} active domain. Therefore, there are fewer chances to set a subterm to zero, meaning the overall query execution is slower.

This evaluation presents and interesting tradeoff between model size, statistic attributes, and query execution. One the one hand, a larger model will take longer to run, in general. On the other hand, more statistics allow for more zero setting optimizations in query evaluation. We also see that while \name can handle 2-dimensional group-by queries, especially if $B_a = 2$, it struggles to perform for 3-dimensional ones. However, as the strength of \name is in querying for light hitters, \name will miss fewer groups than sampling techniques which are more impacted by heavy hitters. We leave it as future work further optimize large domain group-by queries.

\subsection{Statistic Selection}
\label{subsec:results:statistic_choice}

\subsubsection{Selection Technique}
\label{subsubsec:results:statistic_selection_technique}
\begin{figure}[!t]
  \centering
  \includegraphics[width=0.3\textwidth]{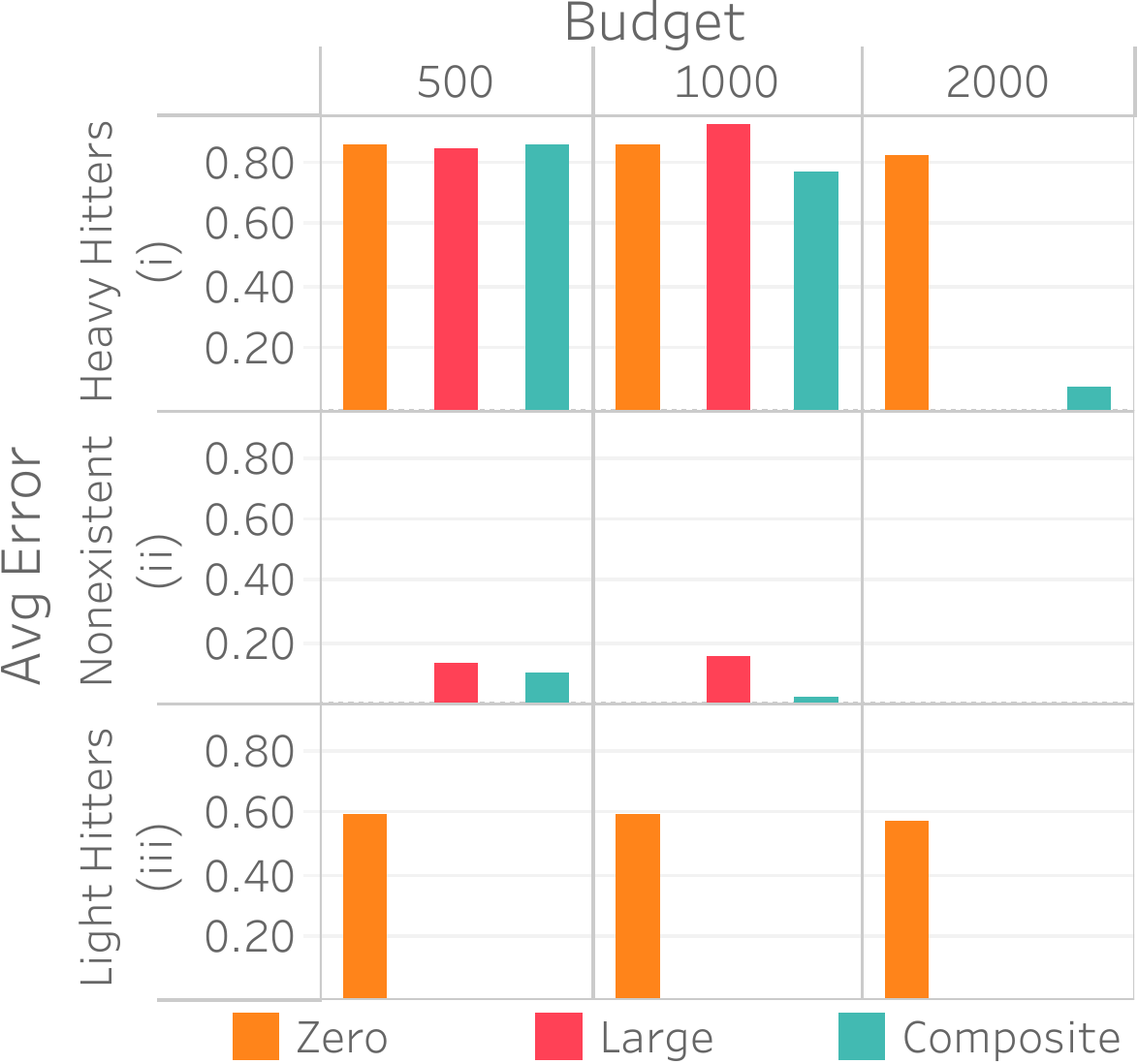}
  \caption{Illustration of query accuracy versus budget for the three different heuristics and three different selections: (i) selecting 100 heavy hitter values, (ii) selecting 200 nonexistent values, and (iii) selecting 100 light hitter values.}
  \label{fig:kd_methodcompare}
\end{figure}

We evaluate the three different statistic selection heuristics, described in~\autoref{subsec:stat_select:optimal_ranges}, on \texttt{FlightsCoarse} restricted to the attributes (date, time, distance). We gather statistics using the three different techniques and using different budgets on the attribute pair (time, distance). There are 5,022 possible 2D statistics, 1,334 of which exist in \texttt{FlightsCoarse}. We evaluate the accuracy of the resultant count of the query

\begin{small}
\begin{lstlisting}
SELECT time, dist, COUNT(*)
FROM Flights WHERE time = x AND dist = y
GROUP BY time, dist
\end{lstlisting}
\end{small}

for 100 heavy hitter (x, y) values, 100 light hitter (x, y) values, and 200 random (x, y) nonexistent/zero values. We choose 200 zero values to match the 100+100 heavy and light hitters.

\autoref{fig:kd_methodcompare} (i) plots the query accuracy versus method and budget for 100 heavy hitter values. Both {\bf LARGE} and {\bf COMPOSITE} achieve almost zero error for the larger budgets while {\bf ZERO} gets around 60 percent error no matter the budget.

(ii) plots the same for nonexistent values, and clearly {\bf ZERO} does best because it captures the zero values first. {\bf COMPOSITE}, however, gets a low error with a budget of 1,000 and outperforms {\bf LARGE}. Interestingly, {\bf LARGE} does slightly worse with a budget of 1,000 than 500. This is a result of the final value of $P$ being larger with a larger budget, and this makes our estimates slightly higher than 0.5, which we round up to 1. With a budget of 500, our estimates are slightly lower than 0.5, which we round down to 0.

Lastly, (iii) plots the same for 100 light hitter values, and while {\bf LARGE} eventually outperforms {\bf COMPOSITE}, {\bf COMPOSITE} gets similar error for all budgets. In fact, {\bf COMPOSITE} outperforms {\bf LARGE} for a budget of 1,000 because {\bf LARGE} predicts that more of the light hitter values are nonexistent than it does with a smaller budget as less weight is distributed to the light hitter values.

Unsurprisingly, we see that {\bf COMPOSITE} is the best method to use across all queries. However, the {\bf COMPOSITE} method is more complex and takes more time to compute. We do learn that if heavy hitter queries are the only relevant queries in a particular workload, it is unnecessary to use the {\bf COMPOSITE} method as {\bf LARGE} does just as well. Also, {\bf ZERO} is the best if existence queries are the most important (\eg if determining set containment). So while {\bf COMPOSITE} is best for a handling a variety of queries, it may not be necessary, depending on the query workload.

\subsubsection{Statistic Accuracy}
\label{subsubsec:results:statistic_selection_accuracy}
We now investigate, in more detail, how the different 2D statistic attribute choices and how presorting the matrix impacts query accuracy. We look at the query accuracy of the four different MaxEnt methods used in~\autoref{fig:coarsequeries} using both \textbf{2D} sort and no sort. We also include the MaxEnt method No2D for comparison although it does not use any sorting. The no sort technique maintains the natural ordering of the domains. We use \texttt{FlightsCoarse} and \texttt{FlightsFine} and the query templates from~\autoref{subsec:results:accuracy}. We run six different two-attribute selection queries over all possible pairs of the attributes covered by pair 1 through 4; \ie, origin, destination, time, and distance. We select 100 heavy hitters, 100 light hitters, and 200 null values. 

\begin{figure}[t]
    \centering
    \subfloat[{\small Heavy Hitters No Sort}]{
    \includegraphics[width=0.2\textwidth]{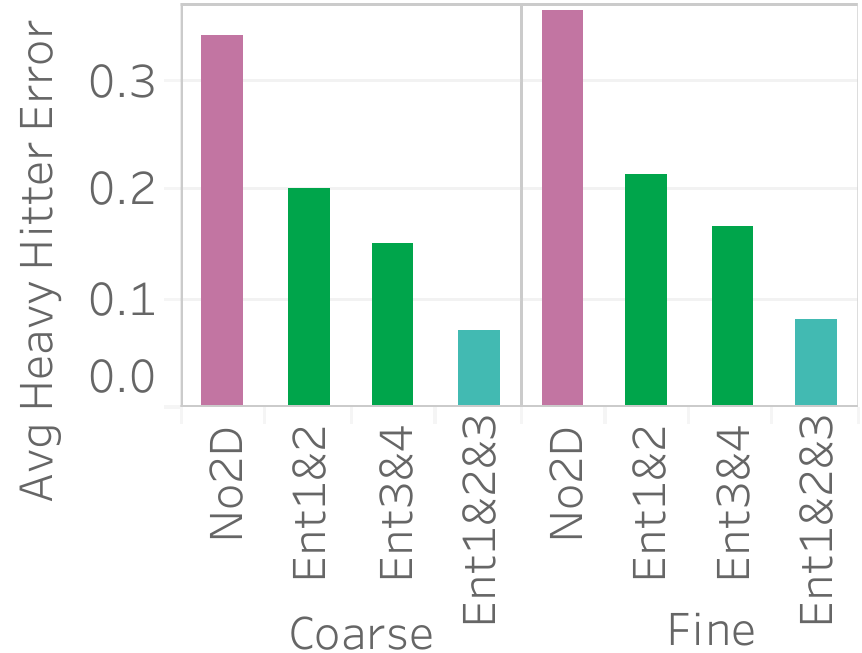}}
    \qquad
    \subfloat[{\small Heavy Hitters {\bf 2D} Sort}]{
    \includegraphics[width=0.2\textwidth]{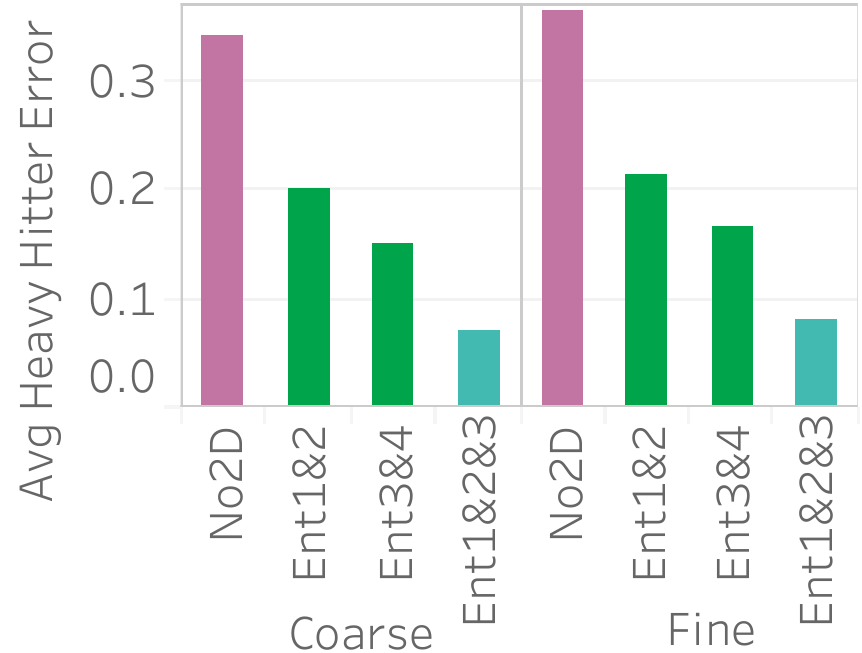}}
    \qquad
    \subfloat[{\small Light Hitters No Sort}]{
    \includegraphics[width=0.2\textwidth]{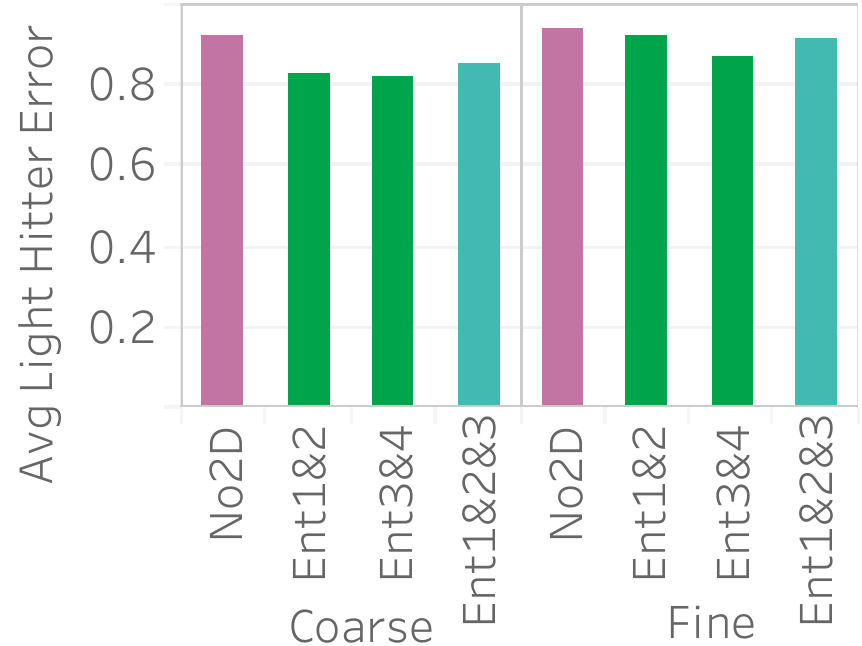}}
    \qquad
    \subfloat[{\small Light Hitters {\bf 2D} Sort}]{
    \includegraphics[width=0.2\textwidth]{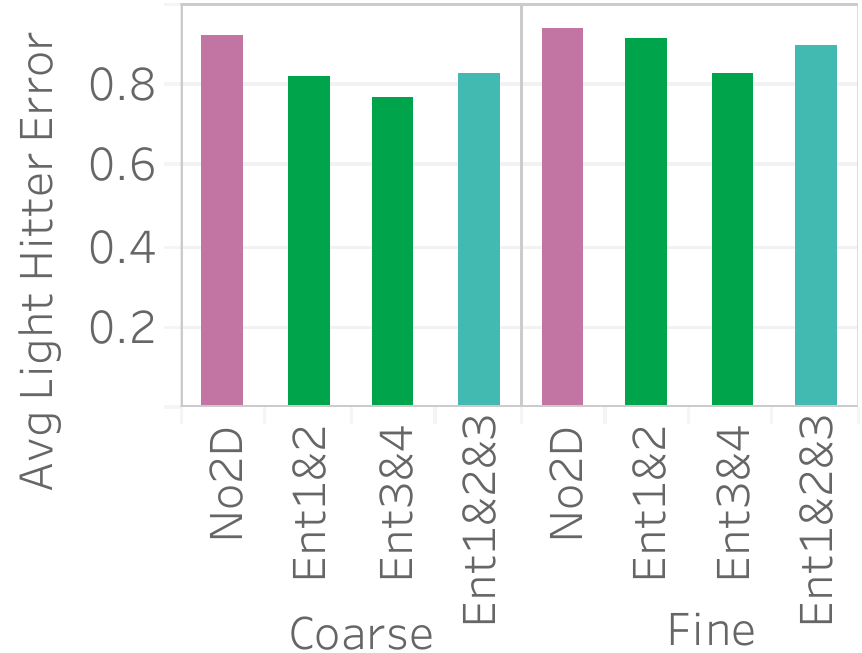}}
    \qquad
    \subfloat[{\small F Measure No Sort}]{
    \includegraphics[width=0.2\textwidth]{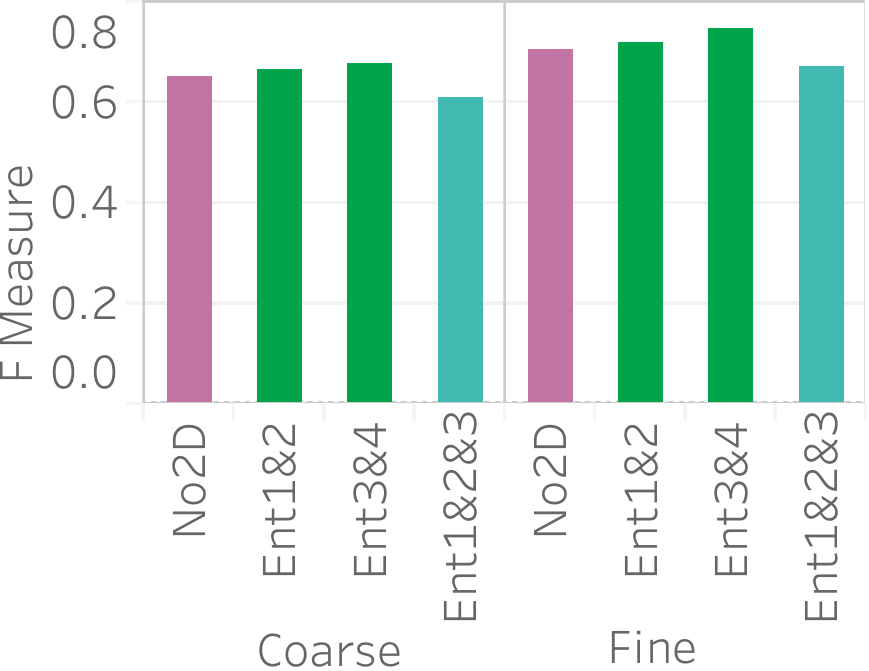}}
    \qquad
    \subfloat[{\small F Measure {\bf 2D} Sort}]{
    \includegraphics[width=0.2\textwidth]{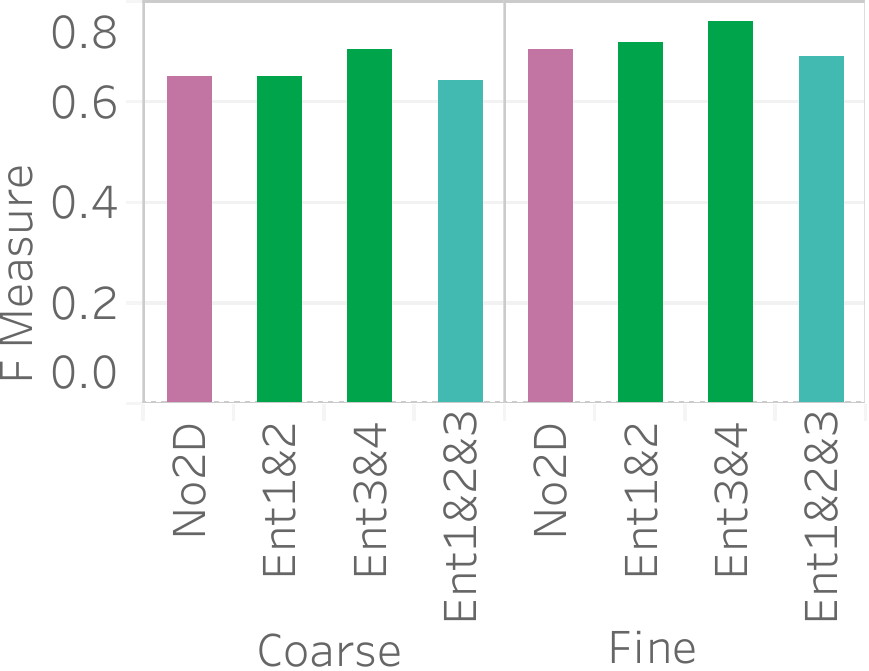}}
    \caption{(a, b) Error over 2D heavy hitter queries, (c, d) error over 2D light hitter queries, and (e, f) F measure over 2D light hitter and null value queries across different MaxEnt methods over \texttt{FlightsCoarse} and \texttt{FlightsFine} using no sort (left side) and {\bf 2D} sort (right side).}
    \label{fig:maxentcompare}
\end{figure}

~\autoref{fig:maxentcompare} shows the average error for the heavy hitters and the light hitters and shows the average F measure across the six queries. The left side shows the error when no presorting is used, and the right side shows the error when {\bf 2D} sort is used. We first consider the different attribute selections. We see that the summary with more attribute pairs but fewer buckets (more ``breadth''), Ent1\&2\&3, does best on the heavy hitters. On the other hand, for the light hitters, we see that the summary with fewer attribute pairs but more buckets (more ``depth'') and still covers the attributes, Ent3\&4, does best. Ent3\&4 doing better than Ent1\&2 implies that choosing the attribute pairs that cover the attributes yields better accuracy than choosing the most correlated pairs because even though Ent1\&2 has the most correlated attribute pairs, it does not have a statistic containing flight time. Lastly, Ent1\&2\&3 does best on the heavy hitter queries yet slightly worse on the light hitter queries because it does not have as many buckets as Ent1\&2 and Ent3\&4 and can thus not capture as many regions in the active domains with no tuples.

When considering presorting the data, we see that {\bf 2D} sort does not have any significant impact on heavy hitter accuracy, with an improvement on the order of 0.001. For light hitters, we see a slight average error improvement using {\bf 2D} sort, and for the F measure, we see a slight decrease in measure. Ent2\&3 has the largest improvement in light hitter query error because of the improvement in resorting pair 3 (distance and time) along with its large K-D tree leaf budget. Ent2\&3 has 1,500 K-D tree leaves which is enough to capture the 1,334 nonzero values of pair 3.

\begin{figure}[th!]
  \centering
  \includegraphics[width=0.5\textwidth]{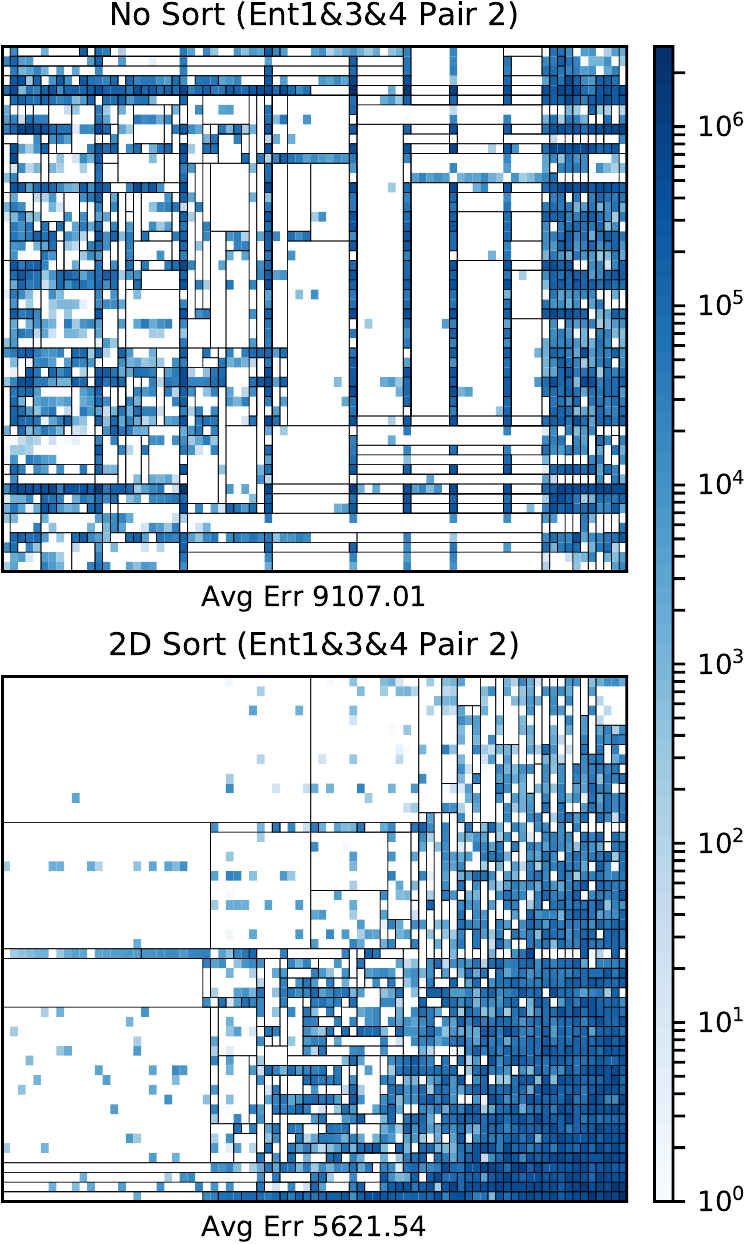}
  \caption{Plots showing the frequency heatmap of the pair 2 attributes of MaxEnt1\&3\&4 of \texttt{FlightsCoarse} that is unsorted (left) and sorted using the {\bf 2D} sort algorithm (right). The average K-D tree error is shown below.}
  \label{fig:kd_optimal_sort_example_real}
\end{figure}

The decrease in F measure and limited improvement for query accuracy is best explained by looking at the sorted and unsorted frequency heatmaps and K-D trees of the pair 2C for Ent1\&2\&3 on \texttt{FlightsCoarse}, shown in~\autoref{fig:kd_optimal_sort_example_real}. We see that the average K-D tree error does, in fact, decrease, which should indicate an improvement in accuracy and F measure. However, upon closer inspection of the K-D tree leaves, we see that the sorted tree actually has put more zeros in leaves with some small, nonzero values. This, in turn, causes MaxEnt1\&2\&3 to believe those zeros actually exist, therefore decreasing the F measure. This result implies that improving K-D tree error is not always enough to guarantee a high F measure because every zero that is in a leaf with some nonzero value will be misclassified as existing.

%% file: discussion.tex
\subsection{Future Work}
\label{subsec:discussion:future}
The above evaluation shows that \name is competitive with stratified sampling overall and better at distinguishing between infrequent and absent values. Importantly, unlike stratified sampling, \name's summaries permit multiple 2D statistics. Further, as \name is based on modeling the data, it does not actually need access to the original, underlying data. If a data scientist only has access to, for example, various 2-dimensional histogram queries of the entire dataset, \name would still be able to build a model of the data and answer queries. Sampling would only be able to handle queries that are directly over one of the histograms. The main limitations of \name are the dependence on the size of the active domain, correlation-based 2D statistic selection, manual bucketization, and limited query support.

To address the first problem, our future work is to investigate using standard algebraic factorization techniques on non-materializable polynomials. By further reducing the polynomial size, we will be able to handle larger domain sizes. We also will explore using statistical model techniques to more effectively decompose the attributes into 2D pairs, similar to~\cite{deshpande2001independence}. To no longer require bucketizing categorical variables (like city), we will research hierarchical polynomials. These polynomials will start with coarse buckets (like states), and build separate polynomials for buckets that require more detail. This may require the user to wait while a new polynomial is being loaded but would allow for different levels of query accuracy without sacrificing polynomial size.

Addressing our queries not reporting error is non-trivial and requires combining the errors in the statistics with the errors in the model parameters with the errors in making the uniformity assumption for the attributes not covered by a statistic. Our future work will be to understand how the error depends on the each of these facets and developing an error equation that propagates these errors through polynomial evaluation.

\subsection{Handling Joins and Data Updates}
\label{subsec:discussion:joinsandupdates}
When building our MaxEnt summary, we assume there was only a single relation being summarized and the underlying data is not updated. We now discuss two extensions of our summarization technique to address both of these assumptions. 
\subsubsection{Joins}
In~\autoref{sec:probabilistic_approach}, we introduce the MaxEnt model over a single, universal (pre-joined) relation $R$ and an instance $\mathbf{I}$ of $R$. Now, suppose the data we want to summarize consists of $r$ relations, $R_1,\ldots,R_r$, each with an associated instance $\mathbf{I}_1,\ldots,\mathbf{I}_r$. To describe our approach, we assume each $R_i$ joins with $R_{i+1}$ by an equi-join on attribute $A_{j_{i,i+1}}$; \ie $R = R_1 \bowtie_{R_1.A_{j_{1,2}} = R_2.A_{j_{1,2}}} R_2 \bowtie \ldots \bowtie R_r$. Our technique can easily be extended to work for multiple equi-join attributes, but for simplicity, we describe the approach for a single join attribute. Let $R(A_1, \ldots, A_m)$ be the global schema of $R_1 \bowtie R_2 \bowtie \ldots \bowtie R_r$ with active domains as described in~\autoref{sec:probabilistic_approach}.

The simplest approach to handle joins is to join all relations and build a summary over the universal relation. Once the summary is built, the universal relation can be removed. While this summary can now handle queries over $R$, it requires $R$ to be computed once, which can be an expensive procedure.

Our approach is to build a separate MaxEnt data summary for each instance: $\{(P_1, \{\alpha_j\}_1, \Phi_1), \ldots, (P_r, \{\alpha_j\}_r, \Phi_r)\}$. A linear query $\mathbf{q}$ with associated predicate $\pi_{\mathbf{q}}$ over $R$ is answered by iteration over the distinct values in the join attributes; \ie
\begin{align*}
\E[\inner{\mathbf{q}}{\mathbf{I}}] &= \sum_{d_1 \in D_{j_{1,2}}}\ldots\sum_{d_{r-1} \in D_{j_{r-1,r}}} \E[\inner{\mathbf{q}'}{\mathbf{I}_1}]\ldots\E[\inner{\mathbf{q}'}{\mathbf{I}_r}]\\
\textrm{s.t. } \pi_{\mathbf{q}'} &= \pi_{\mathbf{q}} \land (R_1.A_{j_{1,2}} = d_1) \land (R_2.A_{j_{1,2}} = d_1) \land \\
 &\ldots \land (R_r.A_{j_{r-1,r}} = d_{r-1}).
\end{align*}
where $\mathbf{q}'$ is the linear query associated with $\pi_{\mathbf{q}'}$ and $R_i.A_j$ denotes attribute $A_j$ in relation $R_i$. We abuse notation slightly in that $\E[\inner{\mathbf{q}'}{\mathbf{I}_i}]$ is the answer to $\mathbf{q}'$ projected on to the attributes of $R_i$ (\ie setting $\rho \equiv true$ for attributes not in $R_i$). Note that if $\mathbf{q}$ is only over a subset of relations, then the summation only needs to be over the distinct join values of the relations in the query.

While this method will return an approximate answer, it does rely on iteration over the active domain of the join attributes, which, as we shown in~\autoref{fig:groupby_time}, can be expensive for larger domain sizes. However, for each relation $R_i$, if we modify the statistic constraints associated with the 1D statistics of  the join attribute $A_{j_{i, i+1}}$, we can improve runtime by decreasing the number of iterations in the summation.

At a high level, before learning multi-dimensional statistics, for each $\ell \in J_{j_{i, i+1}}$ (\ie each 1D statistic index for attribute $A_{j_{i, i+1}}$), we replace $(\mathbf{c}_{\ell}, s_{\ell})$ by $(\mathbf{c}_{\ell}, \bar{s})$ where $\bar{s}$ is the average $s_{\ell}$ value of a group of statistics in $J_{j_{i, i+1}}$. This is similar to building a {\bf COMPOSITE} statistic over $A_{j_{i, i+1}}$ except instead of replacing each individual statistic by the composite, we are modifying the constraint for each statistic. We do do not replace the 1D statistics because we still want to be able to query at the level of an individual tuple. As our querying technique is equivalent to derivation, if we remove the fine-grained 1D statistics, there is nothing to derivate by if a query is issued over a 1D statistic.

Specifically, with $B'_s \leq B_s$ as the budget for the 1D statistic, suppose we learn that $\setof{g_k^{i,i+1} = [l^{i}_k,u^{i}_k]}{k = 1, B'_s}$ is the optimal set of boundaries for join attribute $A_{j_{i, i+1}}$ from relation $R_i$ to $R_{i+1}$. These can be learned with the K-D tree method in~\autoref{sec:stat_selection} by sorting and then repeatedly splitting on the single axis until the budget $B'_s$ is reached. We then apply the same bounds of $\setof{[l^{i}_k,u^{i}_k]}{k = 1, B'_s}$ on any multi-dimensional statistic covering $A_{j_{i, i+1}}$ in $R_{i}$ and $R_{i+1}$ and the 1D statistic covering $A_{j_{i, i+1}}$ in $R_{i+1}$ (see~\autoref{ex:join}). As this boundary is learned before multi-dimensional statistics are built, when we build a 2-dimensional statistic covering $A_{j_{i, i+1}}$, we seed the K-D tree with the $A_{j_{i, i+1}}$ axis splits of $\setof{g_k^{i,i+1}}{k = 1, B'_s}$ and repeatedly split on the other axis until reaching our budget $B_s$.

Using this transfer boundary technique and rewriting the summation, we can answer queries over joins by iterating over a single point in each range boundary rather than all individual values. The following example gives intuition as to how this boundary transfer works.
\begin{example}
\label{ex:join}
Suppose we have two relations $R(A, B)$ and $S(B, C)$ with instances $\mathbf{I}_R$ and $\mathbf{I}_S$ where each attribute has an active domain size of 3. We also have a query $\mathbf{q}$ with predicate $ \pi_{\mathbf{q}} = (A = a_1 \land C = c_1)$ over $R \bowtie S$. Let each relation build a data summary $(P, \{\alpha_j\}, \Phi)$ using a 1D composite statistic over $B$ and a 2D statistic over their two attributes with $B_s = 4$ and $B'_s = 2$. Lastly, for the composite statistic, suppose we learn that the optimal boundaries for $B$ are $[b_1,b_2]$ and $[b_3]$, meaning the statistics over $B$ will be
\begin{align*}
&(B = b_1,\ (s_{b_1} + s_{b_2})/2) \\
&(B = b_2,\ (s_{b_1} + s_{b_2})/2) \\
&(B = b_3,\ s_{b_3})
\end{align*}
where $s_{b_i}$ is the number of tuples where $B = b_i$. Note that the constraint for $b_1$ and $b_2$ is the same which implies $n\beta_1\partial P/P\partial \beta_1 = n\beta_2\partial P/P\partial \beta_2$ by~\autoref{eq:e:s}.

By our na\"{i}ve strategy, the answer to $\mathbf{q}$ is
\begin{align*}
\E[\inner{\mathbf{q}}{\mathbf{I_R \bowtie \mathbf{I}_S}}] &= \sum_{b \in [b_1, b_3]} \E[\inner{\mathbf{q}'}{\mathbf{I}_R}]\E[\inner{\mathbf{q}'}{\mathbf{I}_S}]\\
\textrm{s.t. } \pi_{\mathbf{q}'} &= (A = a_1 \land C = c_1 \land B = b).
\end{align*}
This can be rewritten in terms of polynomial derivation as
\begin{align*}
&= \frac{n_R \alpha_1 \beta_1}{P_R} \frac{\partial P_R}{\partial \alpha_1 \partial \beta_1} \frac{n_S \beta_1 \gamma_1}{P_S} \frac{\partial P_S}{\partial \beta_1 \partial \gamma_1}\\
&+ \frac{n_R \alpha_1 \beta_2}{P_R} \frac{\partial P_R}{\partial \alpha_1 \partial \beta_2} \frac{n_S \beta_2 \gamma_1}{P_S} \frac{\partial P_S}{\partial \beta_2 \partial \gamma_1}\\
&+ \frac{n_R \alpha_1 \beta_3}{P_R} \frac{\partial P_R}{\partial \alpha_1 \partial \beta_3} \frac{n_S \beta_3 \gamma_1}{P_S} \frac{\partial P_S}{\partial \beta_3 \partial \gamma_1}.
\end{align*}

Consider two cases: when the other 2-dimensional statistics have the same boundaries on $B$ and when they do not. For the first case, let the 2D statistics over $R$ and $S$ be as shown by the rectangles all in black with the red line for $S$.
\begin{center}
\begin{tikzpicture}

\node[draw=none] at (1.5, 3.8) {\large$R$};
\node[draw=none] at (5.5, 3.8) {\large$S$};

\draw[step=1cm,lightgray,thin] (0,0) grid (3,3);

\draw[step=1cm,lightgray,thin] (4,0) grid (7,3);

\draw[draw=green,thick] (4,2) -- (7,2);

\draw[draw=black,thick] (0,0) rectangle node{$\delta_3$} (2,1);
\draw[draw=black,thick] (0,1) rectangle node{$\delta_1$} (2,3);
\draw[draw=black,thick] (2,0) rectangle node{$\delta_4$} (3,1);
\draw[draw=black,thick] (2,1) rectangle node{$\delta_2$} (3,3);
\draw[draw=none] (5,0) rectangle node{$\delta_8$} (7,1);
\draw[draw=none] (5,1) rectangle node{$\delta_6$} (7,3);
\draw[draw=none] (4,0) rectangle node{$\delta_7$} (5,1);
\draw[draw=none] (4,1) rectangle node{$\delta_5$} (5,3);

\draw[draw=black,thick] (5,0) rectangle (7,3);
\draw[draw=black,thick] (4,0) rectangle (5,3);

\draw[draw=red,thick] (4,1) -- (7,1);

\foreach \x in {1,2,3} {
   \draw (\x - .5, 3) node[anchor=south] {$b_{\x}$};
   \draw (\x + 3.5, 3) node[anchor=south] {$c_{\x}$};
}
\foreach \y in {3,2,1} {
    \draw (0, 4 - \y - .5) node[anchor=east] {$a_{\y}$};
    \draw (4, 4 - \y - .5) node[anchor=east] {$b_{\y}$};
}
\end{tikzpicture}
\end{center}
Using these statistics, $\E[\inner{\mathbf{q}}{\mathbf{I_R \bowtie \mathbf{I}_S}}]$ now becomes
\begin{align*}
&= \frac{n_R \alpha_1 \beta_1}{P_R}\delta_1 \frac{n_S \beta_1 \gamma_1}{P_S} \delta_5 + \frac{n_R \alpha_1 \beta_2}{P_R} \delta_1 \frac{n_S \beta_2 \gamma_1}{P_S} \delta_5\\
&+ \frac{n_R \alpha_1 \beta_3}{P_R} \delta_2 \frac{n_S \beta_3 \gamma_1}{P_S} \delta_7 \\
&= 2\frac{n_R \alpha_1 \beta_1}{P_R}\delta_1 \frac{n_S \beta_1 \gamma_1}{P_S} \delta_5 + \frac{n_R \alpha_1 \beta_3}{P_R} \delta_2 \frac{n_S \beta_3 \gamma_1}{P_S} \delta_7
\end{align*}
where the second line follows because $n\beta_1\partial P/P\partial \beta_1 = n\beta_2\partial P/P\partial \beta_2$. Notice how instead of summing over all distinct values in $B$, we are summing over $B'_s$ values.

In we had statistics over $S$ that were the rectangles in black with the green line (the boundaries on $B$ do {\em not} match those of the composite 1D statistic), $\E[\inner{\mathbf{q}}{\mathbf{I_R \bowtie \mathbf{I}_S}}]$ would be
\begin{align*}
&= \frac{n_R \alpha_1 \beta_1}{P_R}\delta_1 \frac{n_S \beta_1 \gamma_1}{P_S} \delta_5 + \frac{n_R \alpha_1 \beta_2}{P_R} \delta_1 \frac{n_S \beta_2 \gamma_1}{P_S} \delta_6\\
&+ \frac{n_R \alpha_1 \beta_3}{P_R} \delta_2 \frac{n_S \beta_3 \gamma_1}{P_S} \delta_7
\end{align*}
which does not simplify.
\end{example}

Formally, suppose we only use the transfer boundary technique for $A_{j_{r-1,r}}$, the last join attribute; \ie we make the 1D composite statistic boundaries of $A_{j_{r-1,r}}$ the same in $R_{r-1}$ and $R_r$. Using $|g_k^{i,i+1}| = u^{i}_k - l^{i}_k + 1$ (the size of the range), we can rewrite $\E[\inner{\mathbf{q}}{\mathbf{I}}]$ as
\begin{align*}
&\E[\inner{\mathbf{q}}{\mathbf{I}}] = \sum_{d_1 \in D_{j_{1,2}}}\ldots\sum_{d_{r-1} \in D_{j_{r-1,r}}}\prod_{i=1}^{r}\E[\inner{\mathbf{q}'}{\mathbf{I}_i}] \\
&= \sum_{d_1 \in D_{j_{1,2}}}\ldots\sum_{g_k^{r-1,r}}\sum_{d_{r-1} \in g_k^{r-1,r}}\prod_{i=1}^{r}\E[\inner{\mathbf{q}'}{\mathbf{I}_i}] \\
&= \sum_{d_1 \in D_{j_{1,2}}}\ldots\sum_{g_k^{r-1,r}}\E[\inner{\mathbf{q}'_{r}}{\mathbf{I}_{r}}]\sum_{d_{r-1} \in g_k^{r-1,r}}\prod_{i=1}^{r-1}\E[\inner{\mathbf{q}'}{\mathbf{I}_i}] \\
&= \sum_{d_1 \in D_{j_{1,2}}}\ldots\sum_{g_k^{r-1,r}}\Big[|g_k^{r-1,r}|*\E[\inner{\mathbf{q}'_{r-1}}{\mathbf{I}_{r-1}}]\\
&\hspace{80pt} *\E[\inner{\mathbf{q}'_{r}}{\mathbf{I}_{r}}]*\prod_{i=1}^{r-2}\E[\inner{\mathbf{q}'}{\mathbf{I}_i}]\Big]\end{align*}
\vspace{-13pt}
\begin{align*}
\textrm{s.t. } &\pi_{\mathbf{q}'_{r-1}} = \pi \land (R_1.A_{j_{1,2}} = d_1) \land \ldots \land (R_{r}.A_{j_{r-1,r}} = true) \\
&\hspace{24pt}\land (R_{r-1}.A_{j_{r-1,r}} = D_{j_{r-1, r}}[l^{i}_k]) \\
&\pi_{\mathbf{q}'_{r}} = \pi \land (R_1.A_{j_{1,2}} = d_1) \land \ldots \land (R_{r-1}.A_{j_{r-1,r}} = true) \\
&\hspace{24pt}\land (R_r.A_{j_{r-1,r}} = D_{j_{r-1, r}}[l^{i}_k]).
\end{align*}

The step from line one to line two is replacing the sum over $d_{r-1} \in D_{j_{r-1,r}}$ to a sum over boundaries $\set{g_k^{r-1,r}}$ and a sum over distinct values in the boundary. In line three, we pull out the query over $\mathbf{I}_{r}$ because the answer for $\E[\inner{\mathbf{q}'}{\mathbf{I}_{r}}]$ is the same for each $d_{r-1} \in g_k^{r-1,r}$ as they use the same composite statistic. Therefore, we pull out the query and modify the query's predicate to be over the lower boundary value (any value in the boundary would produce equivalent results).

In line four, we perform the same trick and pull out the query over $\mathbf{I}_{r-1}$ because $\E[\inner{\mathbf{q}'}{\mathbf{I}_{r-1}}]$ will also be the same for each $d_{r-1} \in g_k^{r-1,r}$. Lastly, because $\prod_{i=1}^{r-2}\E[\inner{\mathbf{q}'}{\mathbf{I}_i}]$ is independent of $R_{r-1}$ as it does not contain $\mathbf{I}_{r-1}$, we can also pull it out of the sum. At the end, we get a summation over the value one that repeats $|g_k^{r-1,r}|$ times. This sum rewriting trick can be applied to all attributes with shared boundaries on all statistics covering the join attributes.



By performing this boundary transfer trick, we have replaced the sum for distinct values of $A_{j_{r-1,r}}$ with the sum over lower boundary points of $\set{g_k^{i,i+1}}$. We can repeat this boundary transfer for any of the dense distinct join values to make the final join algorithm efficient. Note that this technique does lose accuracy as we are no longer building building fine-grained 1D statistics over the join attributes and are using potentially suboptimal boundaries for other multi-dimensional statistics.

\subsubsection{Updates}
Another key assumption made in~\autoref{sec:probabilistic_approach} is that the data being summarized is read only and not updated. If we relax that assumption and let the underlying data change, our model needs to be updated, too. We make the assumption that data updates are represented as single tuple additions or deletions. For example, a value change can be represented as a tuple deletion followed by a tuple addition. Alg.~\ref{alg:update_model} describes our update technique.

\begin{small}
\begin{algorithm}[t]
\caption{Update Model}
\label{alg:update_model}
\begin{lstlisting}[style=myJava]
for (*@$\Delta t$@*) do
    *@$\Phi$@* = *@\textcolor{Maroon}{updateStats}@*(*@$\Phi$@*, *@$\Delta t$@*)
    if (not *@$\{\alpha_j\}$@* being updated) do
        if *@\textcolor{Maroon}{timeToRebuild}@* do
            *@$(P, \{\alpha_j\}, \Phi)$@* = *@\textcolor{Maroon}{rebuildModel}@*(*@$R$@*)
        else
            *@$\{\alpha_j\}$@* = *@\textcolor{Maroon}{updateParams}@*(*@$\Phi$@*, *@$\{\alpha_j\}$@*)
\end{lstlisting}
\end{algorithm}
\end{small}

The intuition behind our algorithm is that as updates come in, it is satisfactory to initially only update the polynomial parameters $\{\alpha_j\}$ while keeping the statistic predicates the same. However, as the data continues to be updated, the underlying correlations and relationships of the attributes may change, meaning the statistic predicates are no longer optimal. When this occurs, the entire summary needs to be rebuilt. Ideally, the rebuilding would happen overnight or when the summary is not in high demand.

Out algorithm works as follows. For each tuple update, \textcolor{Maroon}{updateStats} modifies $s_j$ for each predicate $\pi_j$ that $t$ satisfies. $s_j$ increases or decreases by one depending of it $t$ is being added or removed. It is important to realize that \textcolor{Maroon}{updateStats} does {\em not} update the predicates defining the statistics, just the predicate values.

After the statistics are updated, we check if either \textcolor{Maroon}{updateParams} or \textcolor{Maroon}{rebuildModel} is currently running or in progress. If it is, we move on to the next update, effectively batching our changes. If no update is in progress, we update or rebuild our model. \textcolor{Maroon}{updateParams} simply updates the polynomial parameters $\{\alpha_j\}$ by running Alg.~\ref{alg:solve} initialized by the last solved for parameters. By initializing our model at the last know solution, we decrease convergence time because many of the parameters are already solved for and do not need to change. In contrast to simply updating our parameters, \textcolor{Maroon}{rebuildModel} starts from scratch and regenerates the statistics, polynomial, and parameters.

The final method to discuss, \textcolor{Maroon}{timeToRebuild}, decides whether to update or rebuild the model. There are numerous different ways to defining \textcolor{Maroon}{timeToRebuild}, and we give three such possibilities below.
\begin{myitemize}
\item When the number of tuple updates reaches some predefined threshold $B$.
\item When the system does not have many users, meaning there is more compute power to rebuild the summary.
\item When attribute correlations are not accurately represented in $\Phi$. \ie when some attribute pair in $\Phi$ is uniformly distributed or when some attribute pair in $R$ is correlated but not included in $\Phi$.
\end{myitemize}

\subsection{Connection to Probabilistic Databases}
\label{subsec:discussion:probdb}
At a high level, \name learned a probability distribution of the data so that each possible instance has some associated probability of existing. Since this possible world semantics is the same semantics as used by probabilistic databases, how does \name relate to probabilistic databases~\cite{suciu2011probabilistic,dalvi2009probabilistic}?

Recall that probabilistic databases store uncertain data, and, like \name, represent the probability of a tuple as $\Pr(t) = \sum_{I \in PWD | t \in I}\Pr(I)$. The uncertainty in the data arises from the application such as data extraction, data integration, or data cleaning. Probabilistic databases are commonly stored as tuple independent (TI) databases where each tuple has an associated marginal probability and is an independent probabilistic event.

\name, on the other hand, does not store the marginal probabilities. Recall that for \name, if we let $\mathbf{q} = \mathbf{t}_{\ell}$ for some tuple $t_{\ell}$, then, from~\autoref{eq:eq}, we can calculate the expected number of times tuple $t_{\ell}$ appears in an instance. If we divide by $n$, we calculate the expected percent of times tuple $t_{\ell}$ appears in an instance. This, in general, is not the same as the marginal probability. It turns out that when $n = 1$, the two are equivalent.

\begin{lemma}
\label{lemma:marginal_prob}
For some tuple $t_{\ell}$ and associated linear query $\mathbf{t}_{\ell}$ (query with 1 in the place of tuple $t_{\ell}$ and 0 elsewhere), when $n = 1$, $\E[\inner{\mathbf{t}_{\ell}}{\mathbf{I}}]$ is the marginal probability of a tuple; \ie, $\Pr(t_{\ell} \in I)$. 
\end{lemma}

\begin{proof}
We will use the same trick as in~\autoref{subsec:query:answer} by extending the polynomial with a new variable $\beta = 1$ representing the query $\mathbf{t}_{\ell}$. Denote this extended polynomial as $P_{\mathbf{t}_{\ell}}$. Following~\autoref{eq:pq}, we get
\begin{align}
(P_{\mathbf{t}_{\ell}})^n &= \left(\sum_{i=1,d} \prod_{j=1,k} \alpha_j^{\inner{\mathbf{c}_j}{\mathbf{t}_i}} \beta^{\inner{\mathbf{t}_{\ell}}{\mathbf{t}_i}}\right)^n \nonumber \\
&= \left(\prod_{j=1,k} \alpha_j^{\inner{\mathbf{c}_j}{\mathbf{t}_{\ell}}}\beta + \sum_{\substack{i=1,d \\ i \neq \ell}} \prod_{j=1,k} \alpha_j^{\inner{\mathbf{c}_j}{\mathbf{t}_i}}\beta^{\inner{\mathbf{t}_{\ell}}{\mathbf{t}_i}}\right)^n \nonumber \\
&= \left(\left(\beta \frac{\partial P_{\mathbf{t}_{\ell}}}{\partial \beta}\right) + \left(P_{\mathbf{t}_{\ell}} - \beta \frac{\partial P_{\mathbf{t}_{\ell}}}{\partial \beta} \right) \right)^n.
\end{align}
Note that $\beta^{\inner{\mathbf{t}_{\ell}}{\mathbf{t}_i}} = 0$ when $\ell \neq i$. We use this rewriting in finding the marginal probability.
\begin{align*}
    \Pr(t_{\ell} \in I) &= 1 - \Pr(t_{\ell} \notin I)\\
    &= 1 - \frac{1}{(P_{\mathbf{t}_{\ell}})^n} \sum_{I : t_{\ell} \notin I}\prod_{j=1,k} \alpha_j^{\inner{\mathbf{c}_j}{\mathbf{I}}}\beta^{\inner{\mathbf{t}_{\ell}}{\mathbf{I}}}\\
    &= 1 - \frac{1}{(P_{\mathbf{t}_{\ell}})^n} \left(\sum_{\substack{i=1,d \\ i \neq \ell}} \prod_{j=1,k} \alpha_j^{\inner{\mathbf{c}_j}{\mathbf{t}_i}}\beta^{\inner{\mathbf{t}_{\ell}}{\mathbf{t}_i}}\right)^n \\
    &= 1 - \frac{1}{(P_{\mathbf{t}_{\ell}})^n} \left(P_{\mathbf{t}_{\ell}} - \beta \frac{\partial P_{\mathbf{t}_{\ell}}}{\partial \beta} \right)^n\\
    &= 1 - \left(1 - \frac{\beta}{P_{\mathbf{t}_{\ell}}}\frac{\partial P_{\mathbf{t}_{\ell}}}{\partial \beta} \right)^n\\
\end{align*}
The third line follows a similar proof as~\autoref{eq:p}. From~\autoref{eq:eq} and since $\beta = 1$ and  $n = 1$, we get
\begin{align*}
\Pr(t_{\ell} \in I) &= \frac{1}{P_{\mathbf{t}_{\ell}}}\frac{\partial P_{\mathbf{t}_{\ell}}}{\partial \beta} \\
&= \E[\inner{\mathbf{t}_{\ell}}{\mathbf{I}}]
\end{align*}
\end{proof}
It is important to note that although we can calculate the marginal probability, the probabilities are not independent; \ie, we do not have a TI probabilistic database.

\subsection{Connection to Graphical Models}
\label{subsec:discussion:graphmodel}
The Principle of Maximum Entropy is well studied, and it is known that the maximum entropy solution with marginal constraints (\ie, {\tt COUNT(*)} constraints) is equivalent to the maximum likelihood solution for exponential family models~\cite{murphy2006lectureongraph,teh2003improving,wainwright2008GME,jordan2003bookedit}. This can be seen by transforming~\autoref{eq:pr:i} into exponential form
\begin{align}
\label{eq:pr:i:exponential}
\Pr(I) = &\exp\left(\left(\sum_{j = 1}^{k}\theta_{j}\phi_{j}(I)\right)\right. \\
& - \left.\log\sum_{I \in PWD}\left(\exp\left(\sum_{j = 1}^{k}\theta_{j}\phi_{j}(I)\right)\right)\right).
\end{align}
Further, this exponential form is equivalent to the probability distribution defined over an exponential family Markov Network or, more generally, factor graph~\cite{yang2015graphical}. Markov networks are factor graphs where the factors representing the parameterization of the probability distribution are defined solely on cliques in the graph.

This connection implies that we can use MLE techniques to solve for the parameters $\theta_j$, and, in fact, the modified gradient descent technique we use in~\autoref{subsec:solving} is the same as the iterative proportional fitting (also called iterative scaling algorithm) used to solve the parameters in exponential family graphical models.

It is important to note that because we use the slotted possible world semantics, we are able to factorize our partition function $Z$ to a multi-linear polynomial raised to the power $n$ (see ~\autoref{eq:z:p}). This simplification allows for drastic performance benefits in terms of solving and query answering~\autoref{sec:logical_optimizations}.

%% file: related_work.tex
Although there has been work in the theoretical aspects of probabilistic databases~\cite{suciu2011probabilistic}, as far as we could find, there is not existing work on using a probabilistic database for data summarization. However, there has been work by Markl~\cite{markl2005consistently} on using the maximum entropy principle to estimate the selectivity of predicates. This is similar to our approach except we are allowing for multiple predicates on an attribute and are using the results to estimate the likelihood of a tuple being in the result of a query rather than the likelihood of a tuple being in the database.

Although not aimed at data summarization, the work in~\cite{bekker2015tractable} builds a probabilistic graphical model that is guaranteed to have efficient query answering for certain classes of queries, \ie a tractable model. They propose a greedy search algorithm that simultaneously learns a Markov network and its underlying sentential decision diagram which gives a tractable representation of the network. While \name also learns a Markov network, our features are count queries over instances while in~\cite{bekker2015tractable}, they only use boolean valued features.

Our work is also similar to that by Suciu and R\'{e}~\cite{re2012understanding} except their goal was to estimate the size of the result of a query rather than tuple likelihood. Their method also relied on statistics on the number of distinct values of an attribute whereas our statistics are based on the selectivity of each value of an attribute. 

Even though approximate query processing has been a major area of research in the database community for decades, there is still no widely accepted solution~\cite{li2018approximate,chaudhuri2017approximate,mozafari2015handbook}. One main AQP technique is to use precomputed samples~\cite{acharya1999aqua,li2018approximate,chaudhuri2017approximate,acharya2000congressional}. In the work by Chaudhiri et al.~\cite{chaudhuri2001robust}, they precompute samples that minimize the errors due to variance in the data for a specific set of queries they predict. The work by~\cite{babcock2003dynamic} chooses multiple samples to use in query execution but only considers single column stratifications. VerdictDB~\cite{park2018verdictdb} introduces variational subsamples as an alternative to bootstrapping and subsampling to provide efficient probabilistic guarantees. The core idea of variational subsampling is to loosen the restrictions on standard subsampling while still providing the same probabilistic guarantees. The work by~\cite{ding2016samplelus} builds a measure-biased sample for each measure dimension to handle sum queries and uniform samples to handle count queries in order to provide a priori accuracy guarantees. Depending on if the query is highly selective or not, they choose an appropriate sample. Along a similar vein is the work of~\cite{li2018bounded} which generates a unified synopses from a set of samples to answer queries approximately within a predefined error bound. The later work of BlinkDB~\cite{agarwal2013blinkdb} removes any assumptions on the queries. BlinkDB only assumes that there is a set of columns that are queried, but the values for these columns can be anything among the possible set of values. BlinkDB then computes samples for each possible value of the predicate column in an online fashion and chooses the single best sample to run when a user executes a query.

A main drawback for many systems relying on precomputed samples is that they require an existing workload to train on. Instead of using precomputed samples and needing a workload, the Quickr system in~\cite{kandula2016quickr} injects sampling operators into query plans to generate samples on the fly for ad-hoc AQP in big data clusters. By using a variety of different samplers, they are able to handle distinct value queries as well as joins.

As \name does not use any samples nor does it require any existing workload, our approach is more closely related to the non-sample based AQP techniques of~\cite{chakrabarti2001approximate}. Much like \name answers queries directly on a probability distribution, the system in~\cite{chakrabarti2001approximate} builds multi-dimensional wavelets from a dataset and answers queries directly over the wavelet coefficient domain. 

There are some AQP hybrid approaches such as the work in~\cite{thirumuruganathan2019approximate} which builds a deep learning model in order to draw better samples for AQP. They first build a collection variational autoencoders to learn the data distribution, and then, during run time, they use the model with rejection sampling to generate samples from the learned distribution. The system AQP++~\cite{peng2018aqp++} combines samples with precomputed aggregates, such as data cubes, to build a unified AQP system. For an input query, it uses samples to generate an estimate of the error between the true answer and the answer run on the closest precomputed aggregate and adds the error to the answer from the precomputed aggregate. IDEA~\cite{galakatos2017revisiting} is a system that combines random sampling with indices targeted at rare populations to answer aggregate queries for interactive data exploration. The system also leverages results from previous aggregate queries and probabilistic query rewrite rules to approximately answer new queries using the old results, if possible.

The last main AQP technique is to use data sketches (\eg count-min or KMV~\cite{bar2002counting})\cite{cormode2011synopses,li2018approximate}. Sketches have the benefit of being able to handle streaming data but are usually built to handle a limited number of queries. For example, a KMV sketch answers how many distinct values there are but doesn't answer other aggregates.

Although not built for AQP specifically, the work in~\cite{wu2007bayesian} uses Bayesian statistics and a few random samples to approximately answer what the extreme values are in the entire dataset. Using a historical query workload and Monte Carlo sampling, they generate a correction factor based on the shape of the query result and its associated error distribution. Then, they use the samples to generate a preliminary estimate and their learned correction factor to update their estimate. While this technique is probabilistic in nature, it is geared towards specifically answering extreme value queries rather than linear queries, like \name.

\cite{ortiz2018learning} and~\cite{kipf2018learned} both use deep learning networks to predict cardinalities for better query optimization. While \name can be used to estimate cardinalities for queries, we do not use deep networks and split our feature vector learning from optimization in that we decide which attributes to use for statistics before running our model learner. The deep network in~\cite{ortiz2018learning} combines learning cardinalities with learning the best representation of a subquery.

%% file: conclusion.tex
We presented, \name, a new approach to generate probabilistic database summaries for interactive data exploration using the Principle of Maximum Entropy. Our approach is complementary to sampling. Unlike sampling, \name's summaries strive to be independent of user queries and capture correlations between multiple different attributes at the same time. Results from our prototype implementation on two real-world datasets up to 210 GB in size demonstrate that this approach is competitive with sampling for queries over frequent items while outperforming sampling on queries over less common items.

\begin{small}
\noindent\textbf{Acknowledgments} This work is supported by NSF 1614738 and NSF 1535565. Laurel Orr is supported by the NSF Graduate Research Fellowship.
\end{small}